\documentclass[a4paper]{IEEEtran}

\usepackage{amsmath}
\usepackage{amsthm}
\usepackage{amssymb}
\usepackage{bm}
\usepackage{cite}
\usepackage{enumitem}
\usepackage{graphicx}
\usepackage{float}
\usepackage{comment}

\usepackage{pgfplots}
\usepackage{pgf}
\usepackage{tikz}
\pgfplotsset{compat = newest}
\usepackage{pgfplotstable}
\usepackage{booktabs}
\usepackage{multirow}

\newtheorem{theorem}{Theorem}
\newtheorem{proposition}{Proposition}
\newtheorem{lemma}{Lemma}
\newtheorem{corollary}{Corollary}

\theoremstyle{remark}
\newtheorem{remark}{Remark}
\newcommand{\ul}[1]{\underline{#1}}
\newcommand{\set}[1]{\mathcal{#1}}
\newcommand{\E}[1]{{\rm E}\left[{#1}\right]}
\newcommand{\Var}[1]{{\rm Var}\left[{#1}\right]}
\renewcommand{\Pr}[1]{{\rm Pr}\left[{#1}\right]}
\newcommand{\tr}[1]{{\rm tr}\left({#1}\right)}

\newcommand{\Gammafun}[1]{\Gamma\left({#1}\right)}
\newcommand{\at}{\makeatletter @\makeatother}
\newcommand{\eps}{\epsilon}
\newcommand{\epsb}{\bar{\epsilon}}

\newcommand{\ulU}{\underline{U}}

\newcommand{\Xv}{\underline{X}}
\newcommand{\Yv}{\underline{Y}}
\newcommand{\Zv}{\underline{Z}}

\newcommand{\Yt}{\tilde{Y}}
\newcommand{\Zt}{\tilde{Z}}

\newcommand{\yt}{\tilde{y}}

\begin{document}
\title{Information Rates for Channels with Fading, \\ Side Information and Adaptive Codewords}
\author{Gerhard Kramer
\thanks{Date  of  current  version  \today. This work was supported by the 6G Future Lab Bavaria funded by the Bavarian State Ministry of Science and the Arts, the project 6G-life funded by the Germany Federal Ministry for Education and Research (BMBF), and by the German Research Foundation (DFG) through projects 390777439 and 509917421.}
\thanks{G.~Kramer is with the Institute for Communications Engineering of the School of Computation, Information, and Technology at the Technical University of Munich, 80333 Munich, Germany (e-mail: gerhard.kramer@tum.de).}}

\maketitle

\begin{abstract}
Generalized mutual information (GMI) is used to compute achievable rates for fading channels with various types of channel state information at the transmitter (CSIT) and receiver (CSIR). The GMI is based on variations of auxiliary channel models with additive white Gaussian noise (AWGN) and circularly-symmetric complex Gaussian inputs. One variation uses reverse channel models with minimum mean square error (MMSE) estimates that give the largest rates but are challenging to optimize. A second variation uses forward channel models with linear MMSE estimates that are easier to optimize. Both model classes are applied to channels where the receiver is unaware of the CSIT and for which adaptive codewords achieve capacity. The forward model inputs are chosen as linear functions of the adaptive codeword's entries to simplify the analysis. For scalar channels, the maximum GMI is then achieved by a conventional codebook, where the amplitude and phase of each channel symbol are modified based on the CSIT. The GMI increases by partitioning the channel output alphabet and using a different auxiliary model for each partition subset. The partitioning also helps to determine the capacity scaling at high and low signal-to-noise ratios. A class of power control policies is described for partial CSIR, including a MMSE policy for full CSIT. Several examples of fading channels with AWGN illustrate the theory, focusing on on-off fading and Rayleigh fading. The capacity results generalize to block fading channels with in-block feedback, including capacity expressions in terms of mutual and directed information.
\end{abstract}


\section{Introduction}
\label{sec:introduction}
The capacity of fading channels is a topic of interest in wireless communications~\cite{OzarowShamaiWyner94,Biglieri-Proakis-Shamai-IT98,Love-etal-JSAC08,Kim-Kramer-22}. Fading refers to model variations over time, frequency, and space. A common approach to track fading is to insert pilot symbols into transmit symbol strings, have receivers estimate fading parameters via the pilot symbols, and have the receivers share their estimated channel state information (CSI) with the transmitters. The CSI available at the receiver (CSIR) and transmitter (CSIT) may be different and imperfect.

Information-theoretic studies on fading channels distinguish between average (ergodic) and outage capacity, causal and non-causal CSI, symbol and rate-limited CSI, and different qualities of CSIR and CSIT that are coarsely categorized as no, perfect, or partial. We refer to~\cite{Keshet-Steinberg-Merhav-08} for a review of the literature up to 2008. We here focus exclusively on average capacity and causal CSIT as introduced in~\cite{Shannon58}. Codes for such CSIT, or more generally for noisy feedback~\cite{Shannon60}, are based on \emph{Shannon strategies}, also called \emph{codetrees}~\cite[Ch.~9.4]{Blahut87}, or \emph{adaptive codewords}~\cite[Sec.~4.1]{Kramer98}.\footnote{The term ``adaptive codeword'' was suggested to the author by J. L. Massey.} Adaptive codewords are usually implemented by a conventional codebook and by modifying the codeword symbols as a function of the CSIT. This approach is optimal for some channels~\cite{Caire-Shamai-IT99} and will be our main interest. 

\subsection{Block Fading}
\label{subsec:block-fading}

A model that accounts for the different time scales of data transmission (e.g., nanoseconds) and channel variations (e.g., milliseconds) is block fading~\cite{McElieceStark84,Stark-McEliece-IT88}. Such fading has the channel parameters constant within blocks of $L$ symbols and varying across blocks.
A basic setup is as follows.
\begin{itemize}
\item The fading is described by a state process $S_{H1},S_{H2},\dots$ independent of the transmitter messages and channel noise. The subscript ``$H$'' emphasizes that the states $S_{Hi}$ may be hidden from the transceivers. 
\item Each receiver sees a state process $S_{R1},S_{R2},\dots$ where $S_{Ri}$ is a noisy function of $S_{Hi}$ for all $i$.
\item Each transmitter sees a state process $S_{T1},S_{T2},\dots$ where $S_{Ti}$ is a noisy function of $S_{Hi}$ for all $i$.
\end{itemize}
The state processes may be modeled as memoryless~\cite{McElieceStark84,Stark-McEliece-IT88} or governed by a Markov chain~\cite{Wang-Moayeri-VT95,Wang-Chang-VT96,Viswanathan99,Zhang-Kassam-COMM99,Tan-Beaulieu-COMM00,Medard-IT00,Riediger-Shwedyk-IT02,Agarwal-Honig-IT12,Ezzine-ISIT22}. The memoryless models are particular cases of Shannon's model~\cite{Shannon58}.
For scalar channels, $S_{Hi}$ is usually a complex number $H_i$.  Similarly, for vector or multi-input, multi-output (MIMO) channels with $M$- and $N$-dimensional inputs and outputs, respectively, $S_{Hi}$ is a $N\times M$ matrix ${\bf H}_i$. 

Consider, for example, a point-to-point channel with block-fading and complex-alphabet inputs $X_{i\ell}$ and outputs
\begin{align}
   Y_{i\ell}=H_i X_{i\ell} + Z_{i\ell}
   \label{eq:block-fading}
\end{align}
where the index $i$, $i=1,\dots,n$, enumerates the blocks and the index $\ell$, $\ell=1,\dots,L$, enumerates the symbols of each block. The additive white Gaussian noise (AWGN) $Z_{11},Z_{12},\dots$ is a sequence of independent and identically distributed (i.i.d.) random variables that have a common circularly-symmetric complex Gaussian (CSCG) distribution.

\subsection{CSI and In-Block Feedback}
\label{subsec:feedback}

The motivation for modeling CSI as independent of the messages is simplicity. If one uses only pilot symbols to estimate the $H_i$ in \eqref{eq:block-fading}, for example, then the independence is valid, and the capacity analysis may be tractable. However, to improve performance, one can implement data and parameter estimation jointly, and one can actively adjust the transmit symbols $X_{i\ell}$ using past received symbols $Y_{ik}$, $k=1,\dots,\ell-1$, if \emph{in-block feedback} is available.\footnote{Across-block feedback does not increase capacity if the state processes are memoryless; see~\cite[Remark~16]{Kramer14}.} An information theory for such feedback was developed in~\cite{Kramer14}, where a challenge is that code design is based on adaptive codewords that are more sophisticated than conventional codewords.

For example, suppose the CSIR is $S_{Ri}=H_i$. Then, one might expect that CSCG signaling is optimal, and the capacity is an average of $\log(1+\text{SNR})$ terms, where SNR is a signal-to-noise ratio. However, this simplification is based on constraints, e.g., that the CSIT is a function of the CSIR and that the $X_{i\ell}$ cannot influence the CSIT. The former constraint can be realistic, e.g., if the receiver quantizes a pilot-based estimate of $H_i$ and sends the quantization bits to the transmitter via a low-latency and reliable feedback link. On the other hand, the latter constraint is unrealistic in general.

\subsection{Auxiliary Models}
\label{subsec:intro-aux}

This paper's primary motivation is to further develop information theory for adaptive codewords. To gain insight, it is helpful to have achievable rates with $\log(1+\text{SNR})$ terms. A common approach to obtain such expressions is to lower bound the channel mutual information $I(X;Y)$ as follows.

Suppose $X$ is continuous and consider two conditional densities: the density $p(x|y)$ and an auxiliary density $q(x|y)$. We will refer to such densities as \emph{reverse} models; similarly, $p(y|x)$ and $q(y|x)$ are called \emph{forward} models. One may write the differential entropy of $X$ given $Y$ as
\begin{align}
   & h(X|Y) = \E{-\log p(X|Y)} \nonumber \\
   & = \underbrace{\E{-\log q(X|Y)}}_{\textstyle \textrm{average cross-entropy}} - \underbrace{\E{\log \frac{p(X|Y)}{q(X|Y)}}}_{\textstyle \textrm{average divergence} \ge 0}
   \label{eq:h-expansion}
\end{align}
where the first expectation in~\eqref{eq:h-expansion} is an average cross-entropy, and the second is an average informational divergence, which is non-negative. Several criteria affect the choice of $q(x|y)$: the cross-entropy should be simple enough to admit theoretical or numerical analysis, e.g., by Monte Carlo simulation; the cross-entropy should be close to $h(X|Y)$; and the cross-entropy should suggest suitable transmitter and receiver structures.

We illustrate how reverse and forward auxiliary models have been applied to bound mutual information. Assume that $\E{X}=\E{Y}=0$ for simplicity.

\subsubsection{Reverse Model}
Consider the reverse density that models $X,Y$ as jointly CSCG:
\begin{align}
    q(x|y) = \frac{1}{\pi\sigma_L^2} \exp\left(-\left| x - \hat x_L \right|^2 \big/ \sigma_L^2 \right)
\end{align}
where $\hat X_L = \left(\E{X\,Y^\text{*}} /\E{|Y|^2}\right) Y$ and
\begin{align}
    \sigma_L^2
    = \E{\left| X - \hat X_L \right|^2}
    = \E{|X|^2} - \frac{|\E{X Y^\text{*}}|^2}{\E{|Y|^2}}
    \label{eq:VarL}
\end{align}
is the mean square error (MSE) of the estimate $\hat X_L$. In fact, $\hat X_L$ is the linear estimate with the minimum MSE (MMSE), and $\sigma_L^2$ is the linear MMSE (LMMSE) which is independent of $Y=y$; see Sec.~\ref{subsec:LMMSE}. The bound in~\eqref{eq:h-expansion} gives
\begin{align}
    h(X|Y) \le \log\left(\pi e \, \sigma_L^2 \right) .
    \label{eq:h-UB}
\end{align}
Thus, if $X$ is CSCG, then we have the desired form
\begin{align}
    I(X;Y) = h(X)-h(X|Y)
    \ge \log\left( 1 + \frac{|h|^2 \E{|X|^2}}{\sigma^2} \right)
    \label{eq:I-LB}
\end{align}
where the parameters $h$ and $\sigma^2$ are
\begin{align}
    h = \frac{\E{Y X^\text{*}}}{\E{|X|^2}},\quad
    \sigma^2 = \E{| Y - h X |^2}.
    \label{eq:I-LB-parameters}
\end{align}
The bound~\eqref{eq:I-LB} is apparently due to Pinsker~\cite{Pinsker:56,Ihara-78,Pinsker-IT95} and is widely used in the literature; see e.g.~\cite{Shamai-COMM90,Kalet-Shamai-COMM90,Medard-IT00,Diggavi-IT01,Klein-Gallager-ISIT01,Bhashyam-COMM02,Hassibi-Hochwald-IT03,Yoo-Goldsmith-IT06,Agarwal-Honig-IT10,Soysal-Ulukus-WCOM10,marzetta-book,Li-Tao-SML-COMML17,Caire-WCOM18,Noam-Zaidel-SP19}. The bound is usually related to channels $p(y|x)$ with additive noise but \eqref{eq:h-expansion}--\eqref{eq:I-LB} show that it applies generally. The extension to vector channels is given in Sec.~\ref{subsec:max-entropy} below.

\subsubsection{Forward Model}
A more flexible approach is to choose the reverse density as
\begin{align}
    q(x|y) = \frac{p(x) q(y|x)^s}{q(y)}
    \label{eq:qchoice2}
\end{align}
where $q(y|x)$ is a forward auxiliary model (not necessarily a density), $s\ge 0$ is a parameter to be optimized, and
\begin{align}
    q(y) = \int_{\mathbb C} p(x) \, q(y|x)^s \, dx .
    \label{eq:qchoice2a}
\end{align}
Inserting~\eqref{eq:qchoice2} into \eqref{eq:h-expansion} we compute
\begin{align}
    I(X;Y) \ge \max_{s\ge 0} \E{\log \frac{q(Y|X)^s}{q(Y)}} .
    \label{eq:I-LB-GMI}
\end{align}
The right-hand side (RHS) of \eqref{eq:I-LB-GMI} is called a \emph{generalized mutual information} (GMI)~\cite{Kaplan-Shamai-A93,Merhav-IT94} and has been applied to problems in  information theory~\cite{Scarlett-FnT-20}, wireless communications~\cite{Lapidoth-IT96,Lapidoth-Shamai-IT02,Weingarten-IT04,Asyhari-IT14,Oestman-WCOMM21,Zhang-COMM12,Zhang-WSL-JSAC19,Pang-Zhang-ITW21,Wang-Zhang-IT22,Nedelcu-Steiner-Kramer-E22}, and fiber-optic communications~\cite{essiambre2010jlt,Dar:14,Secondini-JLT19,garcia2020mismatched,garcia2021mismatched,garcia2022sdm,Secondini-JLT22,Shtaif-P22,mecozzi_capacity_amdd_2018,plabst2022achievable}. 
For example, the bounds \eqref{eq:I-LB} and \eqref{eq:I-LB-GMI} are the same if $s=1$ and
\begin{align}
    q(y|x) = \exp\left( -|y - h x|^2 \big/ \sigma^2 \right)
    \label{eq:I-LB-q}
\end{align}
where $h$ and $\sigma^2$ are given by~\eqref{eq:I-LB-parameters}. Note that~\eqref{eq:I-LB-q} is not a density unless $\sigma^2=1/\pi$ but $q(x|y)$ is a density.\footnote{We require $q(x|y)$ to be a density to apply the divergence bound in~\eqref{eq:h-expansion}.} 

We compare the two approaches. The bound~\eqref{eq:h-UB} is simple to apply and works well since the choices \eqref{eq:I-LB-parameters} give the maximal GMI for CSCG $X$; see Proposition~\ref{prop:GMI} below. However, there are limitations: one must use continuous $X$, the auxiliary model $q(y|x)$ is fixed as \eqref{eq:I-LB-q}, and the bound does not show how to design the receiver. Instead, the GMI applies to continuous/discrete/mixed $X$ and has an operational interpretation: the receiver uses $q(y|x)$ rather than $p(y|x)$ to decode. The framework of such \emph{mismatched} receivers appeared in~\cite[Ex.~5.22]{Gallager68}; see also~\cite{Divsalar78}.  

\subsection{Refined Auxiliary Models}
\label{subsec:refined-models}

The two approaches above can be refined in several ways, and we review selected variations in the literature.

\subsubsection{Reverse Models}
The model $q(x|y)$ can be different for each $Y=y$, e.g., on may choose $X$ as Gaussian with mean $\E{X|Y=y}$ and variance
\begin{align}
    \Var{X|Y=y}=\E{|X|^2\big|Y=y} - \big| \E{X|Y=y} \big|^2
\end{align}
and where the density $q(x|y)$ is
\begin{align}
    \frac{1}{\pi \Var{X|Y=y}} \exp\left(- \frac{\left|x - \E{X|Y=y}\right|^2}{\Var{X|Y=y}} \right) .
    \label{eq:I-LB-q2r}
\end{align}
Inserting~\eqref{eq:I-LB-q2r} in~\eqref{eq:h-expansion} we have the bound
\begin{align}
    h(X|Y) \le \E{\log\left(\pi e \, \Var{X|Y} \right)}
    \label{eq:h-UB2}
\end{align}
which improves~\eqref{eq:h-UB} in general, since $\Var{X|Y=y}$ is the MMSE of $X$ given the event $Y=y$. In other words, we have $\Var{X|Y=y} \le \sigma_L^2$ for all $Y=y$ and the following bound improves~\eqref{eq:I-LB} for CSCG $X$:
\begin{align}
    I(X;Y) \ge \E{
    \log \frac{\E{|X|^2}}{\Var{X|Y}}}.
    \label{eq:I-LB2}
\end{align}

In fact, the bound~\eqref{eq:I-LB2} was derived in~\cite[Sec.~III.B]{Wang-Zhang-IT22} by optimizing the GMI in~\eqref{eq:I-LB-GMI} over all forward models
\begin{align}
    q(y|x) = \exp\left( -\left|g_y - f_y\, x\right|^2 \right)
    \label{eq:I-LB-q3}
\end{align}
where $f_y$, $g_y$ may depend on $y$; see also~\cite{Zhang-COMM12,Zhang-WSL-JSAC19,Pang-Zhang-ITW21}. We provide a simple proof. By inserting~\eqref{eq:I-LB-q3} into~\eqref{eq:qchoice2}--\eqref{eq:qchoice2a} and completing squares,\footnote{Observe that the $s$ parameter can be absorbed in $f_y$ and $g_y$.} one can equivalently optimize over all reverse Gaussian densities
\begin{align}
    & q(x|y) = \frac{1}{\pi \sigma_y^2} \exp\left( -\frac{\left|x - h_y \right|^2}{\sigma_y^2} \right) . 
    \label{eq:I-LB-q4}
\end{align}
We next bound the cross-entropy as
\begin{align}
    & \E{\left. -\log q(X|Y) \right| Y=y} \nonumber \\
    & \quad = \frac{1}{\sigma_y^2} \E{\left. \left|X - h_y\right|^2 \right| Y=y} + \log\left( \pi \sigma_y^2 \right) \nonumber \\
    & \quad \ge \frac{1}{\sigma_y^2} \Var{X|Y=y} +  \log\left( \pi \sigma_y^2 \right)
    \label{eq:I-LB-q5}
\end{align}
with equality if $h_y=\E{X|Y=y}$; see Sec.~\ref{subsec:LMMSE}. The RHS of~\eqref{eq:I-LB-q5} is minimized by $\sigma_y^2=\Var{X|Y=y}$, so the best choice for $h_y$, $\sigma_y^2$ gives the bound~\eqref{eq:h-UB2}.

\begin{remark}
The model~\eqref{eq:I-LB-q3} uses \emph{generalized nearest-neighbor decoding}, improving the rules proposed in~\cite{Lapidoth-IT96,Lapidoth-Shamai-IT02,Weingarten-IT04}. The authors of~\cite{Wang-Zhang-IT22} pointed out that~\eqref{eq:I-LB} and ~\eqref{eq:I-LB2} use the LMMSE and MMSE, respectively; see~\cite[Eq.~(87)]{Wang-Zhang-IT22}. 
\end{remark}

\begin{remark}
A corresponding forward model can be based on~\eqref{eq:qchoice2} and~\eqref{eq:I-LB-q2r}, namely
\begin{align}
    q(y|x)^s = \frac{q(x|y)}{p(x)}
    \quad \Rightarrow \quad q(y) = 1.
    \label{eq:I-LB-q2}
\end{align}
\end{remark}

\begin{remark}
The RHS of~\eqref{eq:I-LB2} has a more complicated form than the RHS of~\eqref{eq:I-LB} due to the outer expectation and conditional variance, and this makes optimizing $X$ challenging when there is CSIR and CSIT. Also, if $p(y|x)$ is known, then it seems sensible to numerically compute $p(y)$ and $I(X;Y)$ directly, e.g., via Monte Carlo or numerical integration.
\end{remark}

\begin{remark}
Decoding rules for discrete $X$ can be based on decision theory as well as estimation theory; see~\cite[Eq.~(11)]{Ozarow-Wyner-IT90}.
\end{remark}

\subsubsection{Forward Models}
Refinements of~\eqref{eq:I-LB-q} appear in the optical fiber literature where the non-linear Schr\"odinger equation describes wave propagation~\cite{essiambre2010jlt}. Such channels exhibit complicated interactions of attenuation, dispersion, nonlinearity, and noise, and the channel density is too challenging to compute. One thus resorts to capacity lower bounds based on GMI and Monte Carlo simulation. The simplest models are memoryless, and they work well if chosen carefully. For example, the paper~\cite{essiambre2010jlt} used auxiliary models of the form
\begin{align}
    q(y|x) = \exp\left( -|y - h x|^2 \big/ \sigma_{|x|}^2 \right)
    \label{eq:q-optical}
\end{align}
where $h$ accounts for attenuation and self-phase modulation, and where the noise variance $\sigma_{|x|}^2$ depends on $|x|$. Also, $X$ was chosen to have concentric rings rather than a CSCG density. Subsequent papers applied progressively more sophisticated models with memory to better approximate the actual channel; see~\cite{Dar:14,Secondini-JLT19,garcia2020mismatched,garcia2021mismatched,garcia2022sdm,Secondini-JLT22,Shtaif-P22}. However, the rate gains over the model~\eqref{eq:q-optical} are minor ($\approx$12\%) for 1000 km links, and the newer models do not suggest practical receiver structures.

A related application is short-reach fiber-optic systems that use direct detection (DD) receivers~\cite{chagnon_optical_comms_short_reach_2019} with photodiodes. The paper~\cite{mecozzi_capacity_amdd_2018} showed that sampling faster than the symbol rate increases the DD capacity. However, spectrally efficient filtering gives the channel a long memory, motivating auxiliary models $q(y|x)$ with reduced memory to simplify GMI computations~\cite{Arnold-etal-IT06,plabst2022achievable}. More generally, one may use channel-shortening filters~\cite{Abou-Faycal00,Rusek-WCOM12,Hu-Rusek-IA18} to increase the GMI.

\begin{remark} 
\label{remark:ultimate}
The ultimate GMI is $I(X;Y)$, and one can compute this quantity numerically for the channels considered in this paper. We are motivated to focus on forward auxiliary models $q(y|x)$ to understand how to improve information rates for more complex channels. For instance, simple $q(y|x)$ let one understand properties of optimal codes, see Lemma~\ref{lemma:AGMI-max}, and they suggest explicit power control policies, see Theorem~\ref{theorem:Partial-CSIR-Full-CSIT}.
\end{remark}

\begin{remark} 
\label{remark:caire}
The paper~\cite{Caire-WCOM18} (see also~\cite[Eq.~(3.3.45)]{Biglieri-Proakis-Shamai-IT98} and~\cite[eq.~(6)]{Mezghani-Nossek-ISIT02}) derives two capacity lower bounds for massive MIMO channels. These bounds are designed for problems where the fading parameters have small variance so that, in effect, $\sigma^2$ in \eqref{eq:I-LB-parameters} is small. We will instead encounter cases where $\sigma^2$ grows in proportion to $\E{|X|^2}$ and the RHS of~\eqref{eq:I-LB} quickly saturates as $\E{|X|^2}$ grows; see Remark~\ref{remark:caire2}.
\end{remark}

\subsection{Organization}
\label{subsec:prelim-organization}
This paper is organized as follows.
Sec.~\ref{sec:prelim} defines notation and reviews basic results.
Sec.~\ref{sec:GMI} develops two results for the GMI of scalar auxiliary models with AWGN: 
\begin{itemize}
\item Proposition~\ref{prop:GMI} in Sec.~\ref{subsec:GMI-Gauss} states a known result, namely that the RHS of \eqref{eq:I-LB} is the maximum GMI for the AWGN auxiliary model~\eqref{eq:I-LB-q} and a CSCG $X$.
\item Lemma~\ref{lemma:GMI-part} in Sec.~\ref{subsec:GMI-generalizations} generalizes Proposition~\ref{prop:GMI} by partitioning the channel output alphabet into $K$ subsets, $K\ge1$. We use $K=2$ to establish capacity properties at high and low SNR.
\end{itemize}
Sec.~\ref{sec:CSIT}--\ref{sec:CSIRT} apply the GMI to channels with CSIT and CSIR.
\begin{itemize}
\item Sec.~\ref{subsec:CSIT-opt} treats adaptive codewords and develops structural properties of their optimal distribution. 
\item Lemma~\ref{lemma:AGMI} in Sec.~\ref{subsec:CSIT-GMI} generalizes Proposition~\ref{prop:GMI} to MIMO channels and adaptive codewords. The receiver models each transmit symbol as a weighted sum of the entries of the corresponding adaptive symbol.
\item Lemma~\ref{lemma:AGMI-max} in Sec.~\ref{subsec:CSIT-codebook-structure-GMI} states that the maximum GMI for scalar channels, an AWGN auxiliary model, adaptive codewords with jointly CSCG entries, and $K=1$ is achieved by using a conventional codebook where each symbol is modified based on the CSIT.
\item Lemma~\ref{lemma:AGMI-max-MIMO} in Sec.~\ref{subsec:CSIT-GMI-MIMO} extends Lemma~\ref{lemma:AGMI-max} to MIMO channels, including diagonal or parallel channels.
\item Theorem~\ref{theorem:AGMI-SR-max} in Sec.~\ref{subsec:CSIRT-capacity} generalizes Lemma~\ref{lemma:AGMI-max} to include CSIR; we use this result several times in Sec.~\ref{sec:Gauss-fading}.
\item Lemma~\ref{lemma:AGMI-part} in Sec.~\ref{subsec:GMI-generalizations2} generalizes Lemmas~\ref{lemma:GMI-part} and~\ref{lemma:AGMI} by partitioning the channel output alphabet.
\end{itemize}
%
Sec.~\ref{sec:Gauss-fading}--\ref{sec:Rayleigh-Fading} apply the GMI to fading channels with AWGN and illustrate the theory for on-off and Rayleigh fading.
\begin{itemize}
\item Lemma~\ref{lemma:Gauss-fading-C-UB} in Sec.~\ref{sec:Gauss-fading} gives a general capacity upper bound.
\item Sec.~\ref{subsec:Partial-CSIR-fullCSIT} introduces a class of power control policies for full CSIT. Theorem~\ref{theorem:Partial-CSIR-Full-CSIT} develops the optimal policy with an MMSE form.
\item Theorem~\ref{theorem:partial-CSIR} in Sec.~\ref{subsec:Partial-CSIR-CSITatR} provides a quadratic waterfilling expression for the GMI with partial CSIR.
\end{itemize}
Sec.~\ref{sec:ibf-CSIT} develops theory for block fading channels with in-block feedback (or in-block CSIT) that is a function of the CSIR and past channel inputs and outputs.
\begin{itemize}
\item Theorem~\ref{theorem:AGMI-SR-max-block-fading} in Sec.~\ref{subsec:ibf-GMI-scalar} generalizes Lemma~\ref{lemma:AGMI-max-MIMO} to MIMO block fading channels;
\item Sec.~\ref{subsec:ibf-CSITatR} develops capacity expressions in terms of directed information;
\item Sec.~\ref{subsec:ibf-fading-AWGN} specializes the capacity to fading channels with AWGN and delayed CSIR;
\item Proposition~\ref{proposition:part-CSIT2-block-fading} generalizes Proposition~\ref{proposition:part-CSIT2} to channels with special CSIR and CSIT.
\end{itemize}
Sec.~\ref{sec:Conclusion} concludes the paper. Finally, Appendices~\ref{appendix:special-functions}--\ref{appendix:proof-lemma-AGMI-max-vec} provide results on special functions, GMI calculations, and proofs.

\section{Preliminaries}
\label{sec:prelim}
\subsection{Basic Notation}
Let $1(\cdot)$ be the indicator function that takes on the value 1 if its argument is true and 0 otherwise. Let $\delta(.)$ be the Dirac generalized function with $\int_\set{X} \delta(x) f(x) dx= f(0)\cdot1(0\in\set{X})$. For $x\in\mathbb R$, define $(x)^+=\max(0,x)$. The complex-conjugate, absolute value, and phase of $x\in\mathbb C$ are written as $x^\text{*}$, $|x|$, and $\arg(x)$, respectively. We write $j=\sqrt{-1}$ and $\epsb=1-\epsilon$.

Sets are written with calligraphic font, e.g., $\set S=\{1,\dots,n\}$ and the cardinality of $\set S$ is $|\set S|$. The complement of $\set S$ in $\set T$ is $\set S^c$ where $\set T$ is understood from the context.

\subsection{Vectors and Matrices}
Column vectors are written as $\ul x = [x_1,\dots,x_M]^T$ where $M$ is the dimension, and $T$ denotes transposition. The complex-conjugate transpose (or Hermitian) of $\ul x$ is written as $\ul x^\dag$. The Euclidean norm of $\ul x$ is $\| \ul x \|$.
Matrices are written with bold letters such as $\bf A$. The letter $\bf I$ denotes the identity matrix. The determinant and trace of a square matrix $\bf A$ are written as $\det{\bf A}$  and ${\rm tr}\,{\bf A}$, respectively.

A singular value decomposition (SVD) is ${\bf A} = {\bf U} {\bf \Sigma} {\bf V}^\dag$ where $\bf U$ and $\bf V$ are unitary matrices and $\bf \Sigma$ is a rectangular diagonal matrix with the singular values of $\bf A$ on the diagonal. The square matrix $\bf A$ is positive semi-definite if $\ul x^\dag {\bf A} \ul x \ge 0$ for all $\ul x$. The notation ${\bf A} \preceq {\bf B}$ means that ${\bf B}-{\bf A}$ is positive semi-definite. Similarly, $\bf A$ is positive definite if $\ul x^\dag {\bf A} \ul x > 0$ for all $\ul x$, and we write ${\bf A} \prec {\bf B}$ if ${\bf B}-{\bf A}$ is positive definite.

\subsection{Random Variables}
Random variables are written with uppercase letters, such as $X$, and their realizations with lowercase letters, such as $x$. We write the distribution of discrete $X$ with alphabet $\set X=\{0,\dots,n-1\}$ as $P_X=[P_X(0),\dots,P_X(n-1)]$. The density of a real- or complex-valued $X$ is written as $p_X$. Mixed discrete-continuous distributions are written using mixtures of densities and Dirac-$\delta$ functions.

Conditional distributions and densities are written as $P_{X|Y}$ and $p_{X|Y}$, respectively. We usually drop subscripts if the argument is a lowercase version of the random variable, e.g., we write $p(y|x)$ for $p_{Y|X}(y|x)$. One exception is that we consistently write the distributions $P_{S_R}(.)$ and $P_{S_T}(.)$ of the CSIR and CSIT with the subscript to avoid confusion with power notation.

\subsection{Second-Order Statistics}
\label{subsec:prelim-order2}
The expectation and variance of the complex-valued random variable $X$ are $\E{X}$ and $\Var{X}=\E{|X-\E{X}|^2}$, respectively. The correlation coefficient of $X_1$ and $X_2$ is $\rho = \E{U_1 U_2^\text{*}}$ where
\begin{align*}
    U_i=(X_i-\E{X_i})/\sqrt{\Var{X_i}}
\end{align*}
for $i=1,2$. We say that $X_1$ and $X_2$ are \emph{fully correlated} if $\rho = e^{j\phi}$ for some real $\phi$. Conditional expectation and variance are written as $\E{X|A=a}$ and
\begin{align*}
    \Var{X|A=a}=\E{(X-\E{X}) (X-\E{X})^\text{*} | A=a}.
\end{align*}
The expressions $\E{X|A}$, $\Var{X|A}$ are random variables that take on the values $\E{X|A=a}$, $\Var{X|A=a}$ if $A=a$.

We slightly simplify and abuse notation by carrying explicit conditioning across expectations:
\begin{align*}
    \E{\E{X|Y,Z=z}} := \E{\E{X|Y,Z}|Z=z}.
\end{align*}
For instance, with this convention, we could have written the left-hand side of~\eqref{eq:I-LB-q5} as $\E{-\log q(X|Y=y)}$.

The expectation and covariance matrix of the random column vector $\ul X = [X_1,\dots,X_M]^T$ are $\E{\ul X}$ and ${\bf Q}_{\ul X}=\E{(\ul X-\E{\ul X}) (\ul X-\E{\ul X})^\dag}$, respectively. We write ${\bf Q}_{\ul X,\ul Y}$ for the covariance matrix of the stacked vector $[\ul X^T \ul Y^T]^T$. We write ${\bf Q}_{\ul X|\ul Y=\ul y}$ for the covariance matrix of $\ul X$ conditioned on the event $\ul Y = \ul y$. ${\bf Q}_{\ul X|\ul Y}$ is a random matrix that takes on the value ${\bf Q}_{\ul X|\ul Y=\ul y}$ when $\ul Y = \ul y$.

We often consider CSCG random variables and vectors. A CSCG $\ul X$ has density
\begin{align*}
  p(\ul x) = \frac{\exp\left(- {\ul x}^\dag \, {\bf Q}_{\ul X}^{-1} \, \ul{x}\right)}{\pi^M \det {\bf Q}_{\ul X}} 
\end{align*}
and we write $\ul X \sim \mathcal{CN}(\ul 0, {\bf Q}_{\ul X})$.

\subsection{MMSE and LMMSE Estimation}
\label{subsec:LMMSE}
Assume that $\E{\Xv}=\E{\Yv}=\ul 0$. The MMSE estimate of $\ul X$ given the event $\ul Y=\ul y$ is the vector $\hat{\ul X}(\ul y)$ that minimizes
\begin{align*}
    \E{\left. \left\|\ul X-\hat{\ul X}(\ul y) \right\|^2 \right| \ul Y=\ul y}.
\end{align*}
Direct analysis gives~\cite[Ch.~4]{papoulis-book}
\begin{align}
    & \hat{\ul X}(\ul y) = \E{\ul X | \ul Y=\ul y}
    \label{eq:MMSE-est} \\
    & \E{\big\|\ul X-\hat{\ul X}\big\|^2} 
    = \E{\|{\ul X}\|^2} - \E{\big\|\hat{\ul X}\big\|^2}
    \label{eq:MMSE-error} \\
    & {\bf Q}_{\ul X - \hat{\ul X}}
    = {\bf Q}_{\ul X} - {\bf Q}_{\hat{\ul X}} 
    \label{eq:MMSE-var} \\
    & \E{\left( \ul X-\hat{\ul X} \right) \, \ul Y^\dag} = {\bf 0}
    \label{eq:MMSE-op}
\end{align}
where the last identity is called the \emph{orthogonality principle}.

The LMMSE estimate of $\ul X$ given $\ul Y$ with invertible ${\bf Q}_{\ul Y}$ is the vector $\hat{\ul X}_L = {\bf C} \, \ul Y$ where $\bf C$ is chosen to minimize $\E{\|\ul X-\hat{\ul X}_L\|^2}$. We compute
\begin{align}
    & \hat{\ul X}_L = \E{\ul X \, \ul Y^\dag} {\bf Q}_{\ul Y}^{-1} \, \ul Y
    \label{eq:LMMSE-est} 
\end{align}
and we also have the properties \eqref{eq:MMSE-error}--\eqref{eq:MMSE-op} with $\hat{\ul X}$ replaced by $\hat{\ul X}_L$. Moreover, if $\ul X$ and $\ul Y$ are jointly CSCG, then the MMSE and LMMSE estimators coincide, and~\eqref{eq:MMSE-op} implies that the error $\ul X - \hat{\ul X}$ is independent of $\ul Y$, i.e., we have
\begin{align}
   & \E{\left. \left(\ul X-\hat{\ul X}\right)\left(\ul X-\hat{\ul X}\right)^\dag \right| \ul Y = \ul y} \nonumber \\
   & = \E{\left. \ul X \, \ul X^\dag \right| \ul Y = \ul y} - \E{\ul X \, \ul Y^\dag} {\bf Q}_{\ul Y}^{-1} \, \ul y \, \ul y^\dag {\bf Q}_{\ul Y}^{-1} \E{\ul X \, \ul Y^\dag}^\dag
   \nonumber \\
   & = {\bf Q}_{\ul X} - {\bf Q}_{\hat{\ul X}} .\label{eq:LMMSE-indep} 
\end{align}

\subsection{Entropy, Divergence, and Information}
\label{subsec:prelim-IT}
Entropies of random vectors with densities $p$ are written as
\begin{align*}
   h({\ul X}) & = \E{-\log p({\ul X})} \\
   h({\ul X} | {\ul Y}) & = \E{-\log p({\ul X} | {\ul Y})}
\end{align*}
where we use logarithms to the base $e$ for analysis. The informational divergence of the densities $p$ and $q$ is
\begin{align*}
   D\left( p \| q \right) = \E{\log \frac{p({\ul X})}{q({\ul X})}}
\end{align*}
and $D(p \| q) \ge 0$ with equality if and only if $p=q$ almost everywhere. The mutual information of $\ul X$ and $\ul Y$ is
\begin{align*}
   I({\ul X};{\ul Y}) & = D\left( p({\ul X},{\ul Y}) \,\|\, p({\ul X}) \, p({\ul Y}) \right) \nonumber \\
   & = \E{\log \frac{p({\ul Y}|{\ul X})}{p({\ul Y})}}.
\end{align*}
The average mutual information of $\ul X$ and $\ul Y$ conditioned on $\ul Z$ is $I({\ul X};{\ul Y}|{\ul Z})$. We write strings as $X^L=(X_1,X_2,\ldots,X_L)$ and use the directed information notation (see~\cite{Kramer98,Kramer03})
\begin{align}
   & I(X^L \rightarrow Y^L | Z) = \sum_{\ell=1}^L I(X^\ell ; Y_\ell | Y^{\ell-1}, Z) \label{eq:DI} \\
   & I(X^L \rightarrow Y^L \| Z^L | W) = \sum_{\ell=1}^L I(X^\ell ; Y_\ell | Y^{\ell-1}, Z^\ell, W) \label{eq:ccDI}
\end{align}
where $Y_0=0$.

\subsection{Entropy and Information Bounds}
\label{subsec:max-entropy}
The expression \eqref{eq:h-expansion} applies to random vectors. Choosing $q(\ul x | \ul y)$ as the conditional density where the $\ul X,\ul Y$ are modeled as jointly CSCG we obtain a generalization of \eqref{eq:h-UB}:
\begin{align}
   & h({\ul X} | {\ul Y}) \le \log \frac{\det\left(\pi e \, {\bf Q}_{\ul X,\ul Y}\right)}{\det\left(\pi e \, {\bf Q}_{\ul Y}\right)} \nonumber \\
   & = \log\det\left( \pi e \left\{ {\bf Q}_{\ul X} - \E{\ul X \, \ul Y^\dag} {\bf Q}_{\ul Y}^{-1} \E{\ul Y \, \ul X^\dag} \right\} \right).
   \label{eq:max-ent2}
\end{align}
The vector generalization of~\eqref{eq:I-LB} for CSCG $\ul X$ is
\begin{align}
   & I({\ul X};{\ul Y}) = h({\ul X}) - h({\ul X} | {\ul Y}) \nonumber \\
   & \ge \log\det\left( \left( {\bf Q}_{\ul X} - \E{\ul X \, \ul Y^\dag} {\bf Q}_{\ul Y}^{-1} \E{\ul Y \, \ul X^\dag} \right)^{-1} {\bf Q}_{\ul X} \right)
   \nonumber \\
   & \overset{(a)}{=} \log\det\left( {\bf I} + {\bf Q}_{\ul Z}^{-1} {\bf H} {\bf Q}_{\ul X}^{-1} {\bf H}^\dag \right)
   \label{eq:I-LB-vector}
\end{align}
where (cf.~\eqref{eq:I-LB-parameters})
\begin{align}
   {\bf H} = \E{\ul Y \, {\ul X}^\dag} {\bf Q}_{\ul{\bar X}}^{-1} ,\quad
   {\bf Q}_{\ul{Z}} = {\bf Q}_{\ul Y} - {\bf H} {\bf Q}_{\ul X} {\bf H}^\dag
   \label{eq:I-LB-vector-parameters}
\end{align}
and step $(a)$ in~\eqref{eq:I-LB-vector} follows by the Woodbury identity
\begin{align}
    & \left( {\bf A} + {\bf B}  {\bf C}  {\bf D} \right)^{-1} \nonumber \\
    & \quad = {\bf A}^{-1} - {\bf A}^{-1} {\bf B} \left( {\bf C}^{-1} + {\bf D}  {\bf A}^{-1} {\bf B} \right)^{-1} {\bf D} {\bf A}^{-1}
    \label{eq:Woodbury}
\end{align}
and the Sylvester identity
\begin{align}
    & \det\left( {\bf I} + {\bf A}  {\bf B} \right)
    = \det\left( {\bf I} + {\bf B}  {\bf A} \right).
    \label{eq:Sylvester}
\end{align}
We also have vector generalizations of~\eqref{eq:h-UB2}--\eqref{eq:I-LB2}:
\begin{align}
    h(\ul X | \ul Y) & \le \E{\log \det\left( \pi e \, {\bf Q}_{\ul X | \ul Y}  \right)}
    \label{eq:h-UB3}\\
   I({\ul X};{\ul Y}) & \ge \E{\log \frac{\det{\bf Q}_{\ul X}}{\det{\bf Q}_{\ul X| \ul Y}}}, \;\; \text{for CSCG $\ul X$}.
   \label{eq:I-LB-vector2}
\end{align}

\subsection{Capacity and Wideband Rates}
\label{subsec:wideband}
Consider the complex-alphabet AWGN channel with output $Y=X+Z$ and noise $Z\sim\mathcal{CN}(0,1)$. The capacity with the block power constraint
$\frac{1}{n} \sum_{i=1}^n |X_i|^2 \le P$ is
\begin{align}
   C(P) = \max_{\E{|X|^2}\le P} I(X;Y) = \log(1+P).
   \label{eq:AWGN-capacity}
\end{align}

The low SNR regime (small $P$) is known as the wideband regime~\cite{Verdu02}. For well-behaved channels such as AWGN channels, the minimum $E_b/N_0$ and the slope $S$ of the capacity vs. $E_b/N_0$ in bits/(3 dB) at the minimum $E_b/N_0$ are (see~\cite[Eq. (35)]{Verdu02} and~\cite[Thm.~9]{Verdu02})
\begin{align}
   \left.\frac{E_b}{N_0}\right|_{\text{min}} = \frac{\log 2}{C'(0)}, \quad
   S = \frac{2 [C'(0)]^2}{- C''(0)}
   \label{eq:wideband}
\end{align}
where $C'(P)$ and $C''(P)$ are the first and second derivatives of $C(P)$ (measured in nats) with respect to $P$, respectively. For example, the wideband derivatives for \eqref{eq:AWGN-capacity} are $C'(0)=1$ and $C''(0)=-1$ so that the wideband values \eqref{eq:wideband} are
\begin{align}
   \left.\frac{E_b}{N_0}\right|_{\text{min}} = \log 2, \quad S = 2 .
   \label{eq:wideband-fullCSIR-noCSIT-awgn}
\end{align}
The minimal $E_b/N_0$ is usually stated in decibels, for example $10 \log_{10}(\log 2)=-1.59$ dB. An extension of the theory to general channels is described in~\cite[Sec.~III]{kramer_SE_2003}.

\begin{remark} \label{remark:flash}
A useful method is \emph{flash} signaling, where one sends with zero energy most of the time. In particular, we will consider the CSCG flash density\footnote{Flash signaling is defined in~\cite[Def.~2]{Verdu02} as a family of distributions satisfying a particular property as $P\rightarrow0$. We use the terminology informally.}
\begin{align}
   p(x) = (1-p) \, \delta(x) + p \, \frac{e^{-|x|^2 / (P/p)}}{\pi (P/p)}
   \label{eq:flash-Gauss}
\end{align}
where $0 < p \le 1$ so that the average power is $\E{|X|^2}=P$.
\end{remark}

\subsection{Uniformly-Spaced Quantizer}
\label{subsec:quantizer}
Consider a uniformly-spaced scalar quantizer $q_u(.)$ with $B$ bits, domain $[0,\infty)$, and reconstruction points
\begin{align*}
   s  \in \{\Delta/2,3\Delta/2,\ldots,\Delta/2+(2^B-1)\Delta\}
\end{align*}
where $\Delta>0$. The quantization intervals are
\begin{align*}
   \set{I}(s) = \left\{ \begin{array}{ll}
   \left[ s-\frac{\Delta}{2},  s+\frac{\Delta}{2} \right), & s \ne s_{\rm max} \\
   \left[ s-\frac{\Delta}{2},  \infty \right), & s = s_{\rm max}
   \end{array} \right.
\end{align*}
where $s_{\rm max}=\Delta/2+(2^B-1)\Delta$.  We will consider $B=0,1,\infty$. For $B=\infty$ we choose $q_u(x)=x$.

Suppose one applies the quantizer to the non-negative random variable $G$ with density $p(g)$ to obtain $S_T=q_u(G)$. Let $P_{S_T}$ and $P_{S_T|G}$ be the probability mass functions of $S_T$ without and with conditioning on $G$, respectively. We have
\begin{align}
   & P_{S_T|G}(s|g) = 1\left(g\in\set{I}(s)\right) \nonumber \\
   & P_{S_T}(s) = \int_{g\in\set{I}(s)} p(g) \; dg
   \label{eq:PST}
\end{align}
and using Bayes' rule, we obtain
\begin{align}
   p(g | s)
   & = \left\{ \begin{array}{ll}
   p(g)/P_{S_T}(s), & g \in \set{I}(s) \\
   0, & \text{else} .
   \end{array} \right.
   \label{eq:pgs2}
\end{align}
  
\section{Generalized Mutual Information}
\label{sec:GMI}
We re-derive the GMI in the usual way, where one starts with the forward model $q(y|x)$ rather than the reverse density $q(x|y)$ in~\eqref{eq:qchoice2}. Consider the joint density $p(x,y)$ and define $q(y)$ as in~\eqref{eq:qchoice2a} for $s\ge 0$. Note that neither $q(y|x)$ nor $q(y)$ must be densities. The GMI is defined in~\cite{Kaplan-Shamai-A93} to be $\max_{s\ge0} I_s(X;Y)$ where (see the RHS of~\eqref{eq:I-LB-GMI})
\begin{align}
   I_s(X;Y) = \E{ \log \frac{q(Y|X)^s}{q(Y)} }
   \label{eq:GMI}
\end{align}
and where the expectation is with respect to $p(x,y)$. The GMI is a lower bound on the mutual information since
\begin{align}
   I_s(X;Y) = I(X;Y) - D\left( \left. p_{X,Y} \right\| p_Y \, q_{X|Y} \right) .
   \label{eq:GMI-identity}
\end{align}
Moreover, by using Gallager's derivation of error exponents, but without modifying his ``$s$'' variable, the GMI $I_s(X;Y)$ is achievable with a mismatched decoder that uses $q(y|x)$ for its decoding metric~\cite{Kaplan-Shamai-A93}.

\subsection{AWGN Forward Model with CSCG Inputs}
\label{subsec:GMI-Gauss}
A natural metric is based on the AWGN auxiliary channel $Y_a=h X+Z$ where $h$ is a channel parameter and $Z\sim\mathcal{CN}(0,\sigma^2)$ is independent of $X$, i.e., we have the auxiliary model (here a density)
\begin{align}
   q(y|x) = \frac{1}{\pi \sigma^2} \exp\left(-|y - h x |^2 /\sigma^2 \right)
   \label{eq:GMI-qyx}
\end{align}
where $h$ and $\sigma^2$ are to be optimized. A natural input is $X\sim\mathcal{CN}(0,P)$ so that~\eqref{eq:qchoice2a} is
\begin{align}
   q(y) = \frac{\pi \sigma^2/s}{(\pi \sigma^2)^s} \cdot
   \frac{\exp\left( \frac{-|y|^2}{\sigma^2/s + |h|^2P} \right)}{\pi(\sigma^2/s + |h|^2 P)} .
   \label{eq:GMI-qy}
\end{align}
We have the following result, see~\cite{Lapidoth-Shamai-IT02} that considers channels of the form~\eqref{eq:block-fading} and~\cite[Prop.~1]{Zhang-COMM12} that considers general $p(y|x)$.

\begin{proposition} \label{prop:GMI}
The maximum GMI \eqref{eq:GMI} for the channel $p(y|x)$, a CSCG input $X$ with variance $P>0$, and the auxiliary model \eqref{eq:GMI-qyx} with $\sigma^2>0$ is
\begin{align}
   I_1(X;Y) = \log\left( 1 + \frac{|\tilde h|^2 P}{\tilde \sigma^2} \right)
   \label{eq:GMI-2}
\end{align}
where $s=1$ and (cf.~\eqref{eq:I-LB-parameters})
\begin{align}
   \tilde h & = \E{Y X^\text{*}} \big/ P
   \label{eq:GMI-parameters-a} \\
   \tilde \sigma^2 & = \E{|Y - \tilde h \,X|^2}
   = \E{|Y|^2} - |\tilde h|^2 P .
   \label{eq:GMI-parameters-b}
\end{align}
The expectations are with respect to the actual density $p(x,y)$.
\end{proposition}
\begin{proof}
The GMI \eqref{eq:GMI} for the model~\eqref{eq:GMI-qyx} is
\begin{align}
   I_s(X;Y) & = \log\left( 1 + \frac{|h|^2 P}{\sigma^2/s} \right) \nonumber \\
   & \quad + \frac{\E{|Y|^2}}{\sigma^2/s + |h|^2 P} - \frac{\E{|Y - h X|^2}}{\sigma^2/s}.
   \label{eq:GMI-3}
\end{align}
Since \eqref{eq:GMI-3} depends only on the ratio $\sigma^2/s$ one may as well set $s=1$. Thus, choosing $h=\tilde h$ and $\sigma^2=\tilde \sigma^2$ gives \eqref{eq:GMI-2}.

Next, consider $Y_a = \tilde h X + \Zt$ where $\Zt\sim\mathcal{CN}(0,\tilde \sigma^2)$ is independent of $X$. We have
\begin{align}
    & \E{\big|Y_a\big|^2}=\E{|Y|^2}
    \label{eq:aux-equiv1} \\
    & \E{\big|Y_a - \tilde h X\big|^2} = \E{\big|Y - \tilde h X\big|^2}.
    \label{eq:aux-equiv2}
\end{align}
In other words, the second-order statistics for the two channels with outputs $Y$ (the actual channel output) and $Y_a $ are the same. But the GMI \eqref{eq:GMI-2} is the mutual information $I(X;Y_a)$. Using \eqref{eq:GMI-identity} and \eqref{eq:GMI-3}, for any $s$, $h$ and $\sigma^2$ we have
\begin{align}
   I(X;Y_a) & = \log\left( 1 + \frac{|\tilde h|^2 P}{\tilde \sigma^2} \right) \nonumber \\
   & \ge I_s(X;Y_a) = I_s(X;Y) 
   \label{eq:GMI-bound}
\end{align}
and equality holds if $h=\tilde h$ and $\sigma^2/s=\tilde \sigma^2$.
\end{proof} 

\begin{remark}
The rate~\eqref{eq:GMI-2} is the same as the RHS of~\eqref{eq:I-LB}.
\end{remark}

\begin{remark}
Proposition~\ref{prop:GMI} generalizes to vector models and adaptive input symbols; see Sec.~\ref{subsec:CSIT-GMI}.
\end{remark}

\begin{remark}
The estimate $\tilde h$ is the MMSE estimate of $h$:
\begin{align}
  \tilde h = \arg \min_h \E{| Y - h X |^2}
  \label{eq:h-MMSE-est}
\end{align}
and $\tilde \sigma^2$ is the variance of the error. To see this, expand
\begin{align}
  \E{| Y - h X |^2} & = \E{| (Y - \tilde h X) + (\tilde h - h) X |^2} \nonumber \\
  & = \tilde \sigma^2 + |\tilde h - h|^2 P
  \label{eq:MMSE-proof}
\end{align}
where the final step follows by the definition of $\tilde h$ in \eqref{eq:GMI-parameters-a}.
\end{remark}

\begin{remark} \label{remark:other-sigma2}
Suppose $h$ is an estimate other than \eqref{eq:h-MMSE-est}. Then if $\E{|Y|^2}>\E{\left| Y - h \,X \right|^2}$ we may choose
\begin{align}
   \sigma^2/s = |h|^2 P \cdot
   \frac{\E{\left| Y - h \,X \right|^2}}
   {\E{|Y|^2} - \E{\left| Y - h \,X \right|^2}}
   \label{eq:GMI-parameters-sigma2}
\end{align}
and the GMI \eqref{eq:GMI-3} simplifies to
\begin{align}
   I_s(X;Y) & = \log\left( \frac{\E{|Y|^2}}{\E{\left| Y - h \,X \right|^2}} \right).
   \label{eq:GMI-3h}
\end{align}
\end{remark}

\begin{remark}
The \emph{LM rate} (for ``lower bound to the mismatch capacity'') improves the GMI for some $q(y|x)$~\cite{Hui:83,Merhav-IT94}. The LM rate replaces $q(y|x)$ with $q(y|x) e^{t(x)/s}$ for some function $t(.)$ and permits optimizing $s$ and $t(.)$; see~\cite[Sec.~2.3.2]{Scarlett-FnT-20}. For example, if $p(y|x)$ has the form  $q(y|x)^s e^{t(x)}$ then the LM rate can be larger than the GMI; see~\cite{Scarlett-IT14,Kangarshahi-IT21}.
\end{remark}

\subsection{CSIR and $K$-Partitions}
\label{subsec:GMI-generalizations}
We consider two generalizations of Proposition~\ref{prop:GMI}. The first is for channels with a state $S_R$ known at the receiver but not at the transmitter. The second expands the class of CSCG auxiliary models. The motivation is to obtain more precise models under partial CSIR, especially to better deal with channels at high SNR and with high rates. We here consider discrete $S_R$ and later extend to continuous $S_R$.

\subsubsection{CSIR}
Consider the average GMI
\begin{align}
   I_1(X;Y | S_R) = \sum_{s_R} P_{S_R}(s_R) \, I_1(X;Y | S_R = s_R)
   \label{eq:GMI-2-avg}
\end{align}
where $I_1(X;Y | S_R = s_R)$ is the usual GMI where all densities are conditioned on $S_R = s_R$. The parameters \eqref{eq:GMI-parameters-a}--\eqref{eq:GMI-parameters-b} for the event $S_R=s_R$ are now
\begin{align}
   \tilde h(s_R) & = \frac{\E{\left. Y X^\text{*} \right| S_R = s_R}}{\E{\left. |X|^2 \right| S_R = s_R}}
   \label{eq:GMI-parameters-c1} \\
   \tilde \sigma^2 (s_R) & = \E{\left. |Y - \tilde h(s_R) \,X|^2 \right| S_R = s_R}.
   \label{eq:GMI-parameters-c2}
\end{align}
The GMI \eqref{eq:GMI-2-avg} is thus
\begin{align}
    I_1(X;Y | S_R) = \sum_{s_R} P_{S_R}(s_R) \, \log\left( 1 + \frac{|\tilde h(s_R)|^2 P}{\tilde \sigma(s_R)^2} \right).
    \label{eq:GMI-2-avg-a}
\end{align}

\subsubsection{$K$-Partitions}
Let $\{\set Y_k: k=1,\dots,K\}$  be a $K$-partition of $\set Y$ and define the auxiliary model
\begin{align}
   q(y | x) = \frac{1}{\pi \sigma_k^2} e^{-|y - h_k x |^2 /\sigma_k^2}, \quad y \in \set Y_k.
   \label{eq:GMI-qyx-part}
\end{align}
Observe that $q(y|x)$ is not necessarily a density. We choose $X\sim\mathcal{CN}(0,P)$ so that~\eqref{eq:qchoice2a} becomes (cf.~\eqref{eq:GMI-qy})
\begin{align}
   q(y) = \frac{\pi \sigma_k^2/s}{(\pi \sigma_k^2)^s} \cdot
   \frac{\exp\left( \frac{-|y|^2}{\sigma_k^2/s + |h_k|^2P} \right)}{\pi(\sigma_k^2/s + |h_k|^2 P)}, \quad y \in \set Y_k.
   \label{eq:GMI-qy-part}
\end{align}
Define the events $\set E_k = \{Y \in \set Y_k\}$ for $k=1,\dots,K$. We have
\begin{align}
      I_s(X;Y)
      & = \sum_{k=1}^K \Pr{\set E_k} \cdot \E{ \left. \log \frac{q(Y|X)^s}{q(Y)} \right| \set E_k}
   \label{eq:GMI-2-part0}
\end{align}
and inserting~\eqref{eq:GMI-qyx-part} and~\eqref{eq:GMI-qy-part} we have the following lemma.

\begin{lemma} \label{lemma:GMI-part}
The GMI \eqref{eq:GMI} for the channel $p(y|x)$, $s=1$, a CSCG input $X$ with variance $P$, and the auxiliary model \eqref{eq:GMI-qyx-part} is (see~\eqref{eq:GMI-3})
\begin{align}
   I_1(X;Y) & = \sum_{k=1}^K \Pr{\set E_k} \left[ \log\left( 1 + \frac{|h_k|^2 P}{\sigma_k^2} \right) \right. \nonumber \\
   & \quad \left. + \frac{\E{|Y|^2 | \set{E}_k}}{\sigma_k^2 + |h_k|^2P} - \frac{\E{|Y - h_k X|^2 | \set{E}_k}}{\sigma_k^2} \right] .
   \label{eq:GMI-2-part}
\end{align}
\end{lemma}

\begin{remark} \label{remark:GMI-part1}
$K$-partitioning formally includes \eqref{eq:GMI-2-avg} as a special case by including $S_R$ as part of the receiver's ``overall'' channel output $\tilde Y = [Y,S_R]$. For example, one can partition $\tilde{\set Y}$ as $\{\tilde {\set Y}_{s_R}: s_R \in \set S_R\}$ where $\tilde{\set Y}_{s_R} = \set Y \times \{s_R\}$.
\end{remark}

\begin{remark} \label{remark:GMI-part4}
The models~\eqref{eq:I-LB-q3} and~\eqref{eq:GMI-qyx-part} suggest building receivers based on adaptive Gaussian statistics. However, we are motivated to introduce \eqref{eq:GMI-qyx-part} to prove capacity scaling results. For this purpose, we will use $K=2$ with the partition
\begin{align}
   \set E_1=\{|Y|^2 < t_R\}, \quad \set E_2=\{|Y|^2 \ge t_R\}
   \label{eq:2partition}
\end{align}
and $h_1=0$, $\sigma_1^2=1$. The GMI \eqref{eq:GMI-2-part} thus has only the $k=2$ term and it remains to choose $h_2$, $\sigma_2^2$, and $t_R$.
\end{remark}

\begin{remark} \label{remark:GMI-part6}
One can generalize Lemma~\ref{lemma:GMI-part} and partition $\set X \times \set Y$ rather than $\set Y$ only. However, the $q(y)$ in \eqref{eq:GMI-qy-part}  might not have a CSCG form.
\end{remark}

\begin{remark}
Define $P_k=\E{|X|^2 | \set{E}_k}$ and choose the LMMSE auxiliary models with
\begin{align}
    h_k & = \E{\left. Y X^\text{*} \right| \set E_k} \big/ P_k \label{eq:GMI-part-h} \\
   \sigma_k^2 & = \E{\left. |Y - h_k \,X|^2 \right| \set E_k}
   = \E{\left. |Y|^2 \right| \set E_k} - |h_k|^2 P_k
   \label{eq:GMI-part-sigma2}
\end{align}
for $k=1,\dots,K$. The expression \eqref{eq:GMI-2-part} is then
\begin{align}
   \sum_{k=1}^K \Pr{\set E_k} & \left[ \log\left( 1 + \frac{|h_k|^2 P}{\E{|Y|^2 | \set{E}_k} - |h_k|^2 P_k } \right) \right. \nonumber \\
   & \quad \left. - \; \frac{|h_k|^2 (P-P_k) }{\E{|Y|^2 | \set{E}_k}+|h_k|^2 (P-P_k)} \right] .
   \label{eq:GMI-3-part}
\end{align}
\end{remark}

\begin{remark} \label{remark:GMI-part2}
The LMMSE-based GMI \eqref{eq:GMI-3-part} reduces to the GMI of Proposition~\ref{prop:GMI} by choosing the trivial partition with $K=1$ and $\set Y_1=\set Y$. However, the GMI~\eqref{eq:GMI-3-part} may not be optimal for $K \ge 2$. What can be said is that the phase of $h_k$ in \eqref{eq:GMI-2-part} should be the same as the phase of $\E{YX^\text{*} | \set E_k}$ for all $k$. We thus have $K$ two-dimensional optimization problems, one for each pair $(|h_k|,\sigma_k^2)$, $k=1,\dots,K$.
\end{remark}

\begin{remark} \label{remark:large-K}
Suppose we choose a different auxiliary model for each $Y=y$, i.e., consider $K \rightarrow \infty$. The reverse density GMI uses the auxiliary model~\eqref{eq:I-LB-q2} which gives the RHS of~\eqref{eq:I-LB2}:
\begin{align}
    I_1(X;Y) = \int_{\mathbb C} \, p(y) \log \frac{P}{\Var{X|Y=y}} \, dy .
    \label{eq:I-LB2a}
\end{align}
Instead, the suboptimal \eqref{eq:GMI-3-part} is the complicated expression
\begin{align}
   \int_{\mathbb C} \, p(y) & \left[ \log\left( 1 + \frac{|\E{X|Y=y}|^2 (P / P_y)}{\Var{X|Y=y}} \right) \right. \nonumber \\
   & \; \left. - \; \frac{|\E{X|Y=y}|^2 (P/P_y-1) }{\Var{X|Y=y} + |\E{X|Y=y}|^2 (P/P_y)} \right] \, dy.
   \label{eq:GMI-5-part}
\end{align}
where $P_y = \E{|X|^2|Y=y}$. We show how to compute these GMIs in Appendix~\ref{appendix:c-second-order-statistics}.
\end{remark}

\subsection{Example: On-Off Fading}
\label{subsec:GMI-oof}
Consider the channel $Y=HX+Z$ where $H, X, Z$ are mutually independent, $P_H(0)=P_H(\sqrt{2})=1/2$, and $Z \sim \mathcal{CN}(0,1)$. The channel exhibits particularly simple fading, giving basic insight into more realistic fading models. We consider two basic scenarios: full CSIR and no CSIR.

\subsubsection{Full CSIR}
Suppose $S_R=H$ and
\begin{align}
   q(y|x,h) = p(y|x,h) = \frac{1}{\pi \sigma^2} e^{-|y - h x |^2 /\sigma^2}
   \label{eq:GMI-qyxh-oof}
\end{align}
which corresponds to having \eqref{eq:GMI-parameters-c1}--\eqref{eq:GMI-parameters-c2} as 
\begin{align}
    \tilde h(0)=0, \quad
    \tilde h\left(\sqrt{2}\right)=\sqrt{2}, \quad 
    \tilde \sigma^2(0)= \sigma^2\left(\sqrt{2}\right)=1.
\end{align}
The GMI~\eqref{eq:GMI-2-avg-a} with $X\sim\mathcal{CN}(0,P)$ thus gives the capacity
\begin{align}
   C(P) = \frac{1}{2} \log\left( 1 + 2P \right).
   \label{eq:GMI-2-avg-oof}
\end{align}
The wideband values \eqref{eq:wideband} are
\begin{align}
   \left.\frac{E_b}{N_0}\right|_{\text{min}} = \log 2, \quad S = 1 .
   \label{eq:wideband-fullCSIR-noCSIT-oof}
\end{align}
Compared with \eqref{eq:wideband-fullCSIR-noCSIT-awgn},
the minimal $E_b/N_0$ is the same as without fading, namely $-1.59$ dB. However, fading reduces the capacity slope $S$; see the dashed curve in Fig.~\ref{fig:oof1}.

\begin{figure}[t]
      \centering
      \includegraphics[width=0.95\columnwidth]{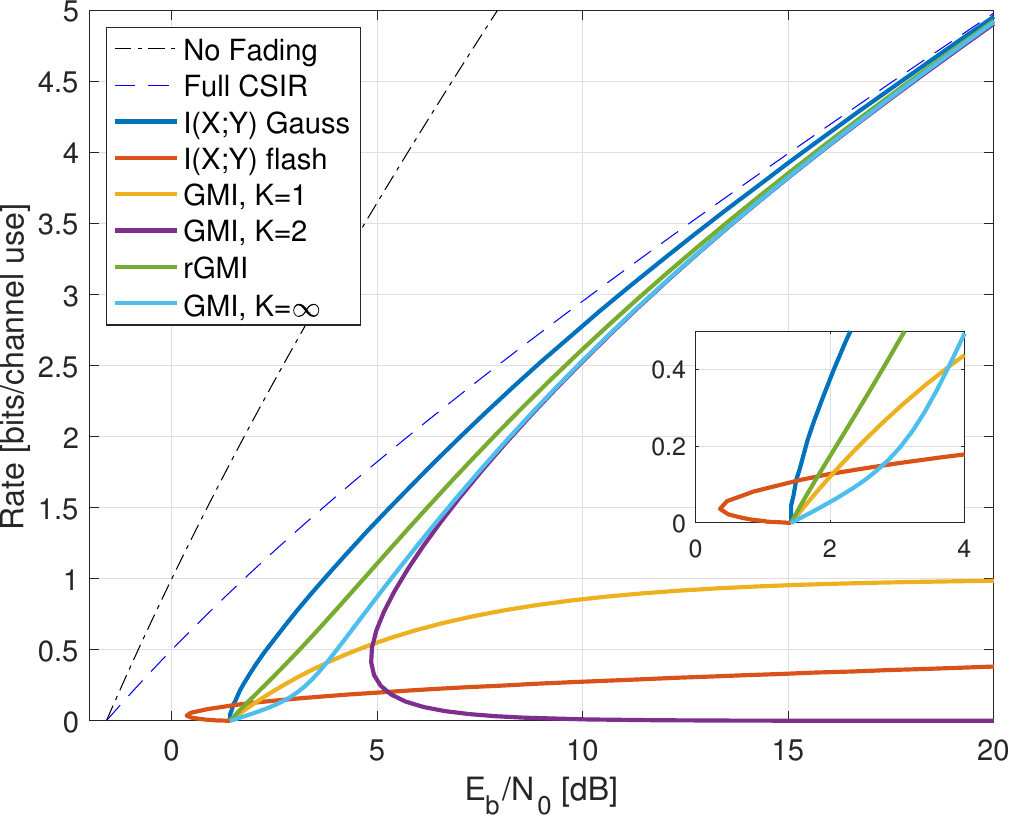}
      \caption{Rates for on-off fading with $S_R=0$. The curve "Full CSIR" refers to $S_R=H$ and is a capacity upper bound. Flash signaling uses $p=0.05$; the GMI for the $K=2$ partition uses the threshold $t_R=P^{0.4}+3$.}
      \label{fig:oof1}
\end{figure}

\subsubsection{No CSIR}
Suppose $S_R=0$ and $X\sim\mathcal{CN}(0,P)$ and consider the densities
\begin{align}
   & p(y|x) = \frac{e^{-|y |^2}}{2\pi} + \frac{e^{-|y - \sqrt{2} x |^2}}{2\pi}
   \label{eq:GMI-example1} \\
   & p(y) = \frac{e^{-|y |^2}}{2\pi} + \frac{e^{-|y|^2/(1+2P)}}{2\pi(1 + 2P)} .
     \label{eq:GMI-example2}
\end{align}
The mutual information can be computed by numerical integration or by Monte Carlo integration:
\begin{align}
   I(X;Y) \approx \frac{1}{N} \sum_{i=1}^N \log \frac{p_{Y|X}(y_i|x_i)}{p_Y(y_i)}
   \label{eq:MI-MC}
\end{align}
where the RHS of \eqref{eq:MI-MC} converges to $I(X;Y)$ for long strings $x^N,y^N$ sampled from $p(x,y)$. The results for $X\sim\mathcal{CN}(0,P)$ are shown in Fig.~\ref{fig:oof1} as the curve labeled ``$I(X;Y)$ Gauss''.

Next, Proposition~\ref{prop:GMI} gives $h=1/\sqrt{2}$, $\sigma^2=1+P/2$, and
\begin{align}
   I_1(X;Y) = \log\left( 1 + \frac{P}{2+P} \right).
   \label{eq:GMI-2-example}
\end{align}
The wideband values \eqref{eq:wideband} are
\begin{align}
   \left.\frac{E_b}{N_0}\right|_{\text{min}} = \log 4, \quad S = 2/3
   \label{eq:wideband-fullCSIR-noCSIT-oof2}
\end{align}
so the minimal $E_b/N_0$ is 1.42 dB and the capacity slope $S$ has decreased further. Moreover, the rate saturates at large SNR at 1 bit per channel use.

The ``$I(X;Y)$ Gauss'' curve in Fig.~\ref{fig:oof1} suggests that the no-CSIR capacity approaches the full-CSIR capacity for large SNR. To prove this, consider the $K=2$ partition specified in Remark~\ref{remark:GMI-part4} with $h_1=0$, $h_2=\sqrt{2}$, and $\sigma_2^2=1$. Since we are not using LMMSE auxiliary models, we must compute the GMI using the general expression~\eqref{eq:GMI-2-part}, which is
\begin{align}
   I_1(X;Y) = & \Pr{\set E_2} \bigg[ \log( 1 + 2 P) \nonumber \\
   & \left. + \frac{\E{|Y|^2 | \set{E}_2}}{1 + 2 P} - \E{\left|Y - \sqrt{2} X\right|^2 | \set{E}_2} \right] .
   \label{eq:GMI-2-part-2}
\end{align}
In Appendix~\ref{appendix:gmi-oof}, we show that choosing $t_R = P^{\lambda_R}+b$ where $0<\lambda_R<1$ and $b$ is a real constant makes all terms behave as desired as $P$ increases:
\begin{align}
\begin{array}{l}
   \Pr{\set E_2} \rightarrow 1/2  \vspace{0.1cm} \\
   \E{\left. |Y|^2 \right| \set E_2}/(1+2P) \rightarrow 1  \vspace{0.1cm} \\
   \E{\left. \left|Y - \sqrt{2} X \right|^2 \right| \set E_2} \rightarrow 1.
\end{array}
\label{eq:gmi-scaling-terms}
\end{align}
The GMI~\eqref{eq:GMI-2-part-2} of Lemma~\ref{lemma:GMI-part} thus gives the maximal value \eqref{eq:GMI-2-avg-oof} for large $P$:
\begin{align}
   \lim_{P \rightarrow \infty} \left[ \frac{1}{2} \log(1+2 P) - I_1(X;Y) \right] = 0 .
   \label{eq:GMI-2-example-limit}
\end{align}
Fig.~\ref{fig:oof1} shows the behavior of $I_1(X;Y)$ for $K=2$, $\lambda_R=0.4$, and $b=3$. Effectively, at large SNR, the receiver can estimate $H$ accurately, and one approaches the full-CSIR capacity.

\begin{remark} \label{remark:GMI-part0}
For on-off fading, one may compute $I(X;Y)$ directly and use the densities \eqref{eq:GMI-example1}--\eqref{eq:GMI-example2} to decode. Nevertheless, the partitioning of Lemma~\ref{lemma:GMI-part} helps prove the capacity scaling~\eqref{eq:GMI-2-example-limit}.
\end{remark}

Consider next the reverse density GMI~\eqref{eq:I-LB2a} and the forward model GMI~\eqref{eq:GMI-5-part}. Appendix~\ref{appendix:csos-1} shows how to compute $\E{X|Y=y}$, $\E{|X|^2\big|Y=y}$, and $\Var{X|Y=y}$, and Fig.~\ref{fig:oof1} plots the GMIs as the curves labeled ``rGMI'' and ``GMI, K=$\infty$'', respectively. The rGMI curve gives the best possible rates for AWGN auxiliary models, as shown in Sec.~\ref{subsec:refined-models}. The results also show that the large-$K$ GMI \eqref{eq:GMI-5-part} is worse than the $K=1$ GMI at low SNR but better than the $K=2$ GMI of Remark~\ref{remark:GMI-part4}. See Fig.~\ref{fig:oof3} below for similar results.

Finally, the curve labeled ``$I(X;Y)$ Gauss'' in Fig.~\ref{fig:oof1} suggests that the minimal $E_b/N_0$ is 1.42 dB even for the capacity-achieving distribution. However, we know from~\cite[Thm.~1]{Verdu02} that flash signaling \eqref{eq:flash-Gauss} can approach the minimal $E_b/N_0$ of $-1.59$ dB. For example, the flash rates $I(X;Y)$ with $p=0.05$ are plotted in Fig.~\ref{fig:oof1}. Unfortunately, the wideband slope is $S=0$ ~\cite[Thm.~17]{Verdu02}, and one requires very large flash powers (very small $p$) to approach $-1.59$ dB.

\begin{remark} \label{remark:caire2}
As stated in Remark~\ref{remark:caire}, the paper~\cite{Caire-WCOM18} (see also~\cite{Biglieri-Proakis-Shamai-IT98,Mezghani-Nossek-ISIT02}) derives two capacity lower bounds. These bounds are the same for our problem, and they are derived using the following steps (see~\cite[Lemmas~3 and 4]{Caire-WCOM18}):
\begin{align}
    I(X;Y) & = I(X,S_H;Y) - I(S_H;Y|X) \nonumber \\
    & \ge I(X;Y|S_H) - I(S_H;Y|X). \label{eq:caire-steps}
\end{align}
Now consider $Y = H X + Z$ where $H,X,Z$ are mutually independent, $S_H=H$, $\Var{Z}=1$, and $X\sim\mathcal{CN}(0,P)$. We have
\begin{align}
    I(X;Y|H) & \ge \E{\log( 1 +|H|^2 P )}
    \label{eq:caire-identity} \\
    I(H;Y|X) & = h(Y|X) - h(Z) \nonumber \\
    & \le \log\left(\pi e (1 + \Var{H} P)\right) - h(Z)
    \label{eq:caire-bound2}
\end{align}
where \eqref{eq:caire-identity}--\eqref{eq:caire-bound2} follow by~\eqref{eq:h-UB}, in the latter case with the roles of $X$ and $Y$ reversed.
The bound~\eqref{eq:caire-bound2} works well if $\Var{H}$ is small, as for massive MIMO with ``channel hardening''. However, for our on-off fading model, the bound \eqref{eq:caire-steps} is
\begin{align}
  I(X;Y) & \ge \E{\log\left( 1 +|H|^2 P \right)} - \log(1 + \Var{H} P) \nonumber \\
  & = \frac{1}{2} \log(1 + 2P) - \log(1 + P/2)
  \label{eq:caire}
\end{align}
which is worse than the $K=1$ and $K=\infty$ GMIs and is not shown in Fig.~\ref{fig:oof1}.
\end{remark}

\section{Channels with CSIT}
\label{sec:CSIT}
This section studies Shannon's channel with side information, or state, known causally at the transmitter~\cite{Shannon58,Keshet-Steinberg-Merhav-08}. We begin by treating general channels and then focus mainly on complex-alphabet channels. The capacity expression has a random variable $A$ that is either a list (for discrete-alphabet states) or a function (for continuous-alphabet states). We refer to $A$ as an adaptive symbol of an adaptive codeword.

\subsection{Model}
\label{subsec:CSIT-model}
The problem is specified by the functional dependence graph (FDG) in Fig.~\ref{fig:Shannon-iBM}. The model has a message $M$, a CSIT string $S_T^n$, and a noise string $Z^n$. The variables $M$, $S_T^n$, $Z^n$ are mutually statistically independent, and $S_T^n$ and $Z^n$ are strings of i.i.d.\ random variables with the same distributions as $S_T$ and $Z$, respectively. $S_T^n$ is available \emph{causally}  at the transmitter, i.e., the channel input $X_i$, $i=1,\ldots,n$, is a function of $M$ and the sub-string $S_T^i$. The receiver sees the channel outputs
\begin{align}
   Y_{i} = f(X_i,S_{Ti},Z_i) 
   \label{eq:Shannon-model}
\end{align}
for some function $f(.)$ and $i=1,2,\ldots,n$.

\begin{figure}[t]
      \centering
      \includegraphics[width=0.9\columnwidth]{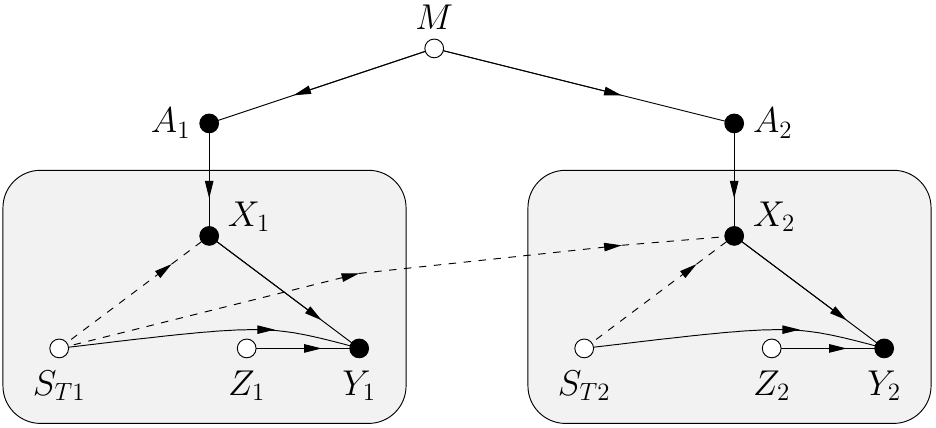}
      \caption{FDG for $n=2$ uses of a channel with CSIT. Open nodes represent statistically independent random variables, and filled nodes represent random variables that are functions of their parent variables. Dashed lines represent the CSIT
      influence on $X^n$.}
      \label{fig:Shannon-iBM}
\end{figure}

Each $A_i$ represents a list of possible choices of $X_i$ at time $i$. For example, suppose $S_T$ has alphabet $\set{S}_T=\{0,1,\ldots,\nu-1\}$ and define the adaptive symbol
\begin{align*}
   A = \big( X(0),\ldots,X(\nu-1) \big)
\end{align*}
whose entries have alphabet $\set{X}$. Here $S_T=s_T$ means that $X(s_T)$ is transmitted, i.e., we have $X=X(S_T)$. If $S_T$ has a continuous alphabet, we make $A$ a function rather than a list, and we may again write $X=X(S_T)$. Some authors therefore write $A$ as $X(.)$.\footnote{Shannon in~\cite{Shannon58} denoted our $A$ and $X$ as the respective $X$ and $x$.}

\begin{remark} \label{remark:conventional-codebook}
The conventional choice for $A$ if $\set X = \mathbb C$ is
\begin{align}
   A = \left( \sqrt{P(0)}\, e^{j\phi(0)}, \ldots, \sqrt{P(\nu-1)}\, e^{j\phi(\nu-1)} \right) \cdot U
   \label{eq:adaptive-codeword-conventional}
\end{align}
where $U$ has $\E{|U|^2}=1$, $P(s_T)=\E{|X(s_T)|^2}$, and $\phi(s_T)$ is a phase shift. The interpretation is that $U$ represents the symbol of a conventional codebook without CSIT, and these symbols are scaled and rotated. In other words, one separates the message-carrying $U$ from an adaptation due to $S_T$ via
\begin{align}
   X = \sqrt{P(S_T)} \, e^{j\phi(S_T)} \, U .
   \label{eq:adaptive-codeword-conventional2}
\end{align}
\end{remark}

\begin{remark}
One may define the channel by the functional relation~\eqref{eq:Shannon-model}, by $p(y|a)$, or by $p(y|x,s_T)$; see Shannon's emphasis in~\cite[Theorem]{Shannon58} (cf.~\cite[Remark~3]{Kramer14}). We generally prefer to use $p(y|a)$ since we interpret $A$ as a channel input.
\end{remark}

\begin{remark}
One can add feedback and let $X_i$ be a function of $(M,S_T^i,Y^{i-1})$, but feedback does not increase the capacity if the state and noise processes are memoryless~\cite[Sec.~V]{Kramer14}.
\end{remark}

\begin{remark} \label{remark:basic-block-fading}
The model \eqref{eq:Shannon-model} permits block fading and MIMO transmission by choosing $X_i$ and $Y_i$ as vectors~\cite{McElieceStark84,Lau-Liu-Chen-IT04}.
\end{remark}

\subsection{Capacity}
\label{subsec:CSIT-memoryless-capacity}
The capacity of the model under study is (see~\cite{Shannon58})
\begin{align} 
   C = \max_{A} \, I(A ; Y )
   \label{eq:Shannon-cap}
\end{align}
where $A-[S_T,X]-Y$ forms a Markov chain. One may limit attention to $A$ with cardinality $|\set{A}|$ satisfying (see~\cite[Eq.~(56)]{Kramer14},~\cite{Shannon57},~\cite[Thm.~1]{Farmanbar09})
\begin{align}
  |\set{A}| \le \min\left(|\set{Y}|,1+|\set{S}_T|(|\set{X}|-1)\right).
  \label{eq:cardinality}
\end{align}
As usual, for the cost function $c(x,y)$ and the average block cost constraint
\begin{align} 
   \frac{1}{n} \sum_{i=1}^n \E{c(X_i,Y_i)} \le P \label{eq:cost}
\end{align}
the unconstrained maximization in \eqref{eq:Shannon-cap} becomes a constrained maximization over the $A$ for which $\E{c(X,Y)}\le P$. Also, a simple upper bound on the capacity is
\begin{align} 
   C(P) & \le \max_{A:\,\E{c(X,Y)}\le P} \,
   I(A ; Y, S_T) \nonumber \\
   & \overset{(a)}{=} \max_{X(S_T):\,\E{c(X(S_T),Y)}\le P} \,
   I(X ; Y | S_T)
\end{align}
where step $(a)$ follows by the independence of $A$ and $S_T$. This bound is tight if the receiver knows $S_T$.

\begin{remark} \label{remark:MLC-MSD}
The chain rule for mutual information gives
\begin{align}
    I(A;Y) & = I\left( X(0) \dots X(\nu-1) ; Y \right) \label{eq:MISO-expression} \\
    & = \sum_{s_T=0}^{\nu-1} I\left( X(s_T) ; Y | X(0),\dots,X(s_T-1) \right).
    \label{eq:chain-rule}
\end{align}
The RHS of~\eqref{eq:MISO-expression} suggests treating the channel as a multi-input, single-output (MISO) channel, and the expression~\eqref{eq:chain-rule} suggests using multi-level coding with multi-stage decoding~\cite{Wachsmann-Huber-IT99}. For example, one may use polar coded modulation~\cite{Stolte02,Arikan:09,Seidl-IT13} with Honda-Yamamoto shaping~\cite{Honda-Yamamoto-IT13,Runge-ISIT22}.
\end{remark}

\begin{remark}
For $\set X = \mathbb C$ and the conventional adaptive symbol~\eqref{eq:adaptive-codeword-conventional}, we compute $I(A;Y)=I(U;Y)$ and
\begin{align}
   C(P) = \max_{P(S_T),\phi(S_T):\,\E{c(X(S_T),Y)}\le P} \,
   I(U ; Y).
   \label{eq:Shannon-cap2}
\end{align}
\end{remark}

\subsection{Structure of the Optimal Input Distribution}
\label{subsec:CSIT-opt}
Let $\set A$ be the alphabet of $A$ and let $\set X={\mathbb C}$, i.e., we have $\set{A}={\mathbb C}^\nu$ for discrete $S_T$. Consider the expansions
\begin{align}
    & p(y | a) = \sum_{s_T} P_{S_T}(s_T) \, p( y | x(s_T), s_T) \\
    & p(y) = \int_\set{A} p(a) \, p(y | a) \, da \nonumber \\
    & =  \sum_{s_T} P_{S_T}(s_T) \int_{\mathbb C} p(x(s_T)) \, p( y | x(s_T), s_T) \, dx(s_T).
    \label{eq:py}
\end{align}
Observe that $p(y)$, and hence $h(Y)$, depends only on the marginals $p(x(s_T))$ of $A$; see~\cite[Sec.~III]{Farmanbar09}. So define the set of densities having the same marginals as $A$:
\begin{align*}
   \set P(A) = \left\{ p(\tilde a): p(\tilde x(s_T)) = p(x(s_T)) \text{ for all } s_T \in \set S_T \right\}.
\end{align*}
This set is convex, since for any $p^{(1)}(a),p^{(2)}(a) \in \set P(A)$ and $0\le\lambda\le1$ we have
\begin{align}
   \lambda p^{(1)}(a) + (1-\lambda) p^{(2)}(a) \in \set P(A).
\end{align}
Moreover, for fixed $p(y)$, the expression $I(A;Y)$ is a convex function of $p(a|y)$, and $p(a|y)=p(a) p(y|a)/p(y)$ is a linear function of $p(a)$. Maximizing $I(A;Y)$ over $\set P(A)$ is thus the same as minimizing the concave function $h(Y|A)$ over the convex set $\set P(A)$. An optimal $p(a)$ is thus an extreme of $\set P(A)$. Some properties of such extremes are developed in~\cite{Parthasarathy-PMS07,Nadkarni-A08}.

For example, consider $|\set S_T|=2$ and $\set X = \set S_T = \{0,1\}$, for which \eqref{eq:cardinality} states that at most $|\set{A}|=3$ adaptive symbols need have positive probability (and at most $|\set{A}|=2$ adaptive symbols if $|\set Y|=2$). Suppose the marginals have $P_{X(0)}(0)=1/2$, $P_{X(1)}(0)=3/4$ and consider the matrix notation
\begin{align*}
  P_A = \begin{bmatrix} P_A(0,0) & P_A(0,1) \\ P_A(1,0)& P_A(1,1) \end{bmatrix}
\end{align*}
where we write $P_A(x_1,x_2)$ for $P_A([x_1,x_2])$. The optimal $P_A$ must then be one of the two extremes
\begin{align}
    \begin{matrix} 
   P_A = \begin{bmatrix} 1/2 & 0 \\ 1/4 & 1/4 \end{bmatrix}, \quad 
   P_A = \begin{bmatrix} 1/4 & 1/4 \\ 1/2 & 0 \end{bmatrix}.
    \end{matrix} 
\end{align}
For the first $P_A$, the codebook has the property that if $X(0)=0$ then $X(1)=0$ while if $X(0)=1$ then $X(1)$ is uniformly distributed over $\set X=\{0,1\}$.

Next, consider $|\set S_T|=2$ and marginals $P_{X(0)}$, $P_{X(1)}$ that are uniform over $\set X=\{0,1,\dots,|\set{X}|-1\}$. This case was treated in detail in~\cite[Sec.~VI.A]{Farmanbar09}, see also~\cite{Farmanbar10}, and we provide a different perspective. A classic theorem of Birkhoff~\cite{Birkhoff-46} ensures that the extremes of $\set P(A)$ are the $|\set X|!$ distributions $P_A$ for which the $|\set X| \times |\set X|$ matrix
\begin{align*}
  P_A = \begin{bmatrix}
  P_A(0,0) & \dots & P_A(0,|\set X|-1) \\
  \vdots & \ddots & \vdots \\
  P_A(|\set X|-1,0) & \dots & P_A(|\set X|-1,|\set X|-1) \end{bmatrix}.
\end{align*}
is a permutation matrix multiplied by $1/|\set X|$. For example, for $|\set X|=2$ we have the two extremes
\begin{align}
    \begin{matrix} 
   P_A = \frac{1}{2} \begin{bmatrix} 1 & 0 \\ 0 & 1 \end{bmatrix}, \quad 
   P_A = \frac{1}{2} \begin{bmatrix} 0 & 1 \\ 1 & 0 \end{bmatrix}.
    \end{matrix} 
    \label{eq:extreme2}
\end{align}
The permutation property means that $X(s_T)$ is a function of $X(0)$, i.e., the encoding simplifies to a conventional codebook as in Remark~\ref{remark:conventional-codebook} with uniformly-distributed $U$ and a permutation $\pi_{s_T}(.)$ indexed by $s_T$ such that $X(S_T)=\pi_{S_T}(U)$. For example, for the first $P_A$ in \eqref{eq:extreme2} we may choose $X(S_T)=U$, which is independent of $S_T$. On the other hand, for the second $P_A$ in \eqref{eq:extreme2} we may choose $X(S_T)=U\oplus S_T$ where $\oplus$ denotes addition modulo-2.

For $|\set S_T|>2$, the geometry of $\set P(A)$ is more complicated; see~\cite[Sec.~VI.B]{Farmanbar09}. For example, consider $\set X=\{0,1\}$ and suppose the marginals $P_{X(s_T)}$, $s_T \in \set S_T$, are all uniform. Then the extremes include $P_A$ related to linear codes and their cosets, e.g., two extremes for $|\set S_T|=3$ are related to the repetition code and single parity check code:
\begin{align*}
    \begin{array}{l}
   P_A(a) = 1/2, \; a \in \{ [0,0,0], [1,1,1] \} \\
   P_A(a) = 1/4, \; a \in \{ [0,0,0], [0,1,1], [1,0,1], [1,1,0] \}.
    \end{array}
\end{align*}
This observation motivates concatenated coding, where the message is first encoded by an outer encoder followed by an inner code that is the coset of a linear code. The transmitter then sends the entries at position $S_T$ of the inner codewords, which are vectors of dimension $|\set S_T|$. We do not know if there are channels for which such codes are helpful.

\subsection{Generalized Mutual Information}
\label{subsec:CSIT-GMI}
Consider the vector channel $p(\ul y | \ul x)$ with input set $\set X = \mathbb C^M$ and output set $\set Y = \mathbb C^N$. The GMI for adaptive symbols is $\max_{s\ge0} I_s(A;\ul Y)$ where
\begin{align}
   I_s(A;\ul Y) = \E{ \log \frac{q(\ul Y|A)^s}{q(\ul Y)} }
   \label{eq:AGMI}
\end{align}
and the expectation is with respect to $p(a,\ul y)$. Suppose the auxiliary model is $q(\ul y|a)$ and define
\begin{align}
   q(\ul y) = \int_\set{A} p(a) q(\ul y|a)^s \, da .
   \label{eq:AGMI-qy-scalar}
\end{align}
The GMI again provides a lower bound on the mutual information since (cf.~\eqref{eq:GMI-identity})
\begin{align}
   I_s(A;\ul Y) = I(A;\ul Y) - D\left( \left. p_{A,\ul Y} \right\| p_{\ul Y} \, q_{A|\ul Y} \right)
   \label{eq:AGMI-identity}
\end{align}
where $q(a|\ul y) = p(a) q(\ul y|a)^s/q(\ul y)$ is a reverse channel density.

We next study reverse and forward models as in Sec.~\ref{subsec:intro-aux} and Sec.~\ref{subsec:refined-models}. Suppose the entries $\ul X(s_T)$ of $A$ are jointly CSCG. 

\subsubsection{Reverse Model}
We write $\ul A$ when we consider $A$ to be a column vector that stacks the $\ul X(s_T)$. Consider the following reverse density $q(\ul a |\ul y)$ motivated by~\eqref{eq:I-LB-q2r}:
\begin{align}
    \frac{\exp\left(- (\ul a - \E{\ul A| \ul Y = \ul y})^\dag {\bf Q}_{\ul A | \ul Y = \ul y}^{-1} (\ul a - \E{\ul A| \ul Y = \ul y}) \right)}{\pi^{\nu M} \det {\bf Q}_{\ul A | \ul Y = \ul y}}.
    \label{eq:I-LB-q2-v}
\end{align}
A corresponding forward model is $q\big(\ul y|a\big)=q\big(a|\ul y\big)/p(a)$ and the GMI with $s=1$ becomes (cf.~\eqref{eq:I-LB-vector2})
\begin{align}
   I_1(A;\ul Y) = \E{ \log \frac{\det {\bf Q}_{\ul A}}{\det {\bf Q}_{\ul A | \ul Y}} }.
   \label{eq:AGMI-r}
\end{align}
To simplify, one may focus on adaptive symbols as in~\eqref{eq:adaptive-codeword-conventional2}:
\begin{align}
   \ul X = {\bf Q}_{\ul X(S_T)}^{1/2} \cdot \ul U
   \label{eq:adaptive-codeword-conventional2-v}
\end{align}
where $\ul U \sim \mathcal{CN}(\ul 0, {\bf I})$ and the ${\bf Q}_{\ul X(s_T)}$ are covariance matrices. We thus have $I(A;\ul Y) = I(\ul U ; \ul Y)$ (cf.~\eqref{eq:Shannon-cap2}) and using~\eqref{eq:I-LB-q2-v} but with $\ul A$ replaced with $\ul U$ we obtain 
\begin{align}
   I_1(A;\ul Y) = \E{ - \log \det {\bf Q}_{\ul U | \ul Y} }.
   \label{eq:AGMI-r2}
\end{align}

\subsubsection{Forward Model}
Perhaps the simplest forward model is $q(\ul y | a) = p(\ul y | \ul x(s_T))$ for some fixed value $s_T \in \set{S}_T$. One may interpret this model as having the receiver assume that $S_T=s_T$. A natural generalization of this idea is as follows: define the auxiliary vector
\begin{align}
   \ul{\bar X} = \sum_{s_T} {\bf W}(s_T) \, \ul X(s_T)
   \label{eq:xbar}
\end{align}
where the ${\bf W}(s_T)$ are $M \times M$ complex matrices, i.e., $\ul{\bar X}$ is a linear function of the entries of $A=[\ul X(s_T): s_T \in\set{S}_T]$. For example, the matrices might be chosen based on $P_{S_T}(.)$. However, observe that $\ul{\bar X}$ is \emph{independent} of $S_T$. Now define the auxiliary model
\begin{align*}
   q(\ul y | a) = q(\ul{y} | \ul{\bar x} )
\end{align*}
where we abuse notation by using the same $q(.)$. The expression~\eqref{eq:AGMI-qy-scalar} becomes
\begin{align}
  q(\ul y) &= \int_\set{A} p(a) \, q(\ul y | a)^s \, da
  = \int_{\mathbb C} p(\ul{\bar x}) \, q(\ul y | \ul{\bar x})^s \, d\ul{\bar x}.
  \label{eq:AGMI-qy}
\end{align}

\begin{remark}
We often consider $\set S_T$ to be a discrete set, but for CSCG channels we also consider $\set S_T=\mathbb C$ so that the sum over $\set S_T$ in~\eqref{eq:xbar} is replaced by an integral over $\mathbb C$.
\end{remark}

We now specialize further by choosing the auxiliary channel
$\ul Y_a = {\bf H}\, \ul{\bar X} + \ul Z$
where $\bf H$ is an $N\times M$ complex matrix, $\ul Z$ is an $N$-dimensional CSCG vector that is independent of $\ul{\bar X}$ and has invertible covariance matrix ${\bf Q}_{\ul Z}$, and $\bf H$ and ${\bf Q}_{\ul Z}$ are to be optimized. Further choose $A=[\ul{X}(s_T): s_T \in \set{S}_T]$ whose entries are jointly CSCG
with correlation matrices
\begin{align*}
  {\bf R}(s_{T1},s_{T2}) = \E{\ul X(s_{T1}) \ul X(s_{T2})^\dag}.
\end{align*}
Since $\ul{\bar X}$ in~\eqref{eq:xbar} is independent of $S_T$, we have
\begin{align}
   q(\ul y | a) = \frac{\exp
   \left(- \left( \ul y - {\bf H}\, \ul{\bar x} \right)^\dag {\bf Q}_{\ul{Z}}^{-1} \left( \ul y - {\bf H}\, \ul{\bar x} \right) \right)}
   {\pi^N \det {\bf Q}_{\ul Z}} .
   \label{eq:AGMI-qya}
\end{align}
Moreover, $\ul{\bar X}$ is CSCG so $q(\ul y)$ in \eqref{eq:AGMI-qy} is
\begin{align*}
   \frac{\pi^N \det\left({\bf Q}_{\ul Z} /s \right)}{\left(\pi^N \det {\bf Q}_{\ul Z}\right)^s} \cdot
   \frac{\exp\left( - {\ul y}^\dag \left( {\bf Q}_{\ul Z}/s + {\bf H} {\bf Q}_{\ul{\bar X}} {\bf H}^\dag \right)^{-1} \ul y  \right)}
   {\pi^N \det \left( {\bf Q}_{\ul Z}/s + {\bf H} {\bf Q}_{\ul{\bar X}} {\bf H}^\dag \right) }
\end{align*}
where
\begin{align*}
  {\bf Q}_{\ul{\bar X}}
  & =  \sum_{s_{T1},s_{T_2}} {\bf W}(s_{T1})  {\bf R}(s_{T1},s_{T2}) {\bf W}(s_{T2})^\dag.
\end{align*}
We have the following generalization of Proposition~\ref{prop:GMI}.

\begin{lemma} \label{lemma:AGMI}
The maximum GMI \eqref{eq:AGMI} for the channel $p(\ul{y} | a)$, an adaptive vector $A=[\ul X(s_T): s_T \in \set{S}_T]$ that has jointly CSCG entries, an $\ul{\bar X}$ as in \eqref{eq:xbar} with ${\bf Q}_{\ul{\bar X}}\succ {\bf 0}$, and the auxiliary model \eqref{eq:AGMI-qya} with ${\bf Q}_{\ul Z} 
\succ {\bf 0}$ is
\begin{align}
   & I_1(A;\ul Y) = \log\det \left( {\bf I} + {\bf Q}_{\ul{\Zt}}^{-1} \, \tilde {\bf H}
   {\bf Q}_{\ul{\bar X}} \tilde {\bf H}^\dag \right)
   \label{eq:AGMI-3}
\end{align}
where (cf.~\eqref{eq:I-LB-vector-parameters})
\begin{align}
   \tilde {\bf H} & = \E{\ul Y \, \ul {\bar X}^\dag} {\bf Q}_{\ul{\bar X}}^{-1} \label{eq:AGMI-parameters-2a} \\
   {\bf Q}_{\ul{\Zt}} & = {\bf Q}_{\ul Y} - \tilde {\bf H} {\bf Q}_{\ul{\bar X}} \tilde {\bf H}^\dag.
   \label{eq:AGMI-parameters-2b}
\end{align}
The expectation is with respect to the actual channel with joint distribution/density $p(a,\ul y)$.
\end{lemma}
\begin{proof}
See Appendix~\ref{appendix:proof-lemma-AGMI}.
\end{proof}

\begin{remark}
Since $\ul{\bar X}$ is a function of $A$, the rate~\eqref{eq:AGMI-3} can alternatively be derived by using $I(A ; \ul{Y}) \ge I(\ul {\bar X} ; \ul{Y})$ and applying the bound \eqref{eq:I-LB-vector} with $\ul X$ replaced with $\ul {\bar X}$.
\end{remark}

\begin{remark}
The estimate $\tilde {\bf H}$ is the MMSE estimate of $\bf H$:
\begin{align}
  \tilde {\bf H} = \arg \min_{\bf H} \E{\| \ul Y - {\bf H} \ul{\bar X} \|^2}
  \label{eq:h-MMSE-est2}
\end{align}
and ${\bf Q}_{\tilde{\ul Z}}$ is the resulting covariance matrix of the error. To see this, expand (cf.~\eqref{eq:MMSE-proof})
\begin{align}
  & \E{\| \ul Y - {\bf H} \ul{\bar X} \|^2} = \E{\| (\ul Y - \tilde {\bf H} \ul{\bar X}) + (\tilde {\bf H} - {\bf H}) \ul{\bar X} \|^2} \nonumber \\
  & = \E{\| \ul Y - \tilde {\bf H} \ul{\bar X} \|^2} + \tr{ (\tilde {\bf H} - {\bf H}) {\bf Q}_{\ul{\bar X}} (\tilde {\bf H} - {\bf H})^\dagger }
\end{align}
where the final step follows by the definition of $\tilde {\bf H}$ in \eqref{eq:AGMI-parameters-2a}.
\end{remark}

\begin{remark} \label{remark:other-H}
Suppose ${\bf H}$ is an estimate other than \eqref{eq:h-MMSE-est2}. Generalizing~\eqref{eq:GMI-parameters-sigma2}, if ${\bf Q}_{\ul Y} \succ {\bf Q}_{\ul{\bar Z}}$ we may choose
\begin{align}
   {\bf Q}_{\ul{Z}}\big/s & = \left( {\bf H} {\bf Q}_{\ul{\bar X}} {\bf H}^\dag \right)^{1/2}
   \left( {\bf Q}_{\ul Y}
    - {\bf Q}_{\ul{\bar Z}} \right)^{-1/2} {\bf Q}_{\ul{\bar Z}} \nonumber \\
   & \quad \cdot \left( {\bf Q}_{\ul Y} 
    - {\bf Q}_{\ul{\bar Z}} \right)^{-1/2}
   \left( {\bf H} {\bf Q}_{\ul{\bar X}} {\bf H}^\dag \right)^{1/2} 
   \label{eq:AGMI-matrix-parameters-part}
\end{align}
where
\begin{align}
   {\bf Q}_{\ul{\bar Z}} & = \E{\left( \ul Y - {\bf H} \ul{\bar X} \right) \left( \ul Y - {\bf H} \ul{\bar X} \right)^\dag }.
   \label{eq:AGMI-matrix-parameters-part2}
\end{align}
Appendix~\ref{appendix:proof-lemma-AGMI} shows that
\eqref{eq:AGMI} then simplifies to (cf.~\eqref{eq:GMI-3h})
\begin{align}
   I_s(A;\ul Y) = \log\det \left( 
   {\bf Q}_{\ul{\bar Z}}^{-1} {\bf Q}_{\ul Y}
   \right) .
   \label{eq:AGMI-matrix-3-part}
\end{align}
\end{remark}

\begin{remark} \label{remark:xbar-scaling}
The GMI \eqref{eq:AGMI-3} does not depend on the scaling of $\ul{\bar X}$ since this is absorbed in $\tilde {\bf H}$. For example, one can choose the weighting matrices in \eqref{eq:xbar} so that $\E{\|\ul{\bar X}\|^2}=P$.
\end{remark}

\subsection{Optimal Codebooks for CSCG Forward Models}
\label{subsec:CSIT-codebook-structure-GMI}
The following Lemma maximizes the GMI for scalar channels and $A$ with CSCG entries without requiring $A$ to have the form~\eqref{eq:adaptive-codeword-conventional2}. Nevertheless, this form is optimal, and we refer to~\cite[p.~2013]{Caire-Shamai-IT99} and Sec.~\ref{subsec:Full-CSIR-part-CSIT} for similar results. In the following, let $U(s_T)\sim\mathcal{CN}(0,1)$ for all $s_T$.

\begin{lemma} \label{lemma:AGMI-max}
The maximum GMI \eqref{eq:AGMI} for the channel $p(y|a)$, an adaptive symbol $A$ with jointly CSCG entries, the forward model \eqref{eq:AGMI-qya}, and with fixed $P(s_T)=\E{|X(s_T)|^2}$ is
\begin{align}
   I_1(A;Y) = \log\left( 1 + \frac{\tilde P}{\E{|Y|^2} - \tilde P} \right)
   \label{eq:AGMI-2}
\end{align}
where, writing $X(s_T)=\sqrt{P(s_T)} \, U(s_T)$ for all $s_T$, we have
\begin{align}
  \tilde P & = \E{ \, \left| \E{\left. Y U(S_T)^\text{*} \right| S_T} \right| \, }^2 .
  \label{eq:AGMI-P}
\end{align}
This GMI is achieved by choosing fully-correlated\,\footnote{If $P(s_T)=0$, then the correlation coefficients involving $X(s_T)$ are undefined. However, as long as all $X(s_T)$ with $P(s_T)>0$ are fully correlated we say that all symbols are ``fully correlated''.} symbols:
\begin{align}
   X(s_T) = \sqrt{P(s_T)} \, e^{j\phi(s_T)} \, U
   \label{eq:barX-max}
\end{align}
and $\bar X=c\, U$ for some non-zero constant $c$ and a common $U\sim\mathcal{CN}(0,1)$,
and where
\begin{align}
   \phi(s_T) = - \arg\left( \E{\left. Y \, U(s_T)^\text{*} \right| S_T=s_T} \right).
   \label{eq:phi-sT}
\end{align}
\end{lemma}
\begin{proof}
See Appendix~\ref{appendix:proof-lemma-AGMI-max}.
\end{proof} 

\begin{remark}
The expression \eqref{eq:AGMI-P} is based on~\eqref{eq:AGMI-parameters-RkF} in Appendix~\ref{appendix:proof-lemma-AGMI-max} and can alternatively be written as $\tilde P = \big|\tilde h\big|^2 \bar P$ where
\begin{align*}
  \tilde h & = \E{ Y \bar X^\text{*} } \big/ \bar P.
\end{align*}
\end{remark}

\begin{remark}
The power levels $P(s_T)$ may be optimized, usually under a constraint such as $\E{P(S_T)} \le P$. 
\end{remark}
 
\begin{remark}
By the Cauchy-Schwarz inequality, we have
\begin{align*}
  \E{ \, \left| \E{\left. Y U(S_T)^\text{*} \right| S_T} \right| \, }^2 \le \E{|Y|^2} .
\end{align*}
Furthermore, equality holds if and only if $\big|Y U(s_T)^\text{*}\big|$ is a constant for each $s_T$, but this case is not interesting.
\end{remark}

\subsection{Forward Model GMI for MIMO Channels}
\label{subsec:CSIT-GMI-MIMO}
The following lemma generalizes Lemma~\ref{lemma:AGMI-max} to MIMO channels without claiming a closed-form expression for the optimal GMI. Let $\ul U(s_T)\sim \mathcal{CN}(\ul 0,{\bf I})$ for all $s_T$. 

\begin{lemma} \label{lemma:AGMI-max-MIMO}
A GMI \eqref{eq:AGMI} for the channel $p(\ul y | a)$, an adaptive vector $A$ with jointly CSCG entries, the auxiliary model~\eqref{eq:AGMI-qya}, and with fixed ${\bf Q}_{\ul X(s_T)}$ is given by \eqref{eq:AGMI-3} that we write as
\begin{align}
   I_1(A;\ul Y) = \log \left( \frac{\det{\bf Q}_{\ul Y}}
   {\det\left( {\bf Q}_{\ul Y} - \tilde {\bf D} \, \tilde {\bf D}^\dag \right)} \right).
   \label{eq:AGMI-3-I1}
\end{align}
where for $M \times M$ unitary ${\bf V}_R(s_T)$ we have
\begin{align}
  \tilde {\bf D} = 
  \E{ {\bf U}_T(S_T) \, {\bf \Sigma}(S_T) \, {\bf V}_R(S_T)^\dag}
   \label{eq:AGMI-2-MIMO}
\end{align}
and ${\bf U}_T(s_T)$ and ${\bf \Sigma}(s_T)$ are $N\times N$ unitary and $N\times M$ rectangular diagonal matrices, respectively, of the SVD
\begin{align}
   \E{\left. \ul Y \, {\ul U(s_T)}^\dag \right| S_T=s_T} =  {\bf U}_T(s_T) \, {\bf \Sigma}(s_T) \, {\bf V}_T(s_T)^\dag
   \label{eq:CSIT-SVD-decomposition}
\end{align}
for all $s_T$, and the ${\bf V}_T(s_T)$ are $M\times M$ unitary matrices. The GMI \eqref{eq:AGMI-3-I1} is achieved by choosing the symbols (cf.~\eqref{eq:barX-max} and \eqref{eq:UsT-max-vec} below):
\begin{align}
   \ul X(s_T) = {\bf Q}_{\ul X(s_T)}^{1/2} \, {\bf V}_T(s_T) \, \ul U
   \label{eq:barX-max-vec}
\end{align}
and $\ul{\bar X} = {\bf C} \, \ul U$ for some invertible $M\times M$ matrix $\bf C$ and a common $M$-dimensional vector $\ul U \sim\mathcal{CN}(\ul 0,{\bf I})$. One may maximize \eqref{eq:AGMI-3-I1} over the unitary ${\bf V}_R(s_T)$. 
\end{lemma}
\begin{proof}
See Appendix~\ref{appendix:proof-lemma-AGMI-max-vec}.
\end{proof}

Using Lemma~\ref{lemma:AGMI-max-MIMO}, the theory for MISO channels with $N=1$ is similar to the scalar case of Lemma~\ref{lemma:AGMI-max}; see Remark~\ref{remark:AGMI-N1} below. However, optimizing the GMI is more difficult for $N>1$ because one must optimize over the unitary matrices ${\bf V}_R(s_T)$ in~\eqref{eq:AGMI-2-MIMO}; see Remark~\ref{remark:AGMI-M1} below.

\begin{remark} \label{remark:AGMI-N1}
Consider $N=1$ in which case one may set ${\bf U}_T(s_T)=1$ and \eqref{eq:CSIT-SVD-decomposition} is
a $1 \times M$ vector where ${\bf \Sigma}(s_T)$ has as the only non-zero singular value
\begin{align}
   \sigma(s_T)
   & = \left\| \E{\left. Y \, \ul U(s_T)^\dag \right| S_T=s_T} \right\| \nonumber \\
   & = \left( \sum_{m=1}^M \left| \E{\left. Y \, U_m(s_T)^\text{*} \right| S_T=s_T} \right|^2 \right)^{1/2}.
\end{align}
The absolute value of the scalar~\eqref{eq:AGMI-2-MIMO} is maximized by choosing ${\bf V}_R(s_T)={\bf I}$ for all $s_T$ to obtain (cf.~\eqref{eq:AGMI-P})
\begin{align}
   \tilde {\bf D} \, \tilde {\bf D}^\dag = \E{\sigma(S_T)}^2 . 
\end{align}
\end{remark}

\begin{remark} \label{remark:AGMI-M1}
Consider $M=1$ in which case one may set ${\bf V}_T(s_T)=1$ and \eqref{eq:CSIT-SVD-decomposition} is a $N \times 1$ vector where ${\bf \Sigma}(s_T)$ has as the only non-zero singular value
\begin{align}
    \sigma(s_T)
    & = \left\| \E{\left. \ul Y \, U(s_T)^\dag \right| S_T=s_T} \right\| \nonumber \\
    & = \left( \sum_{n=1}^N \left| \E{\left. Y_n \, U(s_T)^\text{*} \right| S_T=s_T} \right|^2 \right)^{1/2}.
    \label{eq:Lemma-AGMI-Dtilde2-N1}
\end{align}
We should now find the $V_R(s_T)=e^{j\phi_R(s_T)}$ that minimize the determinant in the denominator of \eqref{eq:AGMI-3-I1} where (see~\eqref{eq:AGMI-2-MIMO})
\begin{align}
   \tilde {\bf D}
   = \E{\ul u_T(S_T) \, \sigma(S_T) \, e^{-j\phi_R(S_T)}}
   \label{eq:Lemma-AGMI-Dtilde2-M1}
\end{align}
and where each $\ul u_T(s_T)$ is one of the columns of the $N\times N$ unitary matrix ${\bf U}_T(s_T)$.
\end{remark}

\begin{remark} \label{remark:parallel-H}
Consider $M=N$ and the product channel
\begin{align}
    p(\ul y | a)
    = \prod_{m=1}^M p\big( y_m \,\big|\, [ x_m(s_T): s_T \in \set S_T ] \big)
    \label{eq:product-channel}
\end{align}
where $x_m(s_T)$ is the $m$'th entry of $\ul x(s_T)$. We choose ${\bf Q}_{\ul X(s_T)}$ as diagonal with diagonal entries $\sqrt{P_m(s_T)}$, $m=1,\dots,M$. Also choosing ${\bf V}_R(s_T)={\bf I}$ makes the matrix $\tilde {\bf D} \, \tilde {\bf D}^\dag$ diagonal with the diagonal entries (cf.~\eqref{eq:AGMI-P} where $M=N=1$)
\begin{align}
   \left(\sum_{s_T} P_{S_T}(s_T) \left| \E{\left. Y_m U_m(s_T)^\text{*} \right| S_T=s_T} \right| \right)^2
\end{align}
for $m=1,\dots,M$. The GMI \eqref{eq:AGMI-3-I1} is thus (cf.~\eqref{eq:AGMI-2})
\begin{align}
   \sum_{m=1}^M \log \left( \frac{\E{|Y_m|^2}}
   {\E{|Y_m|^2} - \E{\big| \E{\left. Y_m U_m(S_T)^\text{*} \right| S_T} \big|}^2} \right) .
   \label{eq:CSIT-AGMI-3-product}
\end{align}
\end{remark}

\begin{remark} \label{remark:parallel-q}
For general $p(\ul y | a)$, one might wish to choose diagonal ${\bf Q}_{\ul X(s_T)}$ and a product model
\begin{align*}
    q(\ul y | a)
    = \prod_{m=1}^M q_m( y_m | \bar x_m )
\end{align*}
where the $q_m(.)$ are scalar AWGN channels
\begin{align*}
    q_m(y| x) = \frac{1}{\pi \sigma_m^2} \exp\left( -|y - h_m \, x |^2 / \sigma_m^2 \right)
\end{align*}
with possibly different $h_m$ and $\sigma_m^2$ for each $m$. Consider also
\begin{align*}
    \bar X_m = \sum_{s_T} w_m(s_T) X_m(s_T)
\end{align*}
for some complex weights $w_m(s_T)$, i.e., $\bar X_m$ is a weighted sum of entries from the list $[X_m(s_T):s_T\in\set S_T]$. The maximum GMI is now the same as~\eqref{eq:CSIT-AGMI-3-product} but without requiring the actual channel to have the form~\eqref{eq:product-channel}.
\end{remark}

\begin{remark} \label{remark:H-Tx}
If the actual channel is $\ul Y = {\bf H}\, \ul X + \ul Z$ then
\begin{align}
  \E{\ul Y \, \ul U(s_T)^\dag | S_T = s_T}
  & = \E{{\bf H} \, \ul X(s_T) \, \ul U(s_T)^\dag | S_T = s_T} \nonumber \\
  & = \E{{\bf H} | S_T = s_T} \, {\bf Q}_{\ul X(s_T) }^{1/2}
  \label{eq:appendix-EYU}
\end{align}
where the final step follows because $\ul U(S_T) - S_T - {\bf H}$ forms a Markov chain. 
The expression \eqref{eq:appendix-EYU} is useful because it separates the effects of the channel and the transmitter.
\end{remark}

\begin{remark}
Combining Remarks~\ref{remark:parallel-H} and~\ref{remark:H-Tx}, suppose the actual channel is $\ul Y = {\bf H}\, \ul X + \ul Z$ with $M=N$ and where $\bf H$ is diagonal with diagonal entries $H_m$, $m=1,\dots,M$. The GMI \eqref{eq:AGMI-3-I1} is then (cf.~\eqref{eq:CSIT-AGMI-3-product})
\begin{align}
   \sum_{m=1}^M \log \left( \frac{\E{|Y_m|^2}}
   {\E{|Y_m|^2} - \E{\left| \E{\left. H_m \sqrt{P_m(S_T)} \right| S_T} \right|}^2} \right)
   \label{eq:CSIT-AGMI-3-vector-linear}
\end{align}
where $\E{|Y_m|^2}=1 + \E{|H_m|^2 P_m(S_T)}$.
\end{remark}

\section{Channels with CSIR and CSIT}
\label{sec:CSIRT}
Shannon's model includes CSIR~\cite{McElieceStark84}. The FDG is shown in Fig.~\ref{fig:Salehi-iBM} where there is a hidden state $S_H$, the CSIR $S_R$ and CSIT $S_T$ are functions\footnote{By defining $S_H=[S_{H1},Z_H]$ and calling $S_{H1}$ the hidden channel state we can include the case where $S_R$ and $S_T$ are noisy functions of $S_{H1}$.} of $S_H$, and the receiver sees the channel outputs
\begin{align}
   [Y_i,S_{Ri}] = [f(X_i,S_{Hi},Z_i),S_{Ri}] \label{eq:Shannon-model2}
\end{align}
for some function $f(.)$ and $i=1,2,\ldots,n$. As before, $M$, $S_H^n$, $Z^n$ are mutually statistically independent, and $S_H^n$ and $Z^n$ are i.i.d. strings of random variables with the same distributions as $S_T$ and $Z$, respectively. Observe that we have changed the notation by writing $Y$ for only part of the channel output. The new $Y$ (without the $S_R$) is usually called the ``channel output''. 

\begin{figure}[t]
      \centering
      \includegraphics[width=0.9\columnwidth]{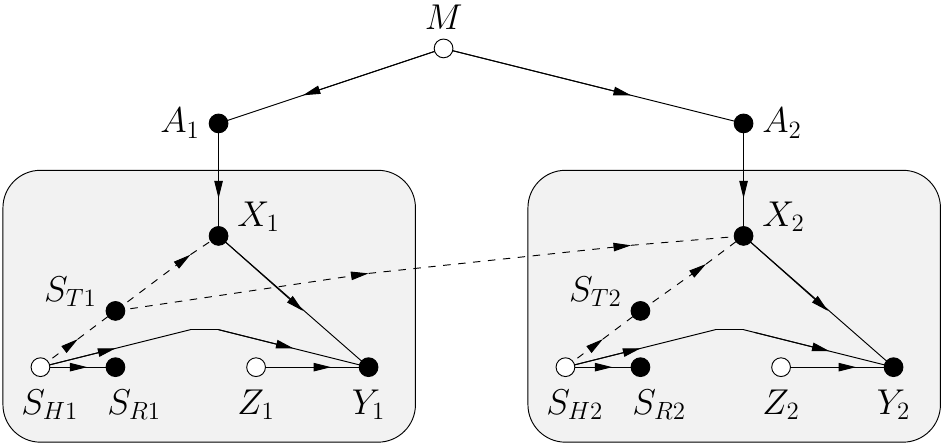}
      \caption{FDG for $n=2$ channel uses with different CSIT and CSIR.
      The hidden channel state $S_{Hi}$ permits dependent $S_{Ri}$ and $S_{Ti}$.}
      \label{fig:Salehi-iBM}
\end{figure}

\subsection{Capacity and GMI}
\label{subsec:CSIRT-capacity}
We begin with scalar channels for which \eqref{eq:Shannon-cap} is
\begin{align} 
   C = \max_{A} \, I(A ; Y, S_R) = \max_{A} \, I(A ; Y | S_R)
   \label{eq:Shannon-cap-SR}
\end{align}
where $A$ and $S_R$ are independent.

\subsubsection{Reverse Model}
The expression~\eqref{eq:AGMI-r2} with the adaptive symbol~\eqref{eq:adaptive-codeword-conventional} is
\begin{align}
    I_1(A;Y,S_R) = \E{ - \log \Var{U | Y,S_R} }.
   \label{eq:AGMI-r3}
\end{align}

\subsubsection{Forward Model}
Consider the expansion
\begin{align}
   I_1(A;Y | S_R) = \int_{\set{S}_R} p(s_R) \, I_1(A;Y | S_R = s_R) \, ds_R
   \label{eq:AGMI-SR0}
\end{align}
where $I_1(A;Y | S_R = s_R)$ is the GMI \eqref{eq:AGMI} with all densities conditioned on $S_R = s_R$. We choose the forward model
\begin{align}
   q(y | a,s_R) = \frac{1}{\pi \sigma(s_R)^2} \exp\left(-\frac{|y - h(s_R)\, \bar x(s_R) |^2}{\sigma(s_R)^2} \right) .
   \label{eq:AGMI-qyasr}
\end{align}
where similar to \eqref{eq:xbar} we define
\begin{align}
   \bar X(s_R) = \sum_{s_T} w(s_T,s_R) \, X(s_T)
   \label{eq:xbar-sR}
\end{align}
for complex weights $w(s_T,s_R)$, i.e.,
$\bar X(s_R)$ is a weighted sum of entries from the list $A=[X(s_T): s_T \in\set{S}_T]$. We have the following straightforward generalization of Lemma~\ref{lemma:AGMI-max}.

\begin{theorem}
\label{theorem:AGMI-SR-max}
The maximum GMI \eqref{eq:AGMI-SR0} for the channel $p(y|a,s_R)$, an adaptive symbol $A$ with jointly CSCG entries, the model \eqref{eq:AGMI-qyasr}, and with fixed $P(s_T)=\E{|X(s_T)|^2}$ is
\begin{align}
   I_1(A;Y|S_R) = \E{ \log\left( 1 + \frac{\tilde P(S_R)}{\E{|Y|^2|S_R} - \tilde P(S_R)} \right) }
   \label{eq:AGMI-SR}
\end{align}
where for all $s_R \in \set{S}_R$ we have
\begin{align}
  \tilde P(s_R) & = \E{\,\left| \E{\left. Y U(S_T)^\text{*} \right| S_T, S_R=s_R} \right| \,}^2 .
  \label{eq:AGMI-P2}
\end{align}
\end{theorem}

\begin{remark} \label{remark:xbar-scaling2}
To establish Theorem~\ref{theorem:AGMI-SR-max}, the receiver may choose $\bar X = \sqrt{P} \, U$ to be independent of $s_R$. Alternatively, the receiver may choose $\bar X(s_R) = \sqrt{\E{|X|^2 | S_R=s_R}} \, U$. Both choices give the same GMI since the expectation in~\eqref{eq:AGMI-P2} does not depend on the scaling of $\bar X$; see Remark~\ref{remark:xbar-scaling}.
\end{remark}

\begin{remark} \label{remark:AGMI-part3}
The partition idea of Lemmas~\ref{lemma:GMI-part} and~\ref{lemma:AGMI-part} carries over to Theorem~\ref{theorem:AGMI-SR-max}. We may generalize \eqref{eq:AGMI-SR} as
\begin{align}
   & I_1(A;Y|S_R) = \int_{\set S_R} p(s_R) \sum_{k=1}^K \Pr{\set E_k|S_R=s_R} \nonumber \\
   & \left[ \log\left( 1 + \frac{|h_k(s_R)|^2 P}{\sigma_k^2(s_R)} \right) 
   + \frac{\E{\left. |Y|^2 \right| \set{E}_k, S_R=s_R}}{\sigma_k^2(s_R) + |h_k(s_R)|^2 P} \right. \nonumber \\
   & \left. \quad - \frac{\E{ \left. |Y - h_k(s_R) \sqrt{P} \, U|^2 \right| \set{E}_k, S_R=s_R} }{\sigma_k^2(s_R)} \right] \, ds_R
   \label{eq:AGMI-SR-part2}
\end{align}
where the $X(s_T)$, $s_T \in \set S_T$, are given by \eqref{eq:barX-max} and the $h_k(s_R)$ and $\sigma_k^2(s_R)$, $k=1,\dots,K$, $s_R \in \set S_R$, can be optimized.
\end{remark}

\begin{remark}
One is usually interested in the optimal power control policy $P(s_T)$ under the constraint $\E{P(S_T)}\le P$.  Taking the derivative of \eqref{eq:AGMI-SR} with respect to $\sqrt{P(s_T)}$ and setting to zero we obtain
\begin{align}
  & \E{ \frac{ \E{|Y|^2 | S_R} \tilde P(S_R)' 
  - \tilde P(S_R) \, \E{|Y|^2|S_R}' 
  }{ \E{|Y|^2 | S_R} \left[ \E{|Y|^2 | S_R} - \tilde P(S_R)\right]} } \nonumber \\
  & = 2 \lambda \sqrt{P(s_T)} P_{S_T}(s_T)
  \label{eq:P-policy-deriv}
\end{align}
where $\tilde P(S_R)'$ and $\E{|Y|^2 | S_R}'$ are derivatives with respect to $\sqrt{P(s_T)}$. We use \eqref{eq:P-policy-deriv} below to derive power control policies.
\end{remark}

\begin{remark}
A related model is a \emph{compound} channel where $p(y | a, s_R)$ is indexed by the parameter $s_R$~\cite[Ch.~4]{Wolfowitz64}. The problem is to find the maximum worst-case reliable rate if the transmitter does not know $s_R$. Alternatively, the transmitter must send its message to all $|\set{S}_R|$ receivers indexed by $s_R\in\set{S}_R$. A compound channel may thus be interpreted as a broadcast channel with a common message.
\end{remark}

\subsection{CSIT\at R}
\label{subsec:CSITatR}
An interesting specialization of Shannon's model is when the receiver knows $S_T$ and can determine $X(S_T)$. We refer to this scenario as CSIT\at R. The model was considered in~\cite[Sec.~II]{Caire-Shamai-IT99} when $S_T$ is a function of $S_R$. More generally, suppose $S_T$ is a function of $[Y,S_R]$. The capacity \eqref{eq:Shannon-cap-SR} then simplifies to
(see \cite[Prop.~1]{Caire-Shamai-IT99})
\begin{align} 
   C & \overset{(a)}{=} \max_{A} \, I(A ; Y, S_T | S_R) \nonumber \\
   & \overset{(b)}{=} \max_{A} \, I(X ; Y | S_R, S_T) \nonumber \\
   & \overset{(c)}{=} \sum_{s_T} P_{S_T}(s_T) \left[ \max_{X(s_T)} \, I(X(s_T) ; Y | S_R, S_T=s_T) \right]
   \label{eq:Shannon-cap1}
\end{align}
where step $(a)$ follows because $S_T$ is a function of $[Y,S_R]$; step $(b)$ follows because $I(A;S_T|S_R)=0$, $X$ is a function of $[A,S_T]$, and $A-[S_T,X]-Y$ forms a Markov chain; and step $(c)$ follows because one may optimize $X(s_T)$ separately for each $s_T \in \set{S}_T$. 
 
As discussed in~\cite{Caire-Shamai-IT99}, a practical motivation for this model is when the CSIT is based on error-free feedback from the receiver to the transmitter. In this case, where $S_T$ is a function of $S_R$, the expression \eqref{eq:AGMI-P2} becomes
\begin{align}
  \tilde P(s_R) & = \left| \E{\left. Y U(s_T)^\text{*} \; \right| S_R=s_R \, } \right|^2 .
  \label{eq:AGMI-P2-CS}
\end{align}

\begin{remark} \label{remark:CSITatR}
The insight that one can replace adaptive symbols $A$ with channel inputs $X$ when $X$ is a function of $A$ and past $Y$ appeared for two-way channels in~\cite[Sec.~4.2.3]{Kramer98} and networks in~\cite[Sec.~V.A]{Kramer14},~\cite[Sec.~IV.F]{Kramer03}.
\end{remark}

\subsection{MIMO Channels and $K$-Partitions}
\label{subsec:GMI-generalizations2}
We consider generalizations to MIMO channels and $K$-partitions as in Sec.~\ref{subsec:GMI-generalizations}.

\subsubsection{MIMO Channels}
Consider the average GMI
\begin{align}
   I_1(A; \ul Y | S_R) = \int_{\set{S}_R} p(s_R) I_1(A; \ul Y | S_R = s_R) \, ds_R
   \label{eq:GMI-2-avg2}
\end{align}
and choose the parameters \eqref{eq:AGMI-parameters-2a}--\eqref{eq:AGMI-parameters-2b} for the event $S_R=s_R$. We have
\begin{align}
   \tilde {\bf H}(s_R) & = \E{\left. \ul Y \, \ul {\bar X}^\dag \right| S_R = s_R} \E{\left. \ul {\bar X} \, \ul {\bar X}^\dag  \right| S_R = s_R}^{-1}
   \label{eq:GMI-parameters-c1-2} \\
   {\bf Q}_{\ul{\Zt}}(s_R) & = \E{\left. \ul Y \, \ul Y^\dag \right| S_R=s_R} \nonumber \\
   & \quad - \tilde {\bf H}(s_R) \E{\left. \ul {\bar X} \, \ul {\bar X}^\dag  \right| S_R = s_R} \tilde {\bf H}(s_R)^\dag
   \label{eq:GMI-parameters-c2-2}
\end{align}
and the GMI~\eqref{eq:GMI-2-avg2} is (cf.~\eqref{eq:GMI-2-avg-a} and~\eqref{eq:AGMI-3})
\begin{align}
   \E{\log\det \left( {\bf I} + {\bf Q}_{\ul{\Zt}}(S_R)^{-1} \, \tilde {\bf H}(S_R)
   {\bf Q}_{\ul{\bar X}} \tilde {\bf H}(S_R)^\dag \right) } .
   \label{eq:GMI-2-avg2a}
\end{align}

\subsubsection{$K$-Partitions}
Let $\{\ul{\set Y}_k: k=1,\dots,K\}$  be a $K$-partition of $\ul{\set Y}$ and define the events $\set E_k = \{\ul Y \in \ul{\set Y}_k\}$ for $k=1,\dots,K$. As in Remark~\ref{remark:GMI-part1},
$K$-partitioning formally includes \eqref{eq:GMI-2-avg2} as a special case by including $S_R$ as part of the receiver's ``overall'' channel output $\ul{\Yt} = [\ul Y,S_R]$. The following lemma generalizes Lemma~\ref{lemma:GMI-part}.

\begin{lemma} \label{lemma:AGMI-part}
A GMI with $s=1$ for the channel $p(\ul y | a)$ is
\begin{align}
   I_1(A;\ul Y) & = \sum_{k=1}^K \Pr{\set E_k} \left\{ \log\det \left( {\bf I}
   + {\bf Q}_{\ul Z_k}^{-1} \,
    {\bf H}_k {\bf Q}_{\ul{\bar X}} {\bf H}_k^\dag \right) \right. \nonumber \\
   & + \E{ \left. {\ul Y}^\dag \left( {\bf Q}_{\ul Z_k} + {\bf H}_k {\bf Q}_{\ul{\bar X}} {\bf H}_k^\dag \right)^{-1} \ul Y \, \right| \set E_k } \nonumber \\
   & \left. - \E{ \left. \left( \ul Y - {\bf H}_k\, \ul{\bar X} \right)^\dag {\bf Q}_{\ul{Z}_k}^{-1} \left( \ul Y - {\bf H}_k\, \ul{\bar X} \right) \right| \set E_k } \right\}
   \label{eq:AGMI-part}
\end{align}
where the ${\bf H}_k$ and ${\bf Q}_{\ul{Z}_k}$, $k=1,\dots,K$, can be optimized.
\end{lemma}

\begin{remark} \label{remark:AGMI-part2}
For scalars the GMI~\eqref{eq:AGMI-part} is
\begin{align}
   I_1(A;Y) & = \sum_{k=1}^K \Pr{\set E_k} \left[ \log\left( 1 + \frac{|h_k|^2 \bar P}{\sigma_k^2} \right) \right. \nonumber \\
   & \left. + \frac{\E{|Y|^2 | \set{E}_k}}{\sigma_k^2 + |h_k|^2 \bar P} - \frac{\E{|Y - h_k \bar X|^2 | \set{E}_k}}{\sigma_k^2} \right]
   \label{eq:AGMI-part2}
\end{align}
which is the same as~\eqref{eq:GMI-2-part} except that $\bar X$, $\bar P$ replace $X,P$. If we follow~\eqref{eq:GMI-part-h}--\eqref{eq:GMI-part-sigma2} then \eqref{eq:AGMI-part2} becomes \eqref{eq:GMI-3-part} but with
\begin{align*}
   & h_k = \E{\left. Y \bar X^\text{*} \right| \set E_k}/P_k, \quad
   P_k = \E{\left. \left|\bar X \right|^2 \right| \set E_k}.
\end{align*}
\end{remark}

\begin{remark} \label{remark:AGMI-part2-2}
Consider Remark~\ref{remark:GMI-part4} and choose $K=2$, $h_1=0$, $\sigma_1^2=1$. The GMI \eqref{eq:AGMI-part2} then has only the $k=2$ term, and it again remains to select $h_2$, $\sigma_2^2$, and $t_R$.
\end{remark}

\begin{remark}
If we define
\begin{align}
   {\bf Q}_{\ul {\bar X}}^{(k)} = \E{ \left. \ul {\bar{X}} \, \ul {\bar{X}}^\dag \right| \set{E}_k }, \quad {\bf Q}_{\ul Y}^{(k)} = \E{ \left. \ul Y \, \ul Y^\dag \right| \set{E}_k }
\end{align}
and choose the LMMSE auxiliary models with
\begin{align}
   {\bf H}_k & = \E{\left. \ul Y \, \ul {\bar X}^\dag \right| \set E_k} \left({\bf Q}_{\ul {\bar X}}^{(k)}\right)^{-1}
   \label{eq:GMI-parameters-part-hk-2} \\
    {\bf Q}_{\ul Z_k} & =
    {\bf Q}_{\ul Y}^{(k)}  - {\bf H}_k {\bf Q}_{\ul {\bar X}}^{(k)} {\bf H}_k^\dag 
   \label{eq:GMI-3-part-denom-2}
\end{align}
for $k=1,\dots,K$ then the expression~\eqref{eq:AGMI-part} is (cf.~\eqref{eq:GMI-3-part})
\begin{align}
   & \sum_{k=1}^K \Pr{\set E_k} \left[ \log\det \left( {\bf I}
   + {\bf Q}_{\ul Z_k}^{-1} \,
    {\bf H}_k {\bf Q}_{\ul{\bar X}} {\bf H}_k^\dag \right) \right. \nonumber \\
   & \left. \qquad - \tr{ \left( {\bf Q}_{\ul Y}^{(k)} + {\bf H}_k {\bf D}_{\ul{\bar X}}^{(k)} {\bf H}_k^\dag \right)^{-1} {\bf H}_k {\bf D}_{\ul{\bar X}}^{(k)} {\bf H}_k^\dag } \right]
   \label{eq:AGMI-3-part}
\end{align}
where ${\bf D}_{\ul{\bar X}}^{(k)}={\bf Q}_{\ul{\bar X}} - {\bf Q}_{\ul{\bar X}}^{(k)}$.
\end{remark}

\begin{remark} \label{remark:large-K2} 
We may proceed as in Remark~\ref{remark:large-K} and consider large $K$. These steps are given in Appendix~\ref{appendix:GMI-MMSE2}.
\end{remark}

\section{Fading Channels with AWGN}
\label{sec:Gauss-fading}
This section treats scalar, complex-alphabet, AWGN channels with CSIR for which the channel output is
\begin{align}
  [Y, S_R] = [H X + Z, S_R]
  \label{eq:Y-Gauss}
\end{align}
where $H,A,Z$ are mutually independent, $\E{|H|^2}=1$, and $Z\sim\mathcal{CN}(0,1)$. The capacity under the power constraint $\E{|X|^2}\le P$ is (cf.~\eqref{eq:Shannon-cap-SR})
\begin{align}
  C(P) = \max_{A:\, \E{|X|^2}\le P} I(A ; Y | S_R).
  \label{eq:Shannon-cap-SR-AWGN}
\end{align}
However, the optimization in~\eqref{eq:Shannon-cap-SR-AWGN} is often intractable, and we desire expressions with $\log(1+\text{SNR})$ terms to gain insight. We develop three such expressions: an upper bound and two lower bounds.  It will be convenient to write $G=|H|^2$.

\subsubsection{Capacity Upper Bound}
We state this bound as a lemma since we use it to prove Proposition~\ref{proposition:part-CSIT2} below.

\begin{lemma} \label{lemma:Gauss-fading-C-UB}
The capacity~\eqref{eq:Shannon-cap-SR-AWGN} is upper bounded as
\begin{align}
  C(P) \le \max \; \E{ \log\left( 1 +  G P(S_T) \right) }
  \label{eq:Shannon-cap-SR-AWGN-UB}
\end{align}
where the maximization is over $P(S_T)$ with $\E{P(S_T)}=P$.
\end{lemma}
\begin{proof}
Consider the steps
\begin{align}
  I(A;Y|S_R)
  & \le I(A;Y,S_T,H|S_R) \nonumber \\
  & \overset{(a)}{=} I(A;Y|S_R,S_T,H) \nonumber \\
  & = h(Y | S_R,S_T,H) - h(Z) \nonumber \\
  & \overset{(b)}{\le} \E{\log \Var{ Y | S_R,S_T,H}}
  \label{eq:IAYSR-UB}
\end{align}
where step $(a)$ is because $A$ and $[S_R,S_T,H]$ are independent, and step $(b)$ follows by the entropy bound
\begin{align}
  h(Y | B=b) \le \log \left( \pi e \, \Var{Y | B=b} \right) \label{eq:h-UB4}
\end{align}
which we applied with $B=[S_R,S_T,H]$. Finally, we compute $\Var{Y|S_R,S_T,H}=1 + G P(S_T)$.
\end{proof}

\subsubsection{Reverse Model GMI}
Consider the adaptive symbol~\eqref{eq:adaptive-codeword-conventional} and the GMI~\eqref{eq:AGMI-r3}. We expand the variances in~\eqref{eq:AGMI-r3} as
\begin{align*}
    & \Var{U | Y=y, S_R=s_R} \nonumber \\
    & = \E{|U|^2 \big| Y=y, S_R=s_R} - \big| \E{U \big| Y=y, S_R=s_R} \big|^2 .
\end{align*}
Appendix~\ref{appendix:c-second-order-statistics} shows that one may write
\begin{align}
    & \E{U \big| Y=y, S_R=s_R} \nonumber \\
    & = \int_{\mathbb C \times \set S_T} p(h,s_T | y,s_R) \, \frac{h\sqrt{P(s_T)}e^{j\phi(s_T)} y}{1+|h|^2 P(s_T)} \, ds_T \, dh
    \label{eq:AGMI-5-part-Es1a} 
\end{align}
and
\begin{align}
    & \E{|U|^2 \big| Y=y, S_R=s_R}
    = \int_{\mathbb C \times \set S_T} p(h,s_T | y,s_R) \nonumber \\
    & \quad \left( \frac{1}{1 + |h|^2 P(s_T)} + \frac{ |h|^2 P(s_T) |y|^2 }{\left(1 + |h|^2 P(s_T)\right)^2} \right) \, ds_T \, dh.
   \label{eq:AGMI-5-part-Es2a} 
\end{align}
We use the expressions~\eqref{eq:AGMI-5-part-Es1a}--\eqref{eq:AGMI-5-part-Es2a} to compute achievable rates by numerical integration. For example, suppose $S_T=0$ and $S_R=H$, i.e., we have full CSIR and no CSIT. The averaging density is then
\begin{align*}
    p(h,s_T | y,s_R) = \delta(h-s_R) \, \delta(s_T)
\end{align*}
and the variance simplifies to the capacity-achieving form
\begin{align*}
    \Var{U | Y=y, S_R=h}
    = \frac{1}{1 + |h|^2 P}.
\end{align*}

\subsubsection{Forward Model GMI}
A forward model GMI is given by Theorem~\ref{theorem:AGMI-SR-max} where
\begin{align}
  & \tilde P(s_R) = \E{ \, \left| \E{\left. H \sqrt{P(S_T)} \right| S_T, S_R=s_R} \right| \,}^2
  \label{eq:AGMI-Pt}
  \\
  & \E{|Y|^2 \big| S_R=s_R} = 1 + \E{G P(S_T) \big| S_R=s_R}
  \label{eq:AGMI-Y2}
\end{align}
so that $I_1(A;Y|S_R)$ in \eqref{eq:AGMI-SR} becomes
\begin{align}
   \E{ \log \left(1 + \frac{\tilde P(S_R)}
   {1 + \E{G P(S_T) \big| S_R} - \tilde P(S_R)} \right) }.
   \label{eq:AGMI-SR-Gauss}
\end{align}

\begin{remark} \label{remark:variance}
Jensen's inequality implies that the denominator in \eqref{eq:AGMI-SR-Gauss} is greater than or equal to
\begin{align}
   1+\Var{\left. \sqrt{G P(S_T)} \, \right| S_R}.
   \label{eq:AGMI-denominator}
\end{align}
Equality requires that for all $S_R=s_R$ we have
\begin{align}
  \tilde P(s_R) = \E{\left. \sqrt{G P(S_T)} \right| S_R=s_R}^2
  \label{eq:AGMI-desire}
\end{align}
which is valid if $H$ is a function of $[S_R,S_T]$, for example.
However, if there is channel uncertainty after conditioning on $[S_R,S_T]$ then $\tilde P(s_R)$ is usually smaller than the RHS of \eqref{eq:AGMI-desire}.
\end{remark}

\begin{remark} \label{remark:full-CSIR}
Consider $S_R=H$ or $S_R = H \sqrt{P(S_T)}$. 
For both cases, $H$ is a function of $[S_R,S_T]$ and the denominator in \eqref{eq:AGMI-SR-Gauss} is the variance \eqref{eq:AGMI-denominator}. In fact, for $S_R = H \sqrt{P(S_T)}$, the expression \eqref{eq:AGMI-denominator} takes on the minimal value 1. This CSIR is thus the best possible; see Proposition~\ref{proposition:part-CSIT2}.
\end{remark}

\begin{remark}
For MIMO channels we replace \eqref{eq:Y-Gauss} with
\begin{align}
  [\Yv, S_R] = [{\bf H} \Xv + \Zv, S_R]
  \label{eq:MIMO-model}
\end{align}
where ${\bf H}, A,\Zv$ are mutually independent and $\Zv \sim \mathcal{CN}(\ul 0,{\bf I})$. One usually considers the constraint $\E{\|\Xv\|^2}\le P$.
\end{remark}

\begin{remark}
The model~\eqref{eq:MIMO-model} includes block fading. For example, choosing $M=N$ and ${\bf H}=H\, {\bf I}$ gives scalar block fading. Moreover, the capacity per symbol without in-block feedback is the same as for the $M=N=1$ case except that $P$ is replaced with $P/M$; see~\cite{McElieceStark84} and Sec.~\ref{sec:ibf-CSIT}.
\end{remark}

\subsection{CSIR and CSIT Models}
\label{subsec:Gauss-fading-channel-models}
We study two classes of CSIR, as shown in Table~\ref{table:model-classes}. The first class has full (or ``perfect'') CSIR, by which we mean either $S_R=H$ or $S_R =H\sqrt{P(S_T)}$. The motivation for studying the latter case is that it models block fading channels with long blocks where the receiver estimates $H\sqrt{P(S_T)}$ using pilot symbols, and the number of pilot symbols is much smaller than the block length~\cite{Caire-Shamai-IT99}. Moreover, one achieves the upper bound \eqref{eq:Shannon-cap-SR-AWGN-UB}, see Proposition~\ref{proposition:part-CSIT2} below.

We coarsely categorize the CSIT as follows:
\begin{itemize}
\item Full CSIT: $S_T=H$;
\item CSIT\at R: $S_T=q_u(G)$ where $q_u(.)$ is the quantizer of Sec.~\ref{subsec:quantizer} with $B=0,1,\infty$;
\item Partial CSIT: $S_T$ is not known exactly at the receiver.
\end{itemize}
The capacity of the CSIT\at R models is given by $\log(1+\text{SNR})$ expressions~\cite{Caire-Shamai-IT99,Kim-Skoglund-07}; see also~\cite{Rosenzweig-IT05}. The partial CSIT model is interesting because achieving capacity generally requires adaptive codewords and closed-form capacity expressions are unavailable. The GMI lower bound of Theorem~\ref{theorem:AGMI-SR-max} and Remark~\ref{remark:AGMI-part3} and the capacity upper bound of Lemma~\ref{lemma:Gauss-fading-C-UB} serve as benchmarks.

\renewcommand{\arraystretch}{1.1}
\begin{table}[t!]
\centering
\caption{Models Studied in Sec.~\ref{sec:Gauss-fading} (General Fading), \protect\linebreak Sec.~\ref{sec:oof} (On-off Fading), and Sec.~\ref{sec:Rayleigh-Fading} (Rayleigh Fading)}
\setlength{\tabcolsep}{2mm}
\begin{tabular}{ll|cc|}
& & \multicolumn{2}{c|}{\bf CSIR} \\
& & Full & Partial\,/\,No \\
\hline
\multirow{3}{*}{\bf CSIT}
& Full & Sec.~\ref{subsec:Full-CSIR-CSITatR} & Sec.~\ref{subsec:Partial-CSIR-fullCSIT}  \\ 
& \at R & Sec.~\ref{subsec:Full-CSIR-CSITatR} & Sec.~\ref{subsec:Partial-CSIR-CSITatR}  \\ 
& Partial\,/\,No & Sec.~\ref{subsec:Full-CSIR-part-CSIT}  & Sec.~\ref{subsec:No-CSIR-No-CSIT}  \\ 
\hline
\end{tabular}
\label{table:model-classes}
\end{table}

The partial CSIR models have $S_R$ being a lossy function of $H$. For example, a common model is based on LMMSE channel estimation with
\begin{align}
    H = \sqrt{\epsb} \, S_R + \sqrt{\epsilon} \, Z_R
    \label{eq:H-mmse-estimate}
\end{align}
where $0\le\epsilon\le1$ and $S_R,Z_R$ are uncorrelated. The CSIT is categorized as above, except that we consider $S_T=f_T(S_R)$ for some function $f_T(.)$ rather than $S_T=q_u(G)$.

To illustrate the theory, we study two types of fading: one with discrete $H$ and one with continuous $H$, namely
\begin{itemize}
\item Sec.~\ref{sec:oof}:\ on-off fading with $P_H(0)=P_H(\sqrt{2})=1/2$;
\item Sec.~\ref{sec:Rayleigh-Fading}:\ Rayleigh fading with $H \sim \mathcal{CN}(0,1)$.
\end{itemize}
For on-off fading we have $p(g)=\frac{1}{2} \delta(g) + \frac{1}{2} \delta(g-2)$ and for Rayleigh fading we have $p(g)=e^{-g} \cdot 1(g\ge0)$.

\begin{remark}
For channels with partial CSIR, we will study the GMI for partitions with $K=1$ and $K=2$. The full CSIT model has received relatively little attention in the literature, perhaps because CSIR is usually more accurate than CSIT~\cite[Sec.~4.2.3]{Keshet-Steinberg-Merhav-08}.
\end{remark}

\subsection{No CSIR, No CSIT}
\label{subsec:No-CSIR-No-CSIT}
Without CSIR or CSIT, the channel is a classic memoryless channel~\cite{Shannon48} for which the capacity~\eqref{eq:Shannon-cap-SR-AWGN} becomes the usual expression with $S_R=0$ and $A=X$. For CSCG $X$ and $U=X/\E{|X|^2}$, the reverse and forward model GMIs~\eqref{eq:AGMI-r3} and~\eqref{eq:AGMI-SR-Gauss} are the respective
\begin{align}
    I_1(X;Y) & = \E{-\log \Var{U|Y}}
    \label{eq:noCSIR-noCSIT-Ir} \\
    I_1(X;Y) & = \log\left(1+\frac{P \, |\E{H}|^2}{1+P\,\Var{H}} \right).
    \label{eq:noCSIR-noCSIT-If}
\end{align}
For example, the forward model GMI is zero if $\E{H}=0$.

\subsection{Full CSIR, CSIT\at R}
\label{subsec:Full-CSIR-CSITatR}
Consider $S_R=H$ and CSIT\at R. The capacity is given by $\log(1+\text{SNR})$ expressions that we review. 

First, the capacity with $B=0$ (no CSIT) is 
\begin{align}
   C(P) & = \E{\log\left( 1 + G\, P \right)} \nonumber \\
   & = \int_0^{\infty} p(g) \log\left( 1 + g P \right) \, dg .
   \label{eq:no-CSIT}
\end{align}
The wideband derivatives are (see~\eqref{eq:wideband})
\begin{align}
  C'(0)=\E{G}=1, \quad C''(0)=-\E{G^2}
  \label{eq:wideband2}
\end{align}
so that the wideband values \eqref{eq:wideband} are (see~\cite[Thm.~13]{Verdu02})
\begin{align}
   \left.\frac{E_b}{N_0}\right|_{\text{min}} = \log 2, \quad S = \frac{2}{\E{G^2}} .
   \label{eq:wideband-fullCSIR-noCSIT}
\end{align}
The minimal $E_b/N_0$ is the same as without fading, namely $-1.59$ dB. However, Jensen's inequality gives $\E{G^2}\ge \E{G}^2=1$ with equality if and only if $G=1$. Thus, fading reduces the capacity slope $S$.

More generally, the capacity with full CSIR and $S_T=q_u(G)$ is (see~\cite{Caire-Shamai-IT99})
\begin{align}
   & C(P) = \max_{P(S_T):\, \E{P(S_T)} \le P} \E{\log\left( 1 + G\, P(S_T) \right)} \nonumber \\
   & = \max_{P(S_T):\, \E{P(S_T)}\le P} \int_0^{\infty} p(g,s_T) \log\left( 1 + g P(s_T) \right) \, dg \, ds_T.
   \label{eq:full-CSIT}
\end{align}
To optimize the power levels $P(s_T)$, consider the Lagrangian
\begin{align}
   \E{\log\left( 1 + G P(S_T) \right)} + \lambda \left(P - \E{P(S_T)}\right)
   \label{eq:Lagrangian}
\end{align}
where $\lambda\ge0$ is a Lagrange multiplier. Taking the derivative with respect to $P(s_T)$, we have
\begin{align}
   \lambda & = \E{ \left. \frac{G}{1+ G P(s_T) } \right| S_T=s_T } \nonumber \\
   & = \int_0^\infty p(g|s_T) \frac{g}{1+ g P(s_T)} dg
   \label{eq:lambda}
\end{align}
as long as $P(s_T)\ge 0$. If this equation cannot be satisfied, choose $P(s_T)=0$. Finally, set $\lambda$ so that $\E{P(S_T)}=P$.

For example, consider $B=\infty$ and $S_T=G$. We then have $p(g|s_T)=\delta(g-s_T)$ and therefore 
\begin{align}
   P(g) = \left(\frac{1}{\lambda} - \frac{1}{g} \right)^+
   \label{eq:waterfilling}
\end{align}
where $\lambda$ is chosen so that $\E{P(G)}=P$. The capacity \eqref{eq:full-CSIT} is then (see~\cite[Eq.~(7)]{Goldsmith-Varaiya-IT97})
\begin{align}
   C(P) & = \int_{\lambda}^{\infty} p(g) \log\left( g / \lambda \right) \; dg.
   \label{eq:full-CSIT-2}
\end{align}

Consider now the quantizer $q_u(.)$ of Sec.~\ref{subsec:quantizer} with $B=1$. We have two equations for $\lambda$, namely
\begin{align}
   \lambda & = \int_0^{\Delta} \frac{p(g)}{P_{S_T}(\Delta/2)} \cdot \dfrac{g}{1+ g P(\Delta/2)} dg
   \label{eq:lambda-q0} \\
   \lambda &= \int_{\Delta}^{\infty} \frac{p(g)}{P_{S_T}(3\Delta/2)} \cdot \dfrac{g}{1+ g P(3\Delta/2)} dg.
   \label{eq:lambda-q1} 
\end{align}
Observe the following for~\eqref{eq:lambda-q0}--\eqref{eq:lambda-q1}:
\begin{itemize}
\item both $P(\Delta/2)$ and $P(3\Delta/2)$ decrease as $\lambda$ increases;
\item the maximal $\lambda$ permitted by \eqref{eq:lambda-q0} is $\E{G| G\le\Delta}$ which is obtained with $P(\Delta/2)=0$;
\item the maximal $\lambda$ permitted by \eqref{eq:lambda-q1} is $\E{G| G\ge\Delta}$ which is obtained with $P(3\Delta/2)=0$.
\end{itemize}
Thus, if $\E{G| G\ge\Delta} > \E{G| G\le\Delta}$, then at $P$ below some threshold, we have $P(\Delta/2)=0$ and $P(3\Delta/2)=P/P_{S_T}(3\Delta/2)$. The capacity in nats per symbol at low power and for fixed $\Delta$ is thus
\begin{align}
   & C(P) = \int_\Delta^\infty p(g) 
   \log\left( 1 + g P(3\Delta/2) \right)
   dg \nonumber \\ 
   & \approx P\, \E{G | G \ge \Delta } - \frac{P^2}{2 P_{S_T}(3\Delta/2)} \E{G^2 | G \ge \Delta}
\end{align}
where we used
\begin{align*}
  \log(1+x) \approx x - \frac{x^2}{2}
\end{align*}
for small $x$. The wideband values \eqref{eq:wideband} are
\begin{align}
   & \left.\frac{E_b}{N_0}\right|_{\text{min}} = \frac{\log 2}{\E{G | G \ge \Delta}} \label{eq:minEbNo-2} \\
   & S = \frac{2 P_{S_T}(3\Delta/2) \, \E{G | G \ge \Delta}^2}{\E{G^2 | G \ge \Delta}} . \label{eq:S-2}
\end{align}
One can thus make the minimum $E_b/N_0$ approach $-\infty$ if one can make $\E{G | G \ge \Delta}$ as large as desired by increasing $\Delta$.

\begin{remark}
Consider the MIMO model \eqref{eq:MIMO-model} with $S_R={\bf H}$. Suppose the CSIT is $S_T=f_T(S_R)$ for some function $f_T(\cdot)$. The capacity \eqref{eq:full-CSIT} generalizes to
\begin{align}
   & C(P) 
   = \max_{\Xv(S_T):\, \E{\|\Xv(S_T)\|^2}\le P} I(\Xv ; {\bf H} \Xv + \Zv  | {\bf H}, S_T ) \nonumber \\
   & = \max_{{\bf Q}(S_T):\, \E{\tr{{\bf Q}(S_T)}} \le P} \E{\log\det\left({\bf I} + {\bf H} {\bf Q}(S_T) {\bf H}^\dag \right)} .
   \label{eq:Shannon-cap-MIMO3}
\end{align}
\end{remark}

\subsection{Full CSIR, Partial CSIT}
\label{subsec:Full-CSIR-part-CSIT}
Consider first the full CSIR $S_R = H\sqrt{P(S_T)}$ and then the less informative $S_R=H$.

\subsubsection{$S_R = H\sqrt{P(S_T)}$}
We have the following capacity result that implies this CSIR is the best possible since one can achieve the same rate as if the receiver sees both $H$ and $S_T$; see the first step of~\eqref{eq:IAYSR-UB}. We could thus have classified this model as CSIT\at R.

\begin{proposition}[{see~\cite[Prop.~3]{Caire-Shamai-IT99}}]
\label{proposition:part-CSIT2}
The capacity of the channel \eqref{eq:Y-Gauss} with $S_R=H\sqrt{P(S_T)}$ and general $S_T$ is
\begin{align}
  C(P)
  & = \max_{P(S_T):\, \E{P(S_T)} \le P} \int_{\mathbb C} p(s_R) \log\left( 1 + |s_R|^2 \right) \, ds_R \nonumber \\ 
  & = \max_{P(S_T):\, \E{P(S_T)} \le P} \E{ \log\left( 1 +  G P(S_T) \right) }.
  \label{eq:GMI-AYH2} 
\end{align}
\end{proposition}
\begin{proof}
Achievability follows by Theorem~\ref{theorem:AGMI-SR-max} with Remark~\ref{remark:full-CSIR}. The converse is given by Lemma~\ref{lemma:Gauss-fading-C-UB}. 
\end{proof}

\begin{remark} \label{remark:CS-part}
Proposition~\ref{proposition:part-CSIT2} gives an upper bound and (thus) a target rate when the receiver has partial CSIR. For example, we will use the $K$-partition idea of Lemma~\ref{lemma:GMI-part} (see also Remark~\ref{remark:AGMI-part2}) to approach the upper bound for large SNR.
\end{remark}

\begin{remark} \label{remark:CS-care}
Proposition~\ref{proposition:part-CSIT2} partially generalizes to block-fading channels; see Proposition~\ref{proposition:part-CSIT2-block-fading} in Sec.~\ref{subsec:ibf-full-CSIR-part-CSIT}.
\end{remark}

\subsubsection{$S_R = H$}
The capacity is~\eqref{eq:Shannon-cap-SR} with
\begin{align}
  I(A;Y | H) = \E{\log \frac{p(Y|A,H)}{p(Y|H)}}
  \label{eq:IAYH}
\end{align}
where $\E{|X|^2}\le P$ and where
\begin{align}
  p(y | a, h)
  & = \int_{\mathbb C} p(s_T | h) \, \frac{e^{-|y - h \, x(s_T)|^2}}{\pi}  \, ds_T
  \label{eq:pyah-1}
\end{align}
and
\begin{align}
  & p(y | h) = \int_{\mathbb C} p(s_T | h) \left( \int_\set{A} p(a) p(y | a, h, s_T) \, da \right) ds_T \nonumber \\
  & = \int_{\mathbb C} p(s_T | h) \left( \int_{\mathbb C} p(x(s_T)) \frac{e^{-|y - h \, x(s_T)|^2}}{\pi} \, dx(s_T) \right) ds_T.
  \label{eq:pyh-1a}
\end{align}
For example, if each entry $X(s_T)$ of $A$ is CSCG with variance $P(s_T)$
then
\begin{align}
  p(y | h) = \int_{\mathbb C} \, p(s_T | h)
  \frac{\exp\left(- \frac{|y|^2}{1 + g P(s_T)}\right)}{\pi (1 + g P(s_T))} \, ds_T .
  \label{eq:pyh-1b}
\end{align}
In general, one can compute $I(A;Y|H)$ numerically by using \eqref{eq:IAYH}--\eqref{eq:pyh-1a}, but the calculations are hampered if the integrals in \eqref{eq:pyah-1}--\eqref{eq:pyh-1a} do not simplify.

For the reverse model GMI~\eqref{eq:AGMI-r3}, the averaging density in~\eqref{eq:AGMI-5-part-Es1a}--\eqref{eq:AGMI-5-part-Es2a} is here
\begin{align}
    p(h,s_T | y,s_R) = \delta(h-s_R) \,
    \frac{p(s_T|h) \, p(y | h,s_T)}{p(y|h)} .
    \label{eq:reverse-avg-density2}
\end{align}
We use numerical integration to compute the GMI.

To obtain more insight, we state the forward model rates of Theorem~\ref{theorem:AGMI-SR-max}
and Remark~\ref{remark:full-CSIR} as a Corollary.

\begin{corollary} \label{corollary:part-CSIT}
An achievable rate for the fading channels \eqref{eq:Y-Gauss}
with $S_R=H$ and partial CSIT is the forward model GMI
\begin{align}
  I_1(A;Y|H) & = \E{\log\left( 1 + \text{SNR}(H) \right)}
  \label{eq:GMI-AYH}
\end{align}
where
\begin{align}
  \text{SNR}(h) = \frac{|h|^2 \tilde P_T(h)}{1 + |h|^2 \Var{\left. \sqrt{P(S_T)} \right| H=h} }
  \label{eq:SNR-AYH}
\end{align}
and
\begin{align}
    \tilde P_T(h) = \E{ \left. \sqrt{P(S_T)} \right| H=h}^2.
    \label{eq:Pbp}
\end{align}
\end{corollary}

\begin{remark}
Jensen's inequality gives
\begin{align}
  \tilde P_T(h) \le \E{P(S_T) | H=h}
\end{align}
by the concavity of the square root. Equality holds if and only if $P(S_T)$ is a constant given $H=h$. 
\end{remark}

\begin{remark}
Choosing $P(s_T)=P$ for all $s_T $ in Corollary~\ref{corollary:part-CSIT} gives $\tilde P_T(h)=P$ for all $h$ and the rate \eqref{eq:GMI-AYH} is the capacity \eqref{eq:no-CSIT} without CSIT.
\end{remark}

\begin{remark}
For large $P$, the $\text{SNR}(h)$ in \eqref{eq:SNR-AYH} saturates unless $P(s_T)/P \rightarrow 1$ for all $s_T$, i.e., the high-SNR capacity is the same as the capacity without CSIT. The CSIT thus must become more accurate as $P$ increases to improve the rate.
\end{remark}

\begin{remark} \label{remark:P-policy-deriv-part-CSIT}
To optimize the power levels, consider~\eqref{eq:P-policy-deriv} and
\begin{align}
  & \tilde P(h)' = 2 |h|^2 \sqrt{\tilde P_T(h)} \, p(s_T | h) \\
  & \E{|Y|^2|H=h}' = 2 |h|^2 \sqrt{P(s_T)} \, p(s_T | h).
\end{align}
However, the resulting equations give little insight due to the expectation over $H$ in \eqref{eq:P-policy-deriv}. An exception is the on-off fading case where the expectation has only one term; see~\eqref{eq:oof-opc-1}--\eqref{eq:oof-opc-2}.
\end{remark}

\subsection{Partial CSIR, Full CSIT}
\label{subsec:Partial-CSIR-fullCSIT}
Suppose $S_R$ is a (perhaps noisy) function of $H$; see~\eqref{eq:H-mmse-estimate}.  The capacity is given by \eqref{eq:Shannon-cap-SR-AWGN} for which we need to compute $p(y | a, s_R)$ and $p(y | s_R)$. The GMI with a $K$-partition of the output space  $\set Y \times \set S_R$ can be helpful for these problems. We assume that the CSIR is either $S_R=0$ or $S_R=1(G \ge t)$ for some transmitter threshold $t$; see~\cite{Goldsmith-Varaiya-IT97}. 

Suppose that $S_T=H$. We then have
\begin{align*}
    p(y | a, s_R) & = \int_{\mathbb C} p(h | s_R) \frac{\exp\left( -\left| y - h \, x(h) \right|^2 \right)}{\pi} \, dh
    \\
    p(y | s_R) & = \int_{\mathbb C^2} p(h | s_R) \, p(x(h)) \nonumber \\
    & \qquad \frac{\exp\left( -\left| y - h \, x(h) \right|^2 \right)}{\pi} \, dx(h) \, dh.
\end{align*}
Now select the $X(h)$ to be jointly CSCG with variances $\E{|X(h)|^2}=P(h)$ and correlation coefficients 
\begin{align*}
  \rho(h,h') = \frac{\E{ X(h) X(h')^\text{*}}}{\sqrt{P(h) P(h')}}
\end{align*}
and where $\E{P(H)} \le P$. We then have
\begin{align*}
    p(y | s_R) = \int_{\mathbb C} p(h) \frac{e^{-|y|^2/(|h|^2 P(h) + 1)}}{2 \pi (|h|^2 P(h) + 1)} \, dh.
\end{align*}
As in \eqref{eq:py}, $p(y | s_R )$ and therefore $h(Y | S_R)$ depend only on the marginals $p(x(h))$ of $A$ and not on the $\rho(h,h')$. We thus have the problem of finding the $\rho(h,h')$ that minimize
\begin{align*}
    h(Y | S_R, A) = \int_\set{A} p(a) \, h(Y | S_R, A=a) \, da.
\end{align*}
However, we study the conventional $A$ in \eqref{eq:adaptive-codeword-conventional} for simplicity.

For the reverse model GMI~\eqref{eq:AGMI-r3}, the averaging density in~\eqref{eq:AGMI-5-part-Es1a}--\eqref{eq:AGMI-5-part-Es2a} is (cf.~\eqref{eq:reverse-avg-density2})
\begin{align}
    p(h,s_T | y,s_R) = \delta(s_T-h) \,
    \frac{p(h|s_R) \, p(y | h,s_R)}{p(y|s_R)} .
    \label{eq:reverse-avg-density3}
\end{align}
We again use numerical integration to compute the GMI.

For the forward model GMI, consider the same model and CSCG $X$ as in Theorem~\ref{theorem:AGMI-SR-max}. Since $H$ is a function of $S_T$, we use \eqref{eq:AGMI-denominator} in Remark~\ref{remark:variance} to write
\begin{align}
   & I_1(A;Y | S_R) \nonumber \\
   & = \E{ \log\left( 1 + \frac{\tilde P(S_R)}{1 + \Var{\left. \sqrt{G P(H)} \right| S_R}} \right) }
   \label{eq:R-fullCSIT}
\end{align}
where (see~\eqref{eq:AGMI-desire})
\begin{align}
   & \tilde P(s_R) = \E{ \left. \sqrt{G P(H)} \, \right| S_R=s_R}^2
   \label{eq:partialCSIT-PsR} \\
   & \E{|Y|^2 | S_R=s_R} = 1 + \E{G P(H) | S_R=s_R} .
   \label{eq:partialCSIT-EY2}
\end{align}
The transmitter compensates for the phase of $H$, and it remains to adjust the transmit power levels $P(h)$. We study five power control policies and two types of CSIR; see Table~\ref{table:pc-policies}.

\begin{table}[t!]
\centering
\caption{Power Control Policies and Minimal SNRs}
\begin{tabular}{ll|cc|}
& & \multicolumn{2}{c|}{\bf CSIR} \\
& & None: $S_R=0$ & $S_R=1(G \ge t)$ \\
\hline
\multirow{5}{*}{\bf Policy}
& TCP     & Eq.~\eqref{eq:minEbNo-fullCSIT-noCSIR-TCP}    & Eq.~\eqref{eq:minEbNo-fullCSIT-CSIR-TCP}  \\
& TMF     & Eq.~\eqref{eq:minEbNo-fullCSIT-noCSIR-TMF}   &  Eq.~\eqref{eq:minEbNo-fullCSIT-CSIR-TMF}    \\
& TCI       & Eq.~\eqref{eq:minEbNo-fullCSIT-noCSIR-TCI}    &  Eq.~\eqref{eq:minEbNo-fullCSIT-CSIR-TCI}  \\
& GMI-Optimal & \multicolumn{2}{c|}{see Theorem~\ref{theorem:Partial-CSIR-Full-CSIT}}  \\
& TMMSE  & \multicolumn{2}{c|}{see Remark~\ref{remark:TMMSE}}  \\
\hline
\end{tabular}
\label{table:pc-policies}
\end{table}

\subsubsection{Heuristic Policies}
The first three policies are reasonable heuristics and have the form
\begin{align}
   P(h) = \left\{ \begin{array}{ll}
     \hat P \, g^a, & g \ge t \\ 0, & \text{else}
   \end{array} \right.
   \label{eq:heuristic-PC-policies}
\end{align}
for some choice of real $a$ and where
\begin{align}
   \hat P= \frac{P}{\int_{t}^\infty p(g) \, g^a \, dg}.
   \label{eq:heuristic-PC-policies-Phat}
\end{align}
In particular, choosing $a=0,+1,-1$, we obtain policies that we call truncated constant power (TCP), truncated matched filtering (TMF), and truncated channel inversion (TCI), respectively; see~\cite[p.~487]{Keshet-Steinberg-Merhav-08}, \cite{Goldsmith-Varaiya-IT97}. For such policies, we compute
\begin{align}
   & \tilde P(s_R) = \hat P \left( \int_{t}^\infty p(g | s_R) \, \sqrt{g^{1+a}} \, dg \right)^2 \label{eq:Pt-fullCSIT} \\
   & \E{G P(H) | S_R=s_R} = \hat P \, \int_{t}^\infty p(g | s_R) \, g^{1+a} \, dg .
   \label{eq:EGP-fullCSIT}
\end{align}
These policies all have the form $P(h)=P \cdot f(h)$ for some function $f(.)$ that is independent of $P$. The minimum SNR in \eqref{eq:wideband} with $C(P)$ replaced with the GMI is thus
\begin{align}
   \left.\frac{E_b}{N_0}\right|_{\text{min}}
   & = \frac{\left(\int_{t}^\infty p(g) \, g^a \, dg \right) \log 2}{\E{ \left( \int_{t}^\infty p(g | S_R) \, \sqrt{g^{1+a}} \, dg \right)^2 }} .
   \label{eq:minEbNo-fullCSIT}
\end{align}

For instance, consider the threshold $t=0$ (no truncation). The TCP ($a=0$) and TMF ($a=1$) policies have $\hat P = P$ while TCI ($a=-1$) has $P = \hat P / \E{G^{-1}}$. For TCP, TMF, and TCI, we compute the respective
\begin{align}
   & \left.\frac{E_b}{N_0}\right|_{\text{min}} = \frac{\log 2}{\E{ \E{\left. \sqrt{G} \,\right| S_R}^2 }}
   \label{eq:minEbNo-fullCSIT-TCP} \\
   & \left.\frac{E_b}{N_0}\right|_{\text{min}} = \frac{\log 2}{\E{ \E{G|S_R}^2 }}
   \label{eq:minEbNo-fullCSIT-TMF} \\
   & \left.\frac{E_b}{N_0}\right|_{\text{min}} = \E{G^{-1}} \log 2 .
   \label{eq:minEbNo-fullCSIT-TCI}
\end{align}
Applying Jensen's inequality to the square root, square, and inverse functions in  \eqref{eq:minEbNo-fullCSIT-TCP}--\eqref{eq:minEbNo-fullCSIT-TCI}, we find that for $t=0$:
\begin{itemize}
\item the minimum $E_b/N_0$ of TCP and TCI is larger (worse) than $-1.59$ dB unless there is no fading;
\item the minimum $E_b/N_0$ of TMF is smaller (better) than $-1.59$ dB unless $\E{G|S_R}=\E{G}=1$.
\end{itemize}
However, we emphasize that these claims apply to the GMI and not necessarily the mutual information; see Sec.~\ref{subsec:Rayleigh-partial-CSIR-full-CSIT} and Figs.~\ref{fig:rf4}--\ref{fig:rf5}.

\subsubsection{GMI-Optimal Policy}
The fourth policy is optimal for the GMI \eqref{eq:R-fullCSIT} and has the form of an MMSE precoder. This policy motivates a  truncated MMSE (TMMSE) policy that generalizes and improves TMF and TCI.

Taking the derivative of the Lagrangian
\begin{align}
   I_1(A;Y|S_R) + \lambda \left(P - \E{P(H)}\right)
   \label{eq:Lagrangian2}
\end{align}
with respect to $P(h)$ we have the following result.

\begin{theorem} \label{theorem:Partial-CSIR-Full-CSIT}
The optimal power control policy for the GMI $I_1(A;Y|S_R)$ for the fading channels \eqref{eq:Y-Gauss} with $S_T=H$ is
\begin{align}
  \sqrt{P(h)} & = \frac{\alpha(h)  |h|}{\lambda + \beta(h) |h|^2}
  \label{eq:P-policy}
\end{align}
where $\lambda>0$ is chosen so that $\E{P(H)}=P$ and
\begin{align}
  & \alpha(h) = \int_{\mathbb C} p(s_R | h) \frac{\sqrt{\tilde P(s_R)}}{\E{|Y|^2|S_R=s_R}-\tilde P(s_R)} \, ds_R
  \label{eq:P-policy-alpha} \\
  & \beta(h) = \int_{\mathbb C} p(s_R | h) \nonumber \\
  & \cdot \frac{\tilde P(s_R)}{\left[ \E{|Y|^2|S_R=s_R} - \tilde P(s_R) \right] \E{|Y|^2|S_R=s_R} } \, ds_R .
  \label{eq:P-policy-beta}
\end{align}
\end{theorem}
\begin{proof}
Apply \eqref{eq:P-policy-deriv} with \eqref{eq:partialCSIT-PsR}--\eqref{eq:partialCSIT-EY2} to obtain
\begin{align}
  & \tilde P(s_R)' = 2 |h| \sqrt{\tilde P(s_R)} \, p(h | s_R) \\
  & \E{|Y|^2|S_R=s_R}' = 2 |h|^2 \sqrt{P(h)} \, p(h | s_R).
\end{align}
Inserting into \eqref{eq:P-policy-deriv} and rearranging terms we obtain \eqref{eq:P-policy} with \eqref{eq:P-policy-alpha} and \eqref{eq:P-policy-beta}.
\end{proof}

\begin{remark} \label{remark:optimal-policy}
The expressions \eqref{eq:P-policy-alpha} and \eqref{eq:P-policy-beta} are self-referencing, as $\tilde P(s_R)$ itself depends on $\alpha(h)$ and $\beta(h)$. However, one simplification occurs if $S_R$ is a function of $H$: $\alpha(h)$ and $\beta(h)$ are functions of $s_R$ only since the $p(s_R|h)$ in \eqref{eq:P-policy-alpha}--\eqref{eq:P-policy-beta} is a Dirac generalized function.
\end{remark}

\begin{remark}
Consider the expression \eqref{eq:P-policy}. We effectively have a matched filter for small $|h|$; for large $|h|$, we effectively have a channel inversion. Recall that LMMSE filtering has similar behavior for low and high SNR, respectively.
\end{remark}

\begin{remark} \label{remark:TMMSE}
A heuristic based on the optimal policy is a TMMSE policy where the transmitter sets $P(h)=0$ if $G<t$, and otherwise uses~\eqref{eq:P-policy} but where $\alpha(h)$, $\beta(h)$ are independent of $h$. There are thus four parameters to optimize: $\lambda$, $\alpha$, $\beta$, and $t$. This TMMSE policy will outperform TMF and TCI in general, as these are special cases where $\beta=0$ and $\lambda=0$, respectively.
\end{remark}

\subsubsection{$S_R=0$}
For this CSIR, the GMI~\eqref{eq:R-fullCSIT} simplifies to $I_1(A;Y)$ and the heuristic policy (TCP, TMF, TCI) rates are
\begin{align}
   I_1(A;Y) = \log\left( 1 + \frac{\hat P\, \E{\sqrt{G^{1+a}} \cdot 1(G \ge t)}^2}
   {1 + \hat P\, \Var{\sqrt{G^{1+a}} \cdot 1(G \ge t)}} \right).
   \label{eq:R-fullCSIT-noCSIR}
\end{align}
Moreover, the expression \eqref{eq:minEbNo-fullCSIT} gives
\begin{align}
   \left.\frac{E_b}{N_0}\right|_{\text{min}}
   & = \frac{\E{G^a \cdot 1(G \ge t)}}
   {\E{\sqrt{G^{1+a}} \cdot 1(G \ge t)}^2} \log 2 .
   \label{eq:minEbNo-fullCSIT-noCSIR}
\end{align}

For TCP, TMF, and TCI, we compute the respective
\begin{align}
   & \left.\frac{E_b}{N_0}\right|_{\text{min}} = \frac{\log 2}{\Pr{G \ge t} \E{\left. \sqrt{G}\, \right| G \ge t}^2}
   \label{eq:minEbNo-fullCSIT-noCSIR-TCP} \\
   & \left.\frac{E_b}{N_0}\right|_{\text{min}} = \frac{\log 2}{\int_{t}^\infty p(g) \, g \, dg}
   \label{eq:minEbNo-fullCSIT-noCSIR-TMF} \\
   & \left. \frac{E_b}{N_0}\right|_{\text{min}} = \frac{\E{\left. G^{-1} \right| G \ge t} }{\Pr{G \ge t}} \log 2 .
   \label{eq:minEbNo-fullCSIT-noCSIR-TCI}
\end{align}
Again applying Jensen's inequality to the various functions in~\eqref{eq:minEbNo-fullCSIT-noCSIR-TCP}--\eqref{eq:minEbNo-fullCSIT-noCSIR-TCI}, we find that:
\begin{itemize}
\item the minimum $E_b/N_0$ of TMF is smaller (better) than that of TCP and TCI unless there is no fading, or if the minimal $E_b/N_0$ is $-\infty$;
\item the best threshold for TMF is $t=0$ and the minimal $E_b/N_0$ is $-1.59$ dB.
\end{itemize}
For the optimal policy, the parameters $\alpha(h)$ and $\beta(h)$ in \eqref{eq:P-policy-alpha}--\eqref{eq:P-policy-beta} are constants independent of $h$, see Remark~\ref{remark:optimal-policy}, and the TMMSE policy with $t=0$ is the GMI-optimal policy.

\begin{remark} \label{remark:TCI-densities}
The TCI channel densities are
\begin{align*}
    & p(y | a) = \Pr{G<t} \frac{e^{-| y |^2}}{\pi} + \Pr{G\ge t} \frac{e^{-\left| y - \sqrt{\hat P} \, u \right|^2}}{\pi}
    \\
    & p(y) = \Pr{G<t} \frac{e^{-| y |^2}}{\pi} + \Pr{G\ge t} \frac{e^{-| y |^2 / (1+\hat P)}}{\pi(1+\hat P)} .
\end{align*}
\end{remark}

\begin{remark} \label{remark:high-SNR-oof}
At high SNR, one might expect that the receiver can estimate $P(S_T)$ precisely even if $S_R=0$. We show that this is indeed the case for on-off fading by using the $K=2$ partition~\eqref{eq:AGMI-part2} of Remark~\ref{remark:AGMI-part2}. Moreover, the results prove that at high SNR one can approach $I(A;Y)$; see Sec.~\ref{subsec:oof-Partial-CSIR-full-CSIT}.
\end{remark}

\begin{remark} \label{remark:high-SNR-Rayleigh}
For Rayleigh fading, the GMI with $K=2$ in~\eqref{eq:AGMI-part2} is helpful for both high and low SNR. For instance, for $S_R=0$ and TCI, the $K=2$ GMI approaches the mutual information for $S_R=1(G\ge t)$ as the SNR increases; see Remark~\ref{remark:TCI-high-SNR} in Sec.~\ref{subsec:Rayleigh-partial-CSIR-full-CSIT}. We further show that for $S_R=0$, the TCI policy can achieve a minimal $E_b/N_0$ of $-\infty$ dB, see~\eqref{eq:Rayleigh-EbNo-scaling} in Sec.~\ref{subsec:Rayleigh-partial-CSIR-full-CSIT}.
\end{remark}

\subsubsection{$S_R=1(G \ge t)$}
The heuristic policy rates are now (cf.~\eqref{eq:R-fullCSIT-noCSIR} and note the $\Pr{G \ge t}$ term and conditioning)
\begin{align}
   & I_1(A;Y|S_R) \nonumber \\
   & = \Pr{G \ge t} \log\left( 1 + \frac{\hat P \, \E{\left. \sqrt{G^{1+a}} \right| G \ge t}^2}
   {1 + \hat P \, \Var{\left. \sqrt{G^{1+a}} \right| G \ge t}} \right).
   \label{eq:R-fullCSIT-CSIR}
\end{align}
Moreover, the expression \eqref{eq:minEbNo-fullCSIT} is
\begin{align}
   & \left.\frac{E_b}{N_0}\right|_{\text{min}} 
   = \frac{\E{\left. G^a \right| G \ge t}}{\E{\left. \sqrt{G^{1+a}} \right| G \ge t}^2} \log 2 .
   \label{eq:minEbNo-fullCSIT-CSIR} 
\end{align}

For TCP, TMF, and TCI we compute the respective
\begin{align}
   & \left.\frac{E_b}{N_0}\right|_{\text{min}} = \frac{\log 2}{\E{\left. \sqrt{G} \right| G \ge t}^2}
   \label{eq:minEbNo-fullCSIT-CSIR-TCP} \\
   & \left.\frac{E_b}{N_0}\right|_{\text{min}} = \frac{\log 2}{\E{G | G \ge t}}
   \label{eq:minEbNo-fullCSIT-CSIR-TMF} \\
   & \left. \frac{E_b}{N_0}\right|_{\text{min}} = \E{\left. G^{-1} \right| G \ge t} \log 2 .
   \label{eq:minEbNo-fullCSIT-CSIR-TCI}
\end{align}
Again applying Jensen's inequality to the various functions in~\eqref{eq:minEbNo-fullCSIT-CSIR-TCP}--\eqref{eq:minEbNo-fullCSIT-CSIR-TCI}, we find that:
\begin{itemize}
\item the minimum $E_b/N_0$ of all policies can be better than $-1.59$ dB by choosing $t>0$;
\item the minimum $E_b/N_0$ of TMF is smaller (better) than that of TCP and TCI unless there is no fading or the minimal $E_b/N_0$ is $-\infty$.
\end{itemize}

For the optimal policy, Remark~\ref{remark:optimal-policy} points out that $\alpha(h)$ and $\beta(h)$ depend on $s_R$ only. We compute
\begin{align}
  \sqrt{P(h)} = \left\{ \begin{array}{ll}
   \dfrac{\alpha_0 \, |h|}{\lambda + \beta_0 |h|^2}, & g < t \vspace{0.2cm} \\
   \dfrac{\alpha_1 \, |h|}{\lambda + \beta_1 |h|^2}, & g \ge t
   \end{array} \right.
\end{align}
where for $s_R \in \{0,1\}$ we have
\begin{align*}
  & \alpha_{s_R} = \dfrac{\sqrt{\tilde P(s_R)}}{\E{|Y|^2|S_R=s_R}-\tilde P(s_R)} \\
  & \beta_{s_R} = \dfrac{\sqrt{\tilde P(s_R)}}{\left[\E{|Y|^2|S_R=s_R}-\tilde P(s_R)\right] \E{|Y|^2|S_R=s_R}}.
\end{align*}

\begin{remark} \label{remark:TCI-MI}
The GMI \eqref{eq:R-fullCSIT-CSIR} for TCI ($a=-1$) is the mutual information $I(A;Y|S_R)$.
To see this, observe that the model $q(y|a,s_R)$ has
\begin{align*}
   q(y|a,0) = \frac{e^{-| y |^2}}{\pi}, \quad
   q(y|a,1) = \frac{e^{-\left| y - \sqrt{\hat P} \, u \right|^2}}{\pi}
\end{align*}
and thus we have $q(y|a,s_R)=p(y|a,s_R)$ for all $y,a,s_R$.
\end{remark}

\subsection{Partial CSIR, CSIT\at R}
\label{subsec:Partial-CSIR-CSITatR}
Suppose next that $S_R$ is a noisy function of $H$ (see for instance~\eqref{eq:H-mmse-estimate}) and $S_T=f_T(S_R)$. The capacity is given by \eqref{eq:Shannon-cap1} and we compute
\begin{align}
  I(X;Y|S_R) = \E{\log \frac{p(Y|X,S_R)}{p(Y|S_R)}}
  \label{eq:IXYH}
\end{align}
where writing $s_T=f_T(s_R)$ we have
\begin{align}
  & p(y | s_R, x) = \int_{\mathbb C} p(h | s_R) \frac{e^{-|y -h \, x(s_T)|^2}}{\pi} \, dh
  \label{eq:pysRx} \\
  & p(y | s_R) = \int_{{\mathbb C}^2} p(h | s_R) \, p(x(s_T)) \, \frac{e^{-|y - h \, x(s_T)|^2}}{\pi} \, dx(s_T) \, dh .
  \label{eq:pysR}
\end{align}
For example, if $X(s_T)$ is CSCG with variance $P(s_T)$ then 
\begin{align}
  p(y | s_R) = \int_{\mathbb C} p(h | s_R) \frac{\exp\left(- \frac{|y|^2}{1 + |h|^2 P(s_T)}\right)}{\pi (1 + |h|^2 P(s_T))} \, dh .
  \label{eq:pysR-2}
\end{align}
One can compute $I(X;Y|S_R)$ numerically using \eqref{eq:pysRx}--\eqref{eq:pysR}. However, optimizing over $X(s_T)$ is usually difficult.

For the reverse model GMI~\eqref{eq:AGMI-r3}, the averaging density in~\eqref{eq:AGMI-5-part-Es1a}--\eqref{eq:AGMI-5-part-Es2a} is now (cf.~\eqref{eq:reverse-avg-density2} and~\eqref{eq:reverse-avg-density3})
\begin{align}
    p(h,s_T | y,s_R) = \delta\left(s_T - f_T(s_R)\right) \,
    \frac{p(h|s_R) \, p(y | h,s_R)}{p(y|s_R)} .
    \label{eq:reverse-avg-density4}
\end{align}
We use numerical integration to compute the rates.

The forward model GMI again gives more insight. Define the channel gain and variance as the respective
\begin{align}
   \tilde g(s_R) & = \left|\E{H | S_R = s_R}\right|^2 \label{eq:partial-CSIR-gt} \\
   \tilde \sigma^2(s_R) & = \Var{H | S_R = s_R}. \label{eq:partial-CSIR-sigmat} 
\end{align}

\begin{theorem} \label{theorem:partial-CSIR}
An achievable rate for AWGN fading channels \eqref{eq:Y-Gauss} with power constraint $\E{|X|^2}\le P$ and with partial CSIR $S_R$ and $S_T=f_T(S_R)$ is
\begin{align}
  I_1(X;Y | S_R)
  = \E{ \log\left( 1 + \frac{\tilde g(S_R) P(S_T)}
  {1 + \tilde \sigma^2(S_R) P(S_T)} \right) }
  \label{eq:GMI-XYH}
\end{align}
where $\E{P(S_T)}=P$. The optimal power levels $P(s_T)$ are obtained by solving
\begin{align}
   & \lambda = \int_{\mathbb R} p(s_R | s_T) \nonumber \\
   & \cdot
   \frac{\tilde g(s_R)}{\left[1 + \left( \tilde g(s_R) + \tilde \sigma^2(s_R) \right) P(s_T) \right]
   \left[1+ \tilde \sigma^2(s_R) P(s_T) \right]} \, ds_R .
   \label{eq:lambda-part-CSIR}
\end{align}
In particular, if $S_T$ determines $S_R$ (CSIR\at T) then we have the quadratic waterfilling expression
\begin{align}
    f\left( P(s_T),\, \tilde g(s_R),\, \tilde \sigma^2(s_R) \right) = \left( \frac{1}{\lambda} - \frac{1}{\tilde g(s_R)} \right)^+
   \label{eq:waterfilling-noCSIR}
\end{align}
where
\begin{align}
  f\left( Q,g,\sigma^2 \right) = \left(1 + 2 \frac{\sigma^2}{g} \right) Q
  +  \left( 1 + \frac{\sigma^2}{g} \right ) \sigma^2 Q^2
\end{align}
and where $\lambda$ is chosen so that $\E{P(H_R)}=P$. 
\end{theorem}
\begin{proof}
Apply Theorem~\ref{theorem:AGMI-SR-max} with
\begin{align}
  & \tilde P(s_R) = \tilde g(s_R) P(s_T) \\
  & \E{|Y|^2|S_R=s_R} = 1 + \left( \tilde g(s_R) + \tilde \sigma^2 (s_R) \right) P(s_T)
\end{align}
to obtain \eqref{eq:GMI-XYH}. To optimize the power levels $P(s_T)$ with \eqref{eq:P-policy-deriv}, consider the derivatives
\begin{align}
  & \tilde P(s_R)' = 2 \tilde g(s_R) \sqrt{P(s_T)} 1(s_T=f_T(s_R)) \\
  & \E{|Y|^2|S_R=s_R}' \nonumber \\
  & \quad = 2 \left( \tilde g(s_R) + \tilde \sigma^2(s_R) \right) \sqrt{P(s_T)} 1(s_T=f_T(s_R)).
\end{align}
The expression \eqref{eq:P-policy-deriv} thus becomes \eqref{eq:lambda-part-CSIR}. If $S_T$ determines $S_R$ then the expression simplifies to
\begin{align*}
   \lambda =
   \frac{\tilde g(s_R)}{\left[1 + \left( \tilde g(s_R) + \tilde \sigma^2(s_R) \right) P(s_T) \right]
   \left[1+ \tilde \sigma^2(s_R) P(s_T) \right]}
\end{align*}
from which we obtain \eqref{eq:waterfilling-noCSIR}.
\end{proof}

\begin{remark}
The optimal power control policy with CSIT\at R and CSIR\at T can be written explicitly by solving the quadratic in \eqref{eq:waterfilling-noCSIR}. The result is that $P(s_T)$ is
\begin{align}
   & \frac{\tilde g + 2 \tilde \sigma^2}{2 \tilde \sigma^2 (\tilde g + \tilde \sigma^2)} 
   \left[ \sqrt{1 + 4 \tilde \sigma^2 \left(\frac{1}{\lambda} - \frac{1}{\tilde g} \right)^+
   \frac{\tilde g \, (\tilde g + \tilde \sigma^2)}{(\tilde g + 2 \tilde \sigma^2)^2}} - 1 \right]
   \label{eq:waterfilling-noCSIR2}
\end{align}
where we have discarded the dependence on $s_R$ for convenience. The alternative form \eqref{eq:waterfilling-noCSIR} relates to the usual waterfilling where the left-hand side of \eqref{eq:waterfilling-noCSIR} is $P(s_T)$. Observe that $\tilde \sigma^2=0$ gives conventional waterfilling.
\end{remark}

\begin{remark} \label{remark:Partial-CSIR-CSITatR-2}
As in Sec.~\ref{subsec:GMI-oof}, we show that at high SNR the $K=2$ GMI of Remark~\ref{remark:AGMI-part3} approaches the upper bound of Proposition~\ref{proposition:part-CSIT2} in some cases; see Sec.~\ref{subsec:oof-Partial-CSIR-CSITatR}. The channel parameters depend on $s_R$, and we choose $h_1(s_R)=0$ and $\sigma_1^2(s_R)=\sigma_2^2(s_R)=1$ for all $s_R$.
\end{remark}

\section{On-Off Fading}
\label{sec:oof}
Consider again on-off fading with $P_G(0)=P_G(2)=1/2$. We study the scenarios listed in Table~\ref{table:model-classes}. The case of no CIR and no CSIT was studied in Sec.~\ref{subsec:GMI-oof}.

\subsection{Full CSIR, CSIT\at R}
Consider $S_R=H$. The capacity with $B=0$ (no CSIT) is given by \eqref{eq:no-CSIT} (cf.~\eqref{eq:GMI-2-avg-oof}):
\begin{align}
   C(P) = \frac{1}{2} \log\left( 1 + 2P \right)
\end{align}
and the wideband values are given by~\eqref{eq:wideband-fullCSIR-noCSIT} (cf.~\eqref{eq:wideband-fullCSIR-noCSIT-oof}); the minimal $E_b/N_0$ is $\log 2$ and the slope is $S=1$.

The capacity with $B=\infty$ (or $S_T=G$) increases to
\begin{align}
   C(P) & = \frac{1}{2} \log\left( 1 + 4P \right)
   \label{eq:oof-cap2}
\end{align}
where $P(0)=0$ and $P(2)=2P$. This capacity is also achieved with $B=1$ since there are only two values for $G$. We compute $C'(0)=2$ and $C''(0)=-8$, and therefore
\begin{align}
   \left.\frac{E_b}{N_0}\right|_{\text{min}} = \frac{\log 2}{2}, \quad S = 1.
\end{align}
The power gain due to CSIT compared to no fading is thus 3.01 dB, but the capacity slope is the same. The rate curves are compared in Fig.~\ref{fig:oof2}.

\begin{figure}[t]
      \centering
      \includegraphics[width=0.95\columnwidth]{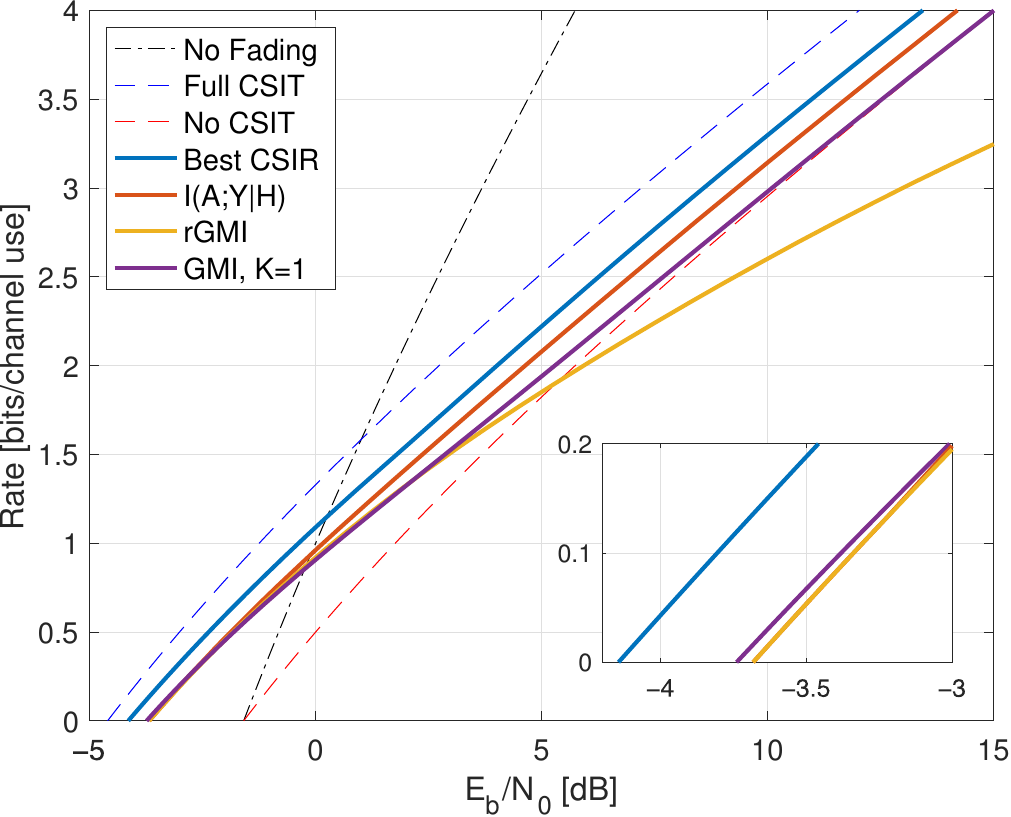}
      \caption{Rates for on-off fading with full CSIR and partial CSIT with noise parameter $\eps=0.1$. The curve ``Best CSIR'' shows the capacity with $S_R =H \sqrt{P(S_T)}$. The curves for $I(A;Y|H)$, the reverse model GMI (rGMI), and the forward model GMI (GMI, K=1) are for $S_R=H$ with CSCG inputs $X(s_T)$. The $I(A;Y|H)$ and rGMI curves are indistinguishable in the inset.}
      \label{fig:oof2}
\end{figure}

\subsection{Full CSIR, Partial CSIT}
\label{subsec:oof-full-CSIR-partial-CSIT}
Consider next noisy CSIT with $0\le\epsilon\le\frac{1}{2}$ and
\begin{align*}
  \Pr{S_T = G} = \epsb, \quad
  \Pr{S_T \ne G} = \epsilon .
\end{align*}

\subsubsection{$S_R = H\sqrt{P(S_T)}$}
The capacity $C(P)$ of Proposition~\ref{proposition:part-CSIT2} is
\begin{align}
  \max_{P(0)+P(2)=2P} \frac{\eps}{2} \log\left( 1 +  2 P(0) \right) + \frac{\epsb}{2} \log\left( 1 +  2 P(2) \right).
  \label{eq:GMI-AYH2-oof} 
\end{align}
Optimizing the power levels, we have
\begin{align}
   P(0) = \left( 2 \eps \, P - \frac{\epsb-\eps}{2}\right)^+, \quad
   P(2) = 2P-P(0).
   \label{eq:powers-fullCSIR-parCSIT}
\end{align}
Fig.~\ref{fig:oof2} shows $C(P)$ for $\eps=0.1$ as the curve labeled ``Best CSIR''. For $P \ge (\epsb-\eps)/(4 \eps)$, we compute
\begin{align}
  C(P) = \frac{1}{2} \log(1 + 2P) + \frac{1}{2} [1 - H_2(\eps)] \log 2
  \label{eq:oof-C-high-SNR}
\end{align}
where $H_2(\eps) = -\eps\log_2\eps - \epsb\log_2\epsb$ is the binary entropy function. For example, if $\epsilon=0.1$ then for $P \ge 2$ one gains $\Delta C=[1-H_2(0.1)]/2 \approx 0.27$ 
bits over the capacity without CSIT. This translates to an SNR gain of $2\Delta C \cdot 10 \log_{10}(2) \approx 1.60$ dB. 
On the other hand, for $P \le (\epsb-\eps)/(4 \eps)$ we have $P(0)=0$, $P(2)=2P$, and the capacity is
\begin{align}
  C(P) = \frac{\epsb}{2} \log\left( 1 +  4 P \right).
  \label{eq:GMI-AYH2-oof-2} 
\end{align}
We have $C'(0) = 2 \, \epsb$ and lose a fraction of $\epsb$ of the power as compared to having full CSIT ($\epsilon=0$). For example, if $\epsilon=0.1$, the minimal $E_b/N_0$ is approximately $-4.14$ dB.

\subsubsection{$S_R = H$}
To compute $I(A;Y|H)$ in \eqref{eq:IAYH}, we write
\eqref{eq:pyah-1} and \eqref{eq:pyh-1b} for CSCG $X(s_T)$ as
\begin{align*}
  & p_{Y|A,H}(y|a,0) = p_{Y|H}(y|0) = \frac{e^{-|y|^2}}{\pi} \\
  & p_{Y|A,H}\big(y | a, \sqrt{2}\,\big) = \eps\, \frac{e^{-\left|y - \sqrt{2} x(0)\right|^2}}{\pi} + \epsb\, \frac{e^{-\left|y - \sqrt{2} x\big(\sqrt{2}\,\big)\right|^2}}{\pi}
  \\
  & p_{Y|H}\big(y | \sqrt{2}\,\big) = \epsilon \frac{\exp\left(- \frac{|y|^2}{1 + 2 P(0)}\right)}{\pi (1 + 2 P(0))}
  + \epsb\,\frac{\exp\left(- \frac{|y|^2}{1 + 2 P(2)}\right)}{\pi (1 + 2 P(2))} .
\end{align*}
Fig.~\ref{fig:oof2} shows the rates as the curve labeled ``$I(A;Y|H)$''. This curve was computed by Monte Carlo integration with $P(0)=0.1\cdot P$ and $P(2)=1.9 \cdot P$, which is near-optimal for the range of SNRs depicted.

The reverse model GMI~\eqref{eq:AGMI-r3} requires $\Var{U|Y,H}$. We show how to compute this variance in Appendix~\ref{appendix:csos-2} by applying~\eqref{eq:AGMI-5-part-Es1a}--\eqref{eq:AGMI-5-part-Es2a}. Fig.~\ref{fig:oof2} shows the GMIs as the curve labeled ``rGMI'', where we used the same power levels as for the $I(A;Y|H)$ curve. The two curves are indistinguishable for small $P$, but the ``rGMI'' rates are poor at large $P$. This example shows that the forward model GMI with optimized powers can be substantially better than the reverse model GMI with a reasonable but suboptimal power policy.

The forward model GMI \eqref{eq:GMI-AYH} is
\begin{align}
  & I_1(A;Y|H) = \frac{1}{2} \log\left( 1 + \text{SNR}\left(\sqrt{2}\right) \right)
\end{align}
where $\text{SNR}\left(\sqrt{2}\right)$ is given by \eqref{eq:SNR-AYH} with
\begin{align*}
  & \tilde P_T\left(\sqrt{2}\right) = \left( \eps\sqrt{P(0)} + \epsb\sqrt{P(2)} \right)^2
  \\
  & \Var{\left. \sqrt{P(S_T)} \right| H=h} = 1 + 2 \, \eps \, \epsb \left( \sqrt{P(2)} - \sqrt{P(0)} \right)^2 .
\end{align*}
Applying Remark \ref{remark:P-policy-deriv-part-CSIT}, the optimal power control policy is
\begin{align}
   \sqrt{P(s_T)}
   & = \frac{p_{H|S_T}\left(\sqrt{2}\,\big| s_T\right)}{\gamma + \beta \, p_{H|S_T}\left(\sqrt{2}\,\big|s_T\right)} \nonumber \\
   & = \left\{ \begin{array}{ll}
   \dfrac{\eps}{\gamma + \beta \, \eps}, & s_T=0 \vspace{0.2cm} \\
   \dfrac{\epsb}{\gamma + \beta \, \epsb}, & s_T=2
   \end{array} \right. 
   \label{eq:oof-opc-1}
\end{align}
where
\begin{align}
   \beta & =  \frac{2 \sqrt{\tilde P_T\left(\sqrt{2}\,\right)}}{\E{|Y|^2 | H=\sqrt{2}\, }}
   \label{eq:oof-opc-2}
\end{align}
and $\gamma \ge 0$ is chosen so that $P(0)+P(2)=2P$. Fig.~\ref{fig:oof2} shows the resulting GMI as the curve labeled ``GMI, K=1''. At low SNR, we achieve the rate $\tilde P_T\left(\sqrt{2}\,\right)$ and the optimal power control has $\beta\rightarrow0$ so that
\begin{align}
   & P(0) = \frac{2 P \eps^2}{\eps^2 + \epsb^2}, \quad
   P(2) = \frac{2 P \epsb^2}{\eps^2 + \epsb^2}
\end{align}
and therefore
\begin{align}
   & \tilde P_T(\sqrt{2}) = 2 \left( \eps^2 + \epsb^2 \right) P.
\end{align}
We have $C'(0)=2\left( \eps^2 + \epsb^2 \right)$ and lose a fraction of $(\eps^2+\epsb^2)$ of the power as compared to having full CSIT ($\epsilon=0$). For example, if $\epsilon=0.1$, the minimal $E_b/N_0$ is approximately $-3.74$ dB.

We remark that the $I(A;Y|H)$ and reverse model GMI curves lie above the forward model curve if we choose the same power policy as for the forward channel.

\subsection{Partial CSIR, Full CSIT}
\label{subsec:oof-Partial-CSIR-full-CSIT}
This section studies $S_T=H$. The capacity with partial CSIR is given by \eqref{eq:Shannon-cap-SR} for which we need to compute $p(y | a, s_R)$ and $p(y | s_R)$. We consider two cases.

\subsubsection{$S_R=1(G \ge t)$}
Here we recover the case with full CSIR by choosing $t$ to satisfy $0<t\le 2$.

\subsubsection{$S_R=0$}
The best power policy clearly has $P(0)=0$ and $P(\sqrt{2}\,)=2P$. The mutual information is thus $I(A;Y)=I\big(X\big(\sqrt{2}\big);Y\big)$ and the channel densities are
(cf.~\eqref{eq:GMI-example1} and \eqref{eq:GMI-example2})
\begin{align*}
    & p(y | a) = \frac{e^{-| y |^2}}{2\pi} + \frac{e^{-\left| y - 2 \sqrt{P} u\left(\sqrt{2}\,\right) \right|^2}}{2\pi} \\
    & p(y) = \frac{e^{-| y |^2}}{2\pi} + \frac{e^{-| y |^2 / (1+4P)}}{2\pi(1+4P)} .
\end{align*}
The rates $I(A;Y)$ are shown in Fig.~\ref{fig:oof3}. Observe that the low-SNR rates are larger than without fading; this is a consequence of the slightly bursty nature of transmission.

\begin{figure}[t]
      \centering
      \includegraphics[width=0.95\columnwidth]{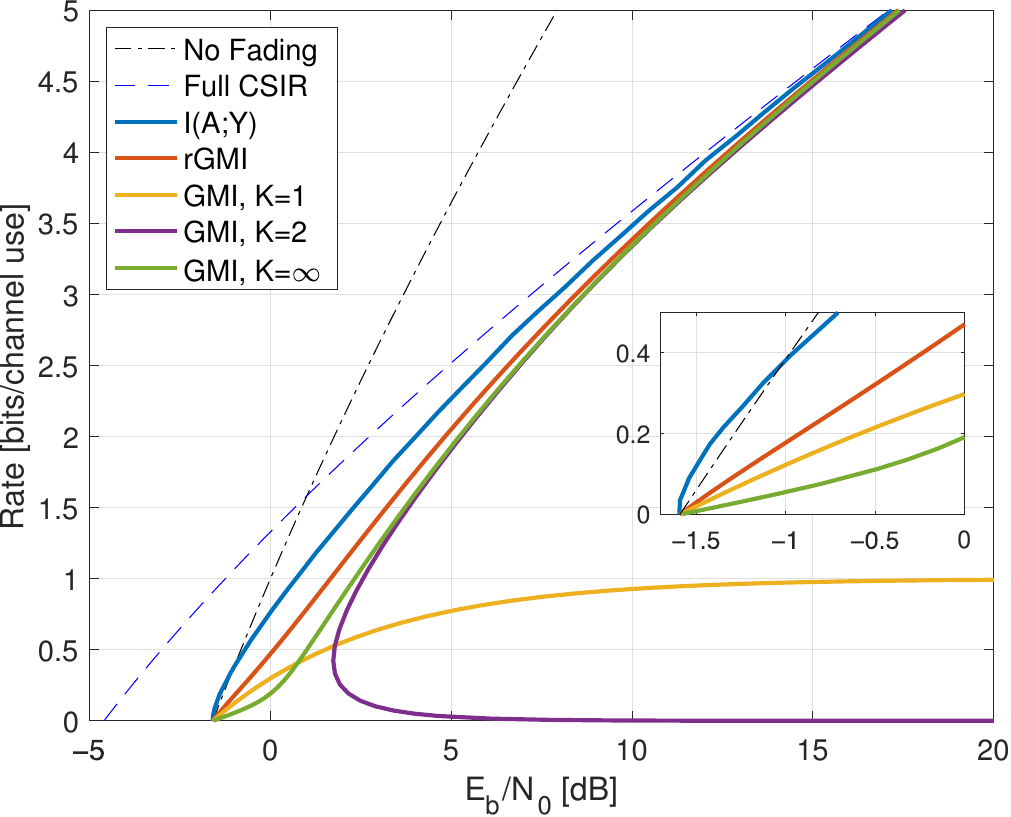}
      \caption{Rates for on-off fading with $S_T=H$ and $S_R=0$.
      The GMI for the $K=2$ partition uses the threshold $t_R=\sqrt{P}+3$.}
      \label{fig:oof3}
\end{figure}

The reverse model GMI~\eqref{eq:AGMI-r3} requires $\Var{U|Y}$. We compute this variance in Appendix~\ref{appendix:csos-3} by using~\eqref{eq:AGMI-5-part-Es1a}--\eqref{eq:AGMI-5-part-Es2a} with~\eqref{eq:reverse-avg-density3} and $\phi(s_T)=0$. Fig.~\ref{fig:oof3} shows the GMIs as the curve labeled ``rGMI''.

Next, the TCP, TMF, TCI, and TMMSE policies are the same for $0<t\le 2$, since they use $P(0)=0$ and $P\left(\sqrt{2}\,\right)=2P$. The resulting rate is given by \eqref{eq:R-fullCSIT}--\eqref{eq:partialCSIT-EY2} with $\tilde P(0) = 0$, $\tilde P(1) = P$, and $\Var{\left. \sqrt{G P(S_T)} \right| S_R=1}=P$ and
\begin{align}
   & I_1(A;Y) = \log\left( 1 + \frac{P}{1 + P} \right).
   \label{eq:AGMI-oof-MF}
\end{align}
The rates are plotted in Fig.~\ref{fig:oof3} as the curve labeled ``GMI, K=1''. This example again shows that choosing $K=1$ is a poor choice at high SNR.

To improve the auxiliary model at high SNR, consider the GMI~\eqref{eq:AGMI-part2} with $K=2$ and the subsets \eqref{eq:2partition}. We further choose the parameters $h_1=0$, $\sigma_1^2=0$, $h_2=2$, $\sigma_2^2=1$, and  adaptive coding with $X(0)=0$, $X\left(\sqrt{2}\,\right)=\sqrt{2P}\,U$, $\bar{X}=\sqrt{P}\,U$, where $U\sim\mathcal{CN}(0,1)$. The GMI \eqref{eq:AGMI-part2} is
\begin{align}
   I_1(A;Y) = \Pr{\set E_2} & \left[
   \log(1+4P) + \frac{\E{|Y|^2 | \set{E}_2}}{1+4P} \right. \nonumber \\
   & \quad \left. - \E{\left. \left| Y - \sqrt{4P} \, U \right|^2 \right| \set E_2} \right].
   \label{eq:GMI-2-part-oof}
\end{align}
In Appendix~\ref{appendix:gmi-oof2}, we show that choosing $t_R = P^{\lambda_R}+b$ where $0<\lambda_R<1$ and $b$ is a real constant makes all terms behave as desired as $P$ increases:
\begin{align}
\begin{array}{l}
   \Pr{\set E_2} \rightarrow 1/2  \vspace{0.1cm} \\
   \E{\left. |Y|^2 \right| \set E_2}/(1+4P) \rightarrow 1  \vspace{0.1cm} \\
   \E{\left. \left|Y - \sqrt{4P} U \right|^2 \right| \set E_2} \rightarrow 1.
\end{array}
\label{eq:gmi-scaling-terms-oof}
\end{align}
We thus have 
\begin{align}
   \lim_{P \rightarrow \infty} \left[ \frac{1}{2} \log(1+4 P) - I_1(X;Y) \right] = 0 .
   \label{eq:GMI-2-example-oof}
\end{align}
Fig.~\ref{fig:oof3} shows the behavior of $I_1(A;Y)$ for $\lambda_R=1/2$ and $b=3$ as the curve labeled ``GMI, K=2''. As for the case without CSIT, the receiver can estimate $H$ accurately at large SNR, and one approaches the capacity with full CSIR.

Finally, the large-$K$ forward model rates are computed using \eqref{eq:GMI-5-part} but where $\bar X$ replaces $X$. One may again use the results of Appendix~\ref{appendix:csos-3} and the relations
\begin{align*}
    \E{\left. \bar X \right| Y=y}
    & =  \sqrt{P} \, \E{U|Y=y} \\
    \E{\left. |\bar X|^2 \right| Y=y}
    & = P \, \E{|U|^2 \big| Y=y} \\
    \Var{\bar X \big | Y=y}
    & = P \, \Var{U | Y=y} .
\end{align*}
The rates are shown as the curve labeled ``GMI, K=$\infty$'' in Fig.~\ref{fig:oof3}. So again, the large-$K$ forward model is good at high SNR but worse than the best $K=1$ model at low SNR.

\subsection{Partial CSIR, CSIT\at R}
\label{subsec:oof-Partial-CSIR-CSITatR}
Consider partial CSIR with $S_T=S_R$ and 
\begin{align}
    \Pr{S_R = H} = \epsb, \quad
    \Pr{S_R \ne H} = \epsilon
\end{align}
where $0\le\epsilon\le\frac{1}{2}$. We thus have both CSIT\at R and CSIR\at T. To compute $I(X;Y|S_R)$ in \eqref{eq:IXYH}, we write \eqref{eq:pysRx}--\eqref{eq:pysR} as
\begin{align*}
  p_{Y|S_R,X}(y | 0, x)
  & = \epsb\, \frac{e^{-|y|^2}}{\pi} + \eps\, \frac{e^{-\left|y - \sqrt{2}\, x(0)\right|^2}}{\pi} \\
  p_{Y|S_R,X}(y | \sqrt{2}, x)
  & = \epsb\, \frac{e^{-\left|y - \sqrt{2}\, x\left(\sqrt{2}\,\right) \right|^2}}{\pi} + \eps\, \frac{e^{-|y|^2}}{\pi} \\
  p_{Y|S_R}(y | 0)
  & = \epsb\, \frac{e^{-|y|^2}}{\pi} + \eps\, \frac{e^{-|y|^2/[1+2P(0)]}}{\pi[1+2P(0)]} \\
  p_{Y|S_R}(y | \sqrt{2})
  & = \epsb\, \frac{e^{-|y|^2/\left[1+2P\left(\sqrt{2}\,\right)\right]}}{\pi\left[1+2P\left(\sqrt{2}\,\right)\right]} + \eps\, \frac{e^{-|y|^2}}{\pi}
\end{align*}
where $X(s_T)$ is CSCG. We choose the transmit powers $P(0)$ and $P\left(\sqrt{2}\,\right)$ as in~\eqref{eq:powers-fullCSIR-parCSIT} to compare with the best CSIR. Fig.~\ref{fig:oof4} shows the resulting rates for $\eps=0.1$ as the curve labeled ``Partial CSIR, $I(X;Y|S_R)$''. Observe that at high SNR, the curve seems to approach the best CSIR curve from Fig.~\ref{fig:oof2} with $S_R =H \sqrt{P(S_T)}$. We prove this by studying a forward model GMI with $K=2$.

\begin{figure}[t]
      \centering
      \includegraphics[width=0.95\columnwidth]{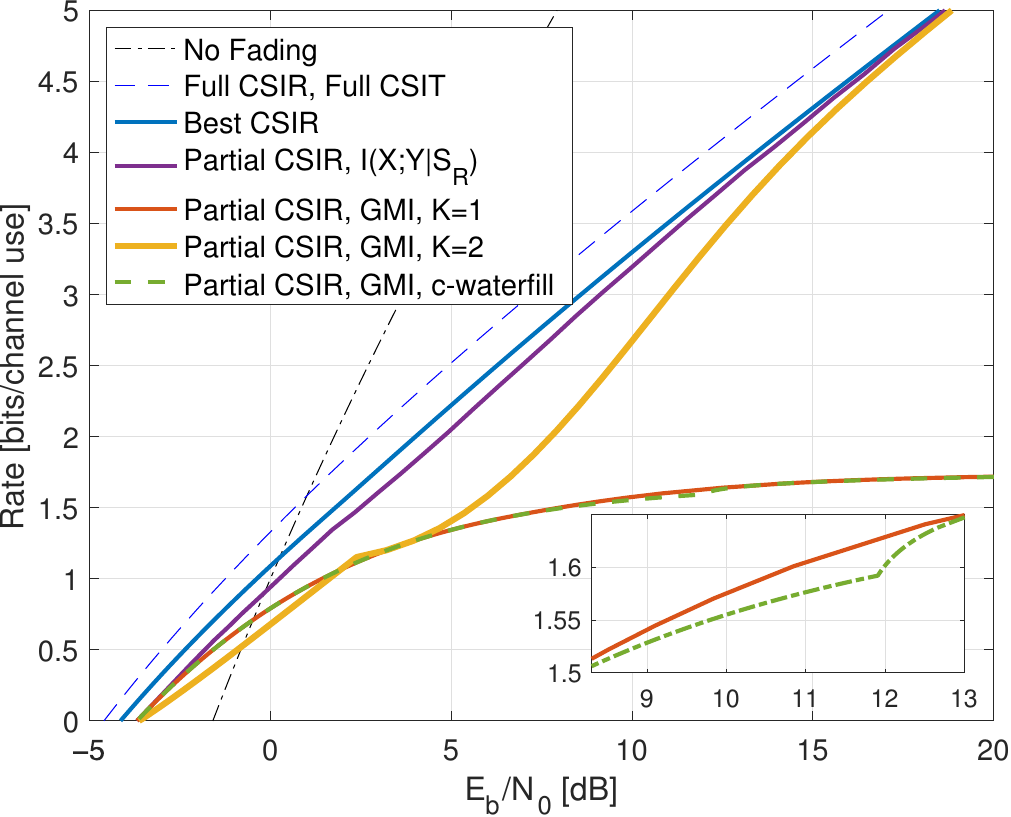}
      \caption{Rates for on-off fading with partial CSIR and CSIT\at R.
      The curve ``Best CSIR'' shows the capacity with $S_R =H \sqrt{P(S_T)}$.
      The mutual information $I(X;Y|S_R)$ and the GMI are for $\Pr{S_R \ne H}=0.1$
      and with CSCG inputs $X(s_T)$.
      The GMI for the $K=2$ partition uses $t_R=P^{0.4}$.  The curve labeled `c-waterfill' shows the conventional waterfilling rates.}
      \label{fig:oof4}
\end{figure}

The reverse model GMI requires $\Var{U|Y,S_R}$, which can be computed by simulation; see Appendix~\ref{appendix:csos-4}. However, optimizing the powers seems difficult. We instead focus on the forward model GMI of Theorem~\ref{theorem:partial-CSIR} 
for which we compute
\begin{align*}
  & \tilde g(0) = 2\,\eps^2, \quad \tilde g\left(\sqrt{2}\right) = 2\,\epsb^2 \\
  & \tilde \sigma^2(0) = \tilde \sigma^2\left(\sqrt{2}\right) = 2\,\eps\,\epsb
\end{align*}
and therefore~\eqref{eq:GMI-XYH} is
\begin{align}
  I_1(X;Y|S_R)
  & = \frac{1}{2} \log\left( 1 + \frac{2\,\eps^2 P(0)}{1 + 2\,\eps\,\epsb\,P(0)} \right) \nonumber \\
  & \quad + \frac{1}{2} \log\left( 1 + \frac{2\,\epsb^2 P\left(\sqrt{2}\,\right)}{1 + 2\,\eps\,\epsb\,P\left(\sqrt{2}\,\right)} \right).
  \label{eq:GMI-partCSIR-K1}
\end{align}
For CSIR\at T, the optimal power control policy is given by the quadratic waterfilling specified by \eqref{eq:waterfilling-noCSIR} or \eqref{eq:waterfilling-noCSIR2}:
\begin{align*}
   P(0) & = \frac{1+\epsb}{4\, \eps\, \epsb} 
   \left[ \sqrt{1 + 8\, \eps\, \epsb \left(\frac{1}{\lambda} - \frac{1}{2\,\eps^2} \right)^+
   \frac{\eps}{(1+\epsb)^2}} - 1 \right] 
   \\
   P\left(\sqrt{2}\right) & = \frac{1+\eps}{4\, \eps\, \epsb} 
   \left[ \sqrt{1 + 8\, \eps\, \epsb \left(\frac{1}{\lambda} - \frac{1}{2\,\epsb^2} \right)^+
   \frac{\epsb}{(1+\eps)^2}} - 1 \right] . 
\end{align*}
The rates are shown in Fig.~\ref{fig:oof4} as the curve labeled ``Partial CSIR, GMI, K=1''.
Observe that at high SNR the GMI \eqref{eq:GMI-partCSIR-K1} saturates at
\begin{align}
    \frac{1}{2} \log\left(1+\frac{\eps}{\epsb}\right) + \frac{1}{2} \log\left(1+\frac{\epsb}{\eps}\right) .
    \label{eq:oof-GMI-XYH-highSNR}
\end{align}
For example, for $\eps=0.1$, we approach $1.74$ bits at high SNR. On the other hand, at low SNR, the rate is maximized with $P(0)=0$ and $P\left(\sqrt{2}\right)=2P$ so that $I_1(X;Y | S_R) \approx 2\,\epsb^2 P$. We thus achieve a fraction of $\epsb^2$ of the power compared to full CSIT. For example, if $\epsilon=0.1$, the minimal $E_b/N_0$ is approximately $-3.69$ dB. 

Fig.~\ref{fig:oof4} also shows the conventional waterfilling rates as the curve labeled ``Partial CSIR, GMI, c-waterfill''. These rates are almost the same as the quadratic waterfilling rates except for the range of $E_b/N_0$ between 9 to 13 dB shown in the inset.

To improve the auxiliary model at high SNR, we use a $K=2$ GMI with (see Remark~\ref{remark:Partial-CSIR-CSITatR-2})
\begin{align*}
   h_1(s_R)=0, \quad h_2(s_R)=\sqrt{2}, \quad \sigma_1^2(s_R)=\sigma_2^2(s_R)=1
\end{align*}
for $s_R=0,\sqrt{2}$. The receiver chooses $\bar X(s_R) = \sqrt{P(s_R)}\, U$ (see Remark~\ref{remark:xbar-scaling2}) and we have (see Remark~\ref{remark:AGMI-part3})
\begin{align}
   & I_1(X;Y|S_R) = \frac{1}{2} \Pr{\set E_2|S_R=0} \nonumber \\
   & \quad \left\{ \log\left( 1 + 2 P(0) \right) 
   + \frac{\E{\left. |Y|^2 \right| \set{E}_2, S_R=0}}{1 + 2 P(0)} \right. \nonumber \\
   & \left. \qquad - \E{ \left. \big|Y - \sqrt{2} \, X(0) \big|^2 \right| \set{E}_2, S_R=0} \right\} \nonumber \\
   & \quad + \frac{1}{2} \Pr{\set E_2|S_R=\sqrt{2}\,} \nonumber \\
   & \quad \left\{ \log\left( 1 + 2 P\big(\sqrt{2}\,\big) \right) 
   + \frac{\E{\left. |Y|^2 \right| \set{E}_2, S_R=\sqrt{2}\,}}{1 + 2 P\left(\sqrt{2}\,\right)} \right. \nonumber \\
   & \left. \qquad - \E{ \left. |Y - \sqrt{2} \, X\big(\sqrt{2}\,\big)|^2 \right| \set{E}_2, S_R=\sqrt{2}} \right\}
   \label{eq:AGMI-SR-part2-oof}
\end{align}
where the $X(s_T)$, $s_T \in \set S_T$, are given by \eqref{eq:barX-max}. We consider $P(0)$ and $P\big(\sqrt{2}\,\big)$ that scale in proportion to $P$. In this case, Appendix~\ref{appendix:gmi-oof3} shows that choosing $t_R = P^{\lambda_R}+b$ where $0<\lambda_R<1$ gives the (best) full-CSIR capacity for large $P$, which is the rate specified in \eqref{eq:GMI-AYH2-oof}:
\begin{align}
   \lim_{P\rightarrow\infty} & \left[
   \frac{\eps}{2} \log\left( 1 + 2 P(0) \right)
   + \frac{\epsb}{2} \log\left( 1 + 2 P\left(\sqrt{2}\,\right) \right) \right. \nonumber \\
   & \quad - I_1(X;Y|S_R) \Big] = 0.
   \label{eq:GMI-2-example-oof2}
\end{align}
In other words, by optimizing $P(0)$ and $P\left(\sqrt{2}\,\right)$, at high SNR the $K=2$ GMI can approach the capacity of Proposition~\ref{proposition:part-CSIT2}. This is expected since the receiver can estimate $H \sqrt{P(S_T)}$ reliably at high SNR.

Fig.~\ref{fig:oof4} shows the behavior of this GMI and $t_R=P^{0.4}$, and where we have chosen $P(0)$ and $P\big(\sqrt{2}\,\big)$ according to \eqref{eq:powers-fullCSIR-parCSIT}. The abrupt change in slope at approximately 2.5 dB is because $P(0)$ becomes positive beyond this $E_b/N_0$. Keeping $P(0)=0$ for $E_b/N_0$ up to about 12 dB gives better rates, but for high SNR one should choose the powers according to \eqref{eq:powers-fullCSIR-parCSIT}.

\section{Rayleigh Fading}
\label{sec:Rayleigh-Fading}
Rayleigh fading has $H \sim \mathcal{CN}(0,1)$. The random variable $G=|H|^2$ thus has the density $p(g)=e^{-g} \cdot 1(g\ge0)$. Sec.~\ref{subsec:Rayleigh-noCSIR-noCSIT} and Sec.~\ref{subsec:Rayleigh-Fading-quantizer} review known results.

\subsection{No CSIR, No CSIT}
\label{subsec:Rayleigh-noCSIR-noCSIT}
Suppose $S_R=S_T=0$ and $X\sim\mathcal{CN}(0,P)$. The densities to compute $I(X;Y)$ for CSCG $X$ are
\begin{align}
   p(y|x) & = \frac{e^{-|y|^2/(|x|^2+1)}}{\pi(|x|^2+1)}
   \label{eq:r-pygx} \\
   p(y) & = \int_0^{\infty} \frac{e^{-g/P}}{P} \, \frac{e^{-|y|^2/(g+1)}}{\pi(g+1)} \, dg .
     \label{eq:r-py}
\end{align}
The minimum $E_b/N_0$ is approximately 9.2 dB, and the forward model GMI~\eqref{eq:noCSIR-noCSIT-If} is zero. The capacity is achieved by discrete and finite $X$~\cite{Abou-Faycal-IT01}, and at large SNR, the capacity behaves as $\log \log P$~\cite{Lapidoth:03}. Further results are derived in~\cite{Taricco-EL97,Marzetta-Hochwald-IT99,Zheng-IT-02,Gursoy-IT05,Chowdhury-ISIT16}.

\subsection{Full CSIR, CSIT\at R}
\label{subsec:Rayleigh-Fading-quantizer}
The capacity \eqref{eq:no-CSIT} for $B=0$ (no CSIT) is
\begin{align}
   C(P) & = \int_0^\infty e^{-g} \log\left( 1 + g\,P \right) dg \nonumber \\
   & = e^{1/P} E_1\left(1/P\right) \log(e)
   \label{eq:Shannon-cap4}
\end{align}
where the exponential integral $E_1(.)$ is given by \eqref{eq:expint} below. The wideband values are given by~\eqref{eq:wideband-fullCSIR-noCSIT}:
\begin{align*}
   \left.\frac{E_b}{N_0}\right|_{\text{min}} = \log 2, \quad S = 1 .
\end{align*}
The minimal $E_b/N_0$ is $-1.59$ dB, but the fading reduces the capacity slope. At high SNR, we have
\begin{align*}
   C(P) \approx \log(P) - \gamma
\end{align*}
where $\gamma \approx 0.57721$ is Euler's constant. The capacity thus behaves as for the case without fading but with an SNR loss of approximately 2.5 dB.

The capacity \eqref{eq:full-CSIT-2} with $B=\infty$ (or $S_T=G$) is (see~\cite[Eq.~(7)]{Goldsmith-Varaiya-IT97})
\begin{align}
   C(P) 
   & = \int_\lambda^\infty e^{-g} \log\left( g / \lambda \right) dg = E_1(\lambda).
   \label{eq:Shannon-cap5}
\end{align}
where $P(g)$ is given by \eqref{eq:waterfilling} and $\lambda$ is chosen so that
\begin{align*}
   P & = \int_\lambda^\infty e^{-g} P(g) \, dg = \frac{e^{-\lambda}}{\lambda} - E_1(\lambda).
\end{align*}
At low SNR we have large $\lambda$ and using the approximation \eqref{eq:E1-approx-large-x} below we compute
\begin{align}
   C(P)\approx e^{-\lambda}/\lambda \text{ and }
   P \approx e^{-\lambda}/\lambda^2.
   \label{eq:full-CSIR-full-CSIT-low-SNR}
\end{align}
We thus have $E_b/N_0 \approx \log(2) /\lambda$
and the minimal $E_b/N_0$ is $-\infty$. 

Consider now $B=1$ for which $P_{S_T}(3\Delta/2)=e^{-\Delta}$ and
\begin{align}
  & \E{G | G \ge \Delta }=1+\Delta \label{eq:EG-Rayleigh1} \\
  & \E{G ^2| G \ge \Delta }=2+2\Delta+\Delta^2. \label{eq:EG2-Rayleigh1}
\end{align}
We thus have the wideband quantities in \eqref{eq:minEbNo-2}--\eqref{eq:S-2}:
\begin{align}
   & \left.\frac{E_b}{N_0}\right|_{\text{min}} = \frac{\log 2}{1+\Delta} \label{eq:EbNo-Rayleigh1} \\
   & S = \frac{2 e^{-\Delta}(1+\Delta)^2}{2+2\Delta+\Delta^2}. \label{eq:S-Rayleigh1}
\end{align}

Fig.~\ref{fig:rf1} shows the capacities for $B=1$ and $\Delta=1,2,1/2$. The minimum $E_b/N_0$ value is 
\begin{align}
   -1.59\text{ dB} - 10 \log_{10}\left(1 + \Delta\right)
\end{align}
and for $\Delta=1,2,1/2$ we gain 3 dB, 4.8 dB, 1.8 dB, respectively, over no CSIT at low power. Note that one bit of feedback allows one to approach the full CSIT rates closely.

\begin{figure}[t]
      \centering
      \includegraphics[width=0.95\columnwidth]{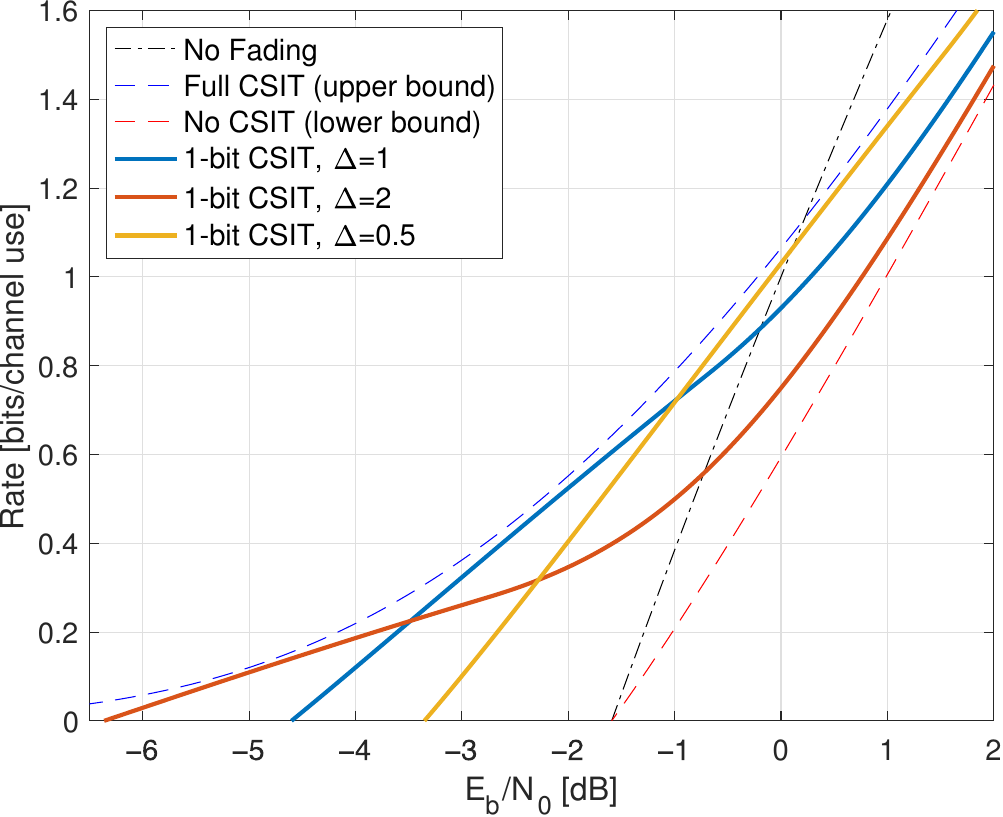}
      \caption{Capacities for Rayleigh fading with full CSIR, a one-bit quantizer with threshold $\Delta$, and CSIT\at R.}
      \label{fig:rf1}
\end{figure}

\begin{remark}
For the scalar channel~\eqref{eq:Y-Gauss}, knowing $H$ at both the transmitter and receiver provides significant gains at low SNR~\cite{Verdu02} but small gains at high SNR~\cite[Fig.~4]{Goldsmith-Varaiya-IT97} as compared to knowing $H$ at the receiver only. Furthermore, the reliability can be improved~\cite[Fig.~5-7]{Lau-Liu-Chen-IT04}. Significant gains are also possible for MIMO channels.
\end{remark}

\begin{remark}
An alternative way to derive \eqref{eq:EG-Rayleigh1}--\eqref{eq:S-Rayleigh1} is as follows. Define $\hat P=Pe^\Delta$ so for small $P$ the capacity is
\begin{align*}
   C(P) & = \int_\Delta^\infty e^{-g} \log\left( 1 +  g\, \hat P \right) dg \nonumber \\ 
   & = e^{1/\hat P} E_1\left( \frac{1}{\hat P} + \Delta \right)  + e^{-\Delta} \log(1+\hat P \Delta) \\
   & \approx P\, (1+\Delta) - \frac{1}{2} P^2 e^\Delta \left(2+2\Delta+\Delta^2\right).
\end{align*}
\end{remark}

\subsection{Full CSIR, Partial CSIT}
\label{subsec:Rayleigh-full-CSIR-partial-CSIT}
Consider noisy CSIT with 
\begin{align*}
    \Pr{S_T = 1(G\ge \Delta)} = \epsb, \quad
    \Pr{S_T \ne 1(G\ge \Delta)} = \epsilon.
\end{align*}
We begin with the most informative CSIR.

\subsubsection{$S_R=\sqrt{P(S_T)}H$}
Proposition~\ref{proposition:part-CSIT2} gives the capacity
\begin{align}
   & C(P) = \int_0^\infty e^{-g} \sum_{s_T} P(s_T|g) \log\left( 1 + g\, P(s_T) \right) \, dg \nonumber \\
   & = \int_0^\Delta e^{-g} \left[ \epsb\, \log\left( 1 + g\, P(0) \right) + \eps\, \log\left( 1 + g\, P(1) \right) \right] \, dg \nonumber \\
   & \quad + \int_\Delta^\infty e^{-g} \left[ \epsb\, \log\left( 1 + g\, P(1) \right) + \eps\, \log\left( 1 + g\, P(0) \right) \right] \, dg.
  \label{eq:GMI-AYH2-Rayleigh} 
\end{align}
It remains to optimize $P(0)$, $P(1)$ and $\Delta$. The two equations for the Lagrange multiplier $\lambda$ are
\begin{align}
   \lambda \cdot P_{S_T}(0) & = \int_0^{\Delta} e^{-g} \cdot \dfrac{\epsb\, g}{1+ g P(0)} dg \nonumber \\
   & \qquad + \int_{\Delta}^{\infty} e^{-g} \cdot \dfrac{\eps\, g}{1+ g P(0)} dg
   \label{eq:lambda-R-q0} \\
   \lambda \cdot P_{S_T}(1)  &= \int_0^{\Delta} e^{-g} \cdot \dfrac{\eps\,g}{1+ g P(1)} dg \nonumber \\
   & \qquad + \int_{\Delta}^{\infty} e^{-g} \cdot \dfrac{\epsb\,g}{1+ g P(1)} dg 
   \label{eq:lambda-R-q1} 
\end{align}
where $P_{S_T}(0)=\epsb - (\epsb - \eps) e^{-\Delta}$ and $P_{S_T}(1)=\eps + (\epsb - \eps) e^{-\Delta}$. The rates are shown in Fig.~\ref{fig:rf2}.

\begin{figure}[t]
      \centering
      \includegraphics[width=0.95\columnwidth]{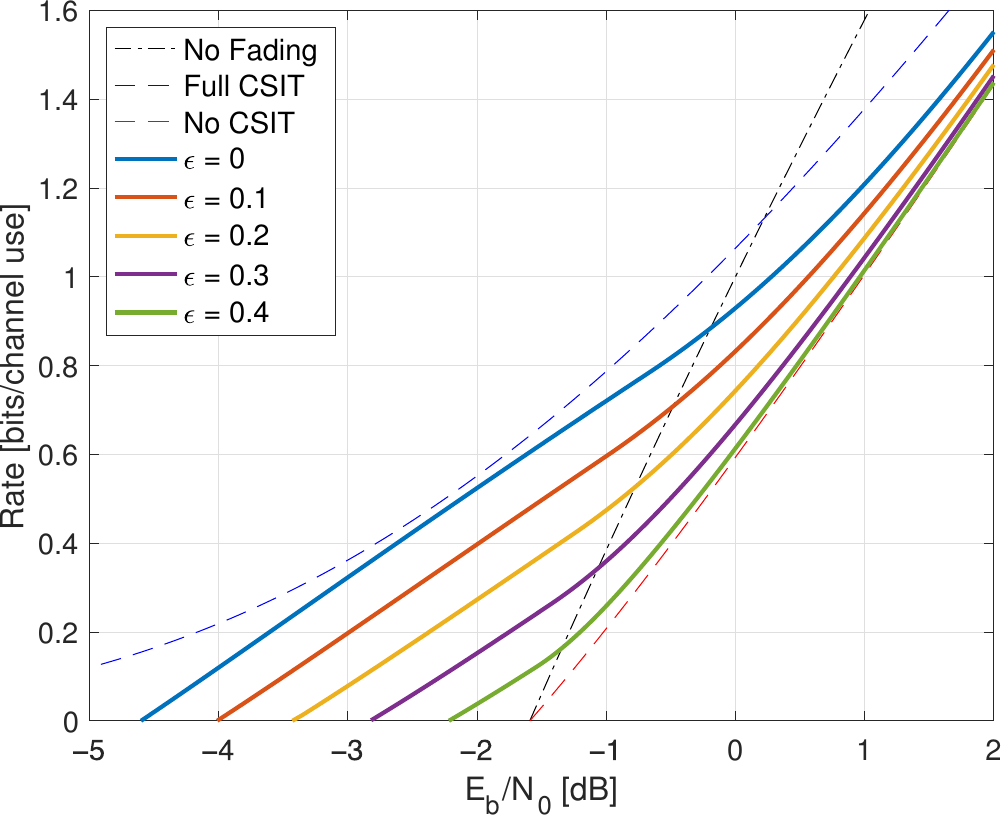}
      \caption{Capacities for Rayleigh fading, $S_R=\sqrt{P(S_T)} H$, and a one-bit quantizer with threshold $\Delta=1$, and various CSIT error probabilities $\eps$.}
      \label{fig:rf2}
\end{figure}

For fixed $\Delta$ and large $P$, we have $1/\lambda \approx P(0) \approx P(1)\approx P$ and approach the capacity~\eqref{eq:Shannon-cap4} without CSIT. In contrast, for small $P$ we may use similar steps as for \eqref{eq:lambda-q0}--\eqref{eq:lambda-q1}. Observe the following for \eqref{eq:lambda-R-q0} and \eqref{eq:lambda-R-q1}:
\begin{itemize}
\item both $P(0)$ and $P(1)$ decrease as $\lambda$ increases;
\item the maximal $\lambda$ in \eqref{eq:lambda-R-q0} is obtained with $P(0)=0$; this value is
\begin{align}
   \E{G| S_T=0} = \frac{\epsb - (\epsb - \eps)(1+\Delta) \, e^{-\Delta}}{P_{S_T}(0)}
   \label{eq:EG-ST1-0}
\end{align}
\item the maximal $\lambda$ in \eqref{eq:lambda-R-q1} is obtained with $P(1)=0$; this value is
\begin{align}
   \E{G| S_T=1} = \frac{\eps + (\epsb - \eps)(1+\Delta) \, e^{-\Delta}}{P_{S_T}(1)}.
   \label{eq:EG-ST1-1}
\end{align}
\end{itemize}
Thus, if $\E{G|S_T=0}<\E{G|S_T=1}$ and $0\le\epsilon<1/2$, then for $P$ below some threshold we have $P(0)=0$, $P(1)=P/P_{S_T}(1)$ and the capacity is
\begin{align}
    C(P)
   & = \int_0^\Delta e^{-g} \, \eps\, \log\left( 1 + \frac{g\,P}{P_{S_T}(1)} \right) \, dg \nonumber \\
   & \quad + \int_\Delta^\infty e^{-g} \, \epsb\, \log\left( 1 + \frac{g\,P}{P_{S_T}(1)} \right) \, dg.
  \label{eq:GMI-AYH2-Rayleigh-2} 
\end{align}
We compute $C'(0)=\E{G| S_T=1}$ which is given by~\eqref{eq:EG-ST1-1} so that $1 \le C'(0) \le 1+\Delta$, as expected from \eqref{eq:EbNo-Rayleigh1}. For example, for $\epsilon=0.1$ and $\Delta=1$ we have $C'(0)\approx1.75$ and therefore the minimal $E_b/N_0$ is approximately $-4.01$ dB. 

The best $\Delta$ is the unique solution $\hat \Delta$ of the equation
\begin{align}
  e^{-\Delta} = \frac{\eps}{\epsb - \eps} (\Delta - 1)
  \label{eq:solution}
\end{align}
and the result is $C'(0) = \hat \Delta \ge 1$. We have the simple bounds
\begin{align}
  1 + \frac{1}{2} \log\left( \frac{1}{\eps} -2 \right) \le C'(0) \le 1 + \frac{1}{e} \left(\frac{1}{\eps} - 2 \right)
\end{align}
where the left inequality follows by taking logarithms and using $\log(\Delta-1) \le \Delta-2$, and the right inequality follows by using $e^{-\Delta} \le e^{-1}$ in \eqref{eq:solution}. For example, for $\eps \rightarrow 0$ we have $C'(0) \rightarrow \infty$, and for $\eps \rightarrow 1/2$ we have $C'(0) \rightarrow 1$.

\subsubsection{$S_R=H$}
For the less informative CSIR, one may use \eqref{eq:pyah-1} and \eqref{eq:pyh-1b} to compute $I(A;Y|H)$. The reverse model GMI requires $\Var{U|Y,S_R}$, which can be computed by simulation; see Appendix~\ref{appendix:csos-2}. Again, however, optimizing the powers seems difficult. We instead focus on the forward model GMI of Corollary~\ref{corollary:part-CSIT}, which is
\begin{align}
  I_1(A;Y|H) = \int_0^\infty e^{-g} \log\left( 1 + \text{SNR}(g) \right) \, dg
  \label{eq:iayh-q2-r}
\end{align}
where
\begin{align}
  \text{SNR}(g) = \frac{g \tilde P_T(g)}
  {1 + g \, \eps \, \epsb \left( \sqrt{P(0)} - \sqrt{P(1)} \right)^2 }
\end{align}
and
\begin{align}
    \tilde P_T(g) & = \left\{ \begin{array}{ll}
     \left(\epsb\sqrt{P(0)} + \eps\sqrt{P(1)} \right)^2, & g < \Delta \\
     \left(\eps\sqrt{P(0)} + \epsb\sqrt{P(1)} \right)^2, & g \ge \Delta.
     \end{array} \right.
\end{align}
It remains to optimize $P(0)$, $P(1)$ and $\Delta$. Computing the derivatives seems complicated, so we use numerical optimization for fixed $\Delta=1$ as in Fig.~\ref{fig:rf2}. The results are shown in Fig.~\ref{fig:rf3}. For fixed $\Delta$ and large $P$, it is best to choose $P(0) \approx P(1)$ so that $SNR(g) \approx g P$ and we approach the rate of no CSIT. For small $P$, however, the best $P(0)$ is no longer zero and $C'(0)$ is smaller than~\eqref{eq:EG-ST1-1}.

\begin{figure}[t]
      \centering
      \includegraphics[width=0.95\columnwidth]{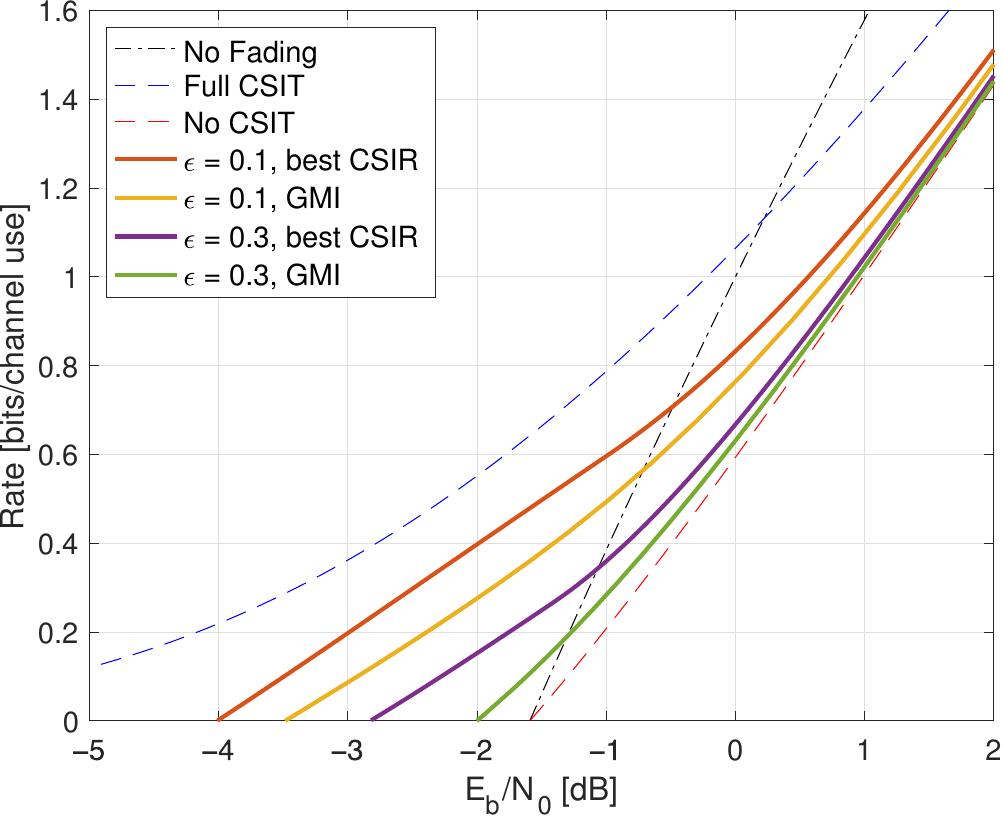}
      \caption{Rates for Rayleigh fading, $S_R=H$ and $S_R =H \sqrt{P(S_T)}$, a one-bit quantizer with threshold $\Delta=1$,  and various $\eps$. The curves labeled ``best CSIR'' show the capacities with $S_R =H \sqrt{P(S_T)}$. The curves labeled ``GMI'' show the rates \eqref{eq:iayh-q2-r} for the optimal powers $P(0)$ and $P(1)$.}
      \label{fig:rf3}
\end{figure}

\subsection{Partial CSIR, Full CSIT}
\label{subsec:Rayleigh-partial-CSIR-full-CSIT}
Consider $S_T=H$ and suppose we choose the $X(h)$ to be jointly CSCG with variances $\E{|X(h)|^2}=P(h)$ and correlation coefficients 
\begin{align*}
  \rho(h,h') = \frac{\E{ X(h) X(h')^\text{*}}}{\sqrt{P(h) P(h')}}
\end{align*}
and where $\E{P(H)} \le P$. We then have
\begin{align*}
    p(y | s_R) = \int_{\mathbb C} p(h|s_R) \frac{e^{-|y|^2/(|h|^2 P(h) + 1)}}{\pi (|h|^2 P(h) + 1)} \, dh.
\end{align*}
As in \eqref{eq:py}, $p(y | s_R )$ and $h(Y | S_R)$ depend only on the marginals of $A$ and not on the $\rho(h,h')$. We thus have the problem of finding the $\rho(h,h')$ that minimize
\begin{align*}
    h(Y | A, S_R) = \int_\set{A} p(a) \, h(Y | S_R, A=a) \, da.
\end{align*}
We will use fully-correlated $X(h)$ as discussed in Sec.~\ref{subsec:Partial-CSIR-fullCSIT}. We again consider $S_R=0$ and $S_R=1(G \ge t)$.

\subsubsection{$S_R=0$}
For the heuristic policies, the power \eqref{eq:heuristic-PC-policies-Phat} is
\begin{align}
   \hat P = \frac{P}{\Gammafun{1+a,t}}
   \label{eq:heuristic-PC-policies-Phat-noCSIR}
\end{align}
and the rate \eqref{eq:R-fullCSIT-noCSIR} is
\begin{align}
   & I_1(A;Y) = \log \Bigg( 1 + \nonumber \\
   & \left. \frac{P\, \Gammafun{\frac{3+a}{2},t}^2}
   {\Gammafun{1+a,t} + P\left[ \Gammafun{2+a,t} - \Gammafun{\frac{3+a}{2},t}^2 \right]} \right)
   \label{eq:R-fullCSIT-noCSIR-Rayleigh}
\end{align}
where $\Gamma(s,x)$ is the upper incomplete gamma function; see Appendix~\ref{appendix:gamma}. Moreover, the expression \eqref{eq:minEbNo-fullCSIT-noCSIR} is
\begin{align}
   \left.\frac{E_b}{N_0}\right|_{\text{min}}
   & = \frac{\Gammafun{1+a,t}}{\Gammafun{\frac{3+a}{2},t}^2} \cdot \log 2 .
   \label{eq:minEbNo-fullCSIT-noCSIR-Rayleigh}
\end{align}
We remark that $\Gamma(s,0)=\Gamma(s)$ where $\Gamma(x)$ is the gamma function. We further have
\begin{align*}
   \begin{array}{ll}
   \Gamma(0,t) = E_1(t), & \Gamma(1,t)=e^{-t}, \\
   \Gamma(2,t)=e^{-t} (t+1), & \Gamma(3,t)=e^{-t} (t^2+2t+2) .
   \end{array}
\end{align*}

For example, the TCP policy ($a=0$) has $\hat P = P \, e^t$.
At low SNR, it turns out that the best choice is $t=0.283$ 
for which we have $\Gamma(1,t) / \Gamma(3/2,t)^2 \approx 1.174$. 
The minimum $E_b/N_0$ in \eqref{eq:minEbNo-fullCSIT-noCSIR-TMF} is thus $-0.90$ dB. At high SNR, the best choice is $t=0$ so that \eqref{eq:R-fullCSIT-noCSIR-Rayleigh} with $\Gamma(3/2,0)=\Gamma(3/2)=\sqrt{\pi}/2$ gives 
\begin{align}
   I_1(A;Y) = \log\left( 1 + \frac{P \, \pi/4}{1+P(1-\pi/4)} \right) .
   \label{eq:AGMI-Rayleigh-TCP}
\end{align}
The TCP rate thus saturates at $2.22$ bits per channel use; see the curve labeled ``TCP, GMI, K=1'' in Fig.~\ref{fig:rf4}.

\begin{figure}[t]
      \centering
      \includegraphics[width=0.95\columnwidth]{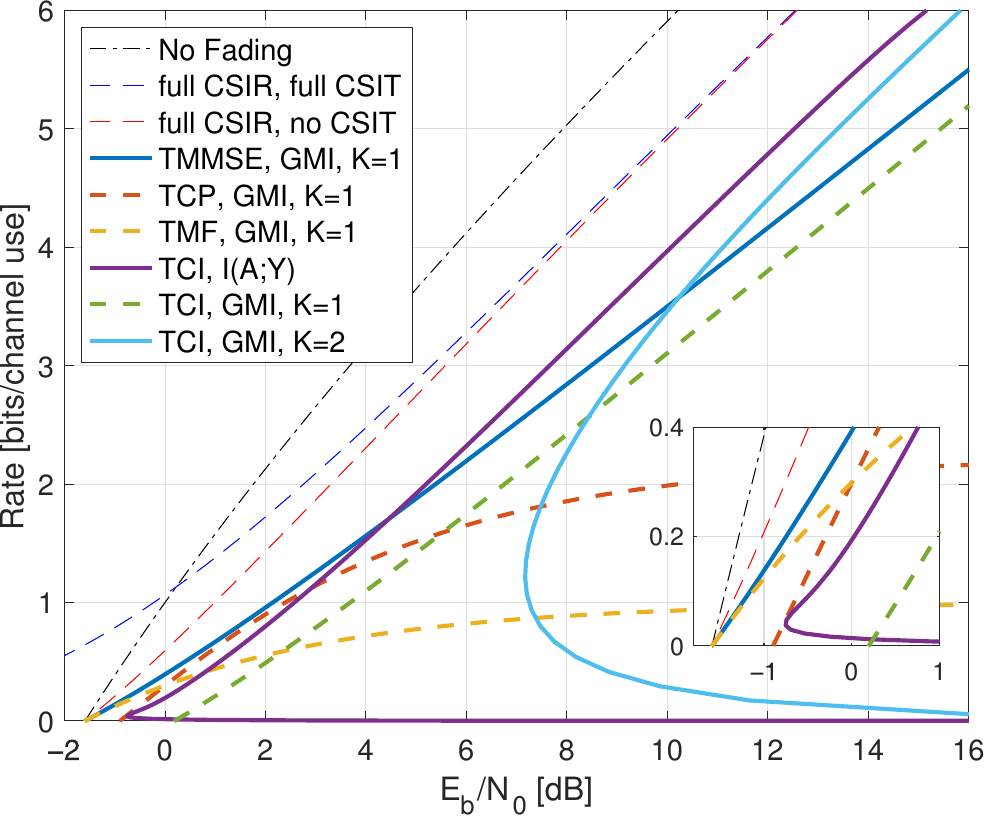}
      \caption{Rates for Rayleigh fading with $S_T=H$ and $S_R=0$. The threshold $t$ was optimized for the $K=1$ curves, while $t=P^{-0.4}$ for the $I(A;Y)$ and $K=2$ curves. The $K=2$ GMI uses $t_R=P^{0.4}$.
      }
      \label{fig:rf4}
\end{figure}

The TMF policy ($a=1$) has $\hat P = P\, e^t / (t+1)$. The best choice is $t=0$ for which we have $\Gamma(2)=1$ and $\Gamma(3)=2$ and therefore \eqref{eq:R-fullCSIT-noCSIR-Rayleigh} is
\begin{align}
   I_1(A;Y) = \log\left( 1 + \frac{P}{1+P} \right).
   \label{eq:AGMI-Rayleigh-TMF}
\end{align}
The minimum $E_b/N_0$ in \eqref{eq:minEbNo-fullCSIT-noCSIR-TMF} is $-1.59$ dB, and at high SNR, the TMF rate saturates at $1$ bit per channel use.  The rates are shown as the curve labeled ``TMF, GMI, K=1'' in Fig.~\ref{fig:rf4}.

The TCI policy ($a=-1$) has $\hat P = P / E_1(t)$ and using $\Gamma(0,t)=E_1(t)$ and $\Gamma(1,t)=e^{-t}$ gives
\begin{align}
   I_1(A;Y) = \log\left( 1 + \frac{P}
   {e^{2t} \, E_1(t) + P \left( e^t - 1 \right)} \right) .
   \label{eq:R-fullCSIT-noCSIR-Rayleigh-TCI}
\end{align}
The minimum $E_b/N_0$ in \eqref{eq:minEbNo-fullCSIT-noCSIR-Rayleigh} is
\begin{align}
   \left.\frac{E_b}{N_0}\right|_{\text{min}} = E_1(t) \, e^{2t} \cdot \log 2 .
\end{align}
Optimizing over $t$ by taking derivatives (see \eqref{eq:expint-derivative} below), the best $t$ satisfies the equation $2 t e^t E_1(t)=1$ which gives $t\approx 0.61$ and the minimal $E_b/N_0$ is approximately 0.194 dB. On the other hand, for large SNR, we may choose $t=1/P$ and using $E_1(t) \approx \log(1/t)$ for small $t$ gives
\begin{align*}
   I_1(A;Y) \approx \log\left( 1 + \frac{P}{1 + \log P} \right).
\end{align*}
Since the pre-log is at most 1, the capacity grows with pre-log 1 for large $P$.
We see that TMF is best at small $P$ while TCI is best at large $P$.
The rates are shown as the curve labeled ``TCI, GMI, K=1'' in Fig.~\ref{fig:rf4}.

The simple channel output of TCI permits further analysis. Using Remark~\ref{remark:TCI-densities}, we compute the mutual information $I(A;Y)$ by numerical integration; see the curve labeled ``TCI, $I(A;Y)$'' in Fig.~\ref{fig:rf4}. We see that at high SNR, the TCI mutual information is larger than the GMI for TCP, TMF, and (of course) TCI. Moreover, as we show, the TCI mutual information can work well at low SNR.

Motivated by Sec.~\ref{subsec:oof-Partial-CSIR-full-CSIT} and Fig.~\ref{fig:oof3}, we again use the GMI~\eqref{eq:AGMI-part2} with $K=2$ and \eqref{eq:2partition}. We further choose $h_1=0$, $\sigma_1^2=\sigma_2^2=1$, and \begin{align*}
    \bar{X}=\frac{\sqrt{\hat P}}{h_2} \, U,\quad U\sim\mathcal{CN}(0,1) .
\end{align*}
The expression \eqref{eq:AGMI-part2} simplifies to
\begin{align}
   I_1(A;Y) = \Pr{\set E_2} & \left[
   \log\left(1 + \hat P\right) + \frac{\E{|Y|^2 | \set{E}_2}}{1 + \hat P} \right. \nonumber \\
   & \quad \left. - \; \E{ \left. \left| Y - \sqrt{\hat P} \, U \right|^2 \right| \set E_2 } \right].
   \label{eq:AGMI-2-part-Rayleigh}
\end{align}
The GMI~\eqref{eq:AGMI-2-part-Rayleigh} exhibits interesting high and low SNR scaling by choosing the following thresholds $t,t_R$.
\begin{itemize}
\item For high SNR, we choose
\begin{align}
   t=P^{-\lambda} \text{ and } t_R=\hat P^{\lambda_R}
   \label{eq:gmi-rf1-parameters1}
\end{align}
where  $0 < \lambda < 1$ and $0 < \lambda_R < 1$. As $P$ increases, $t$ decreases and Appendix~\ref{appendix:gmi-rf1} shows that
\begin{align}
\begin{array}{l}
   \Pr{\set E_2} \rightarrow 1  \vspace{0.1cm} \\
   \E{|Y|^2 | \set E_2} \Big/ \left( 1+ \hat P \right) \rightarrow 1  \vspace{0.1cm} \\
   \E{\left|Y - \sqrt{\hat P} \, U\right|^2 | \set E_2} \rightarrow 1.
\end{array}
\label{eq:gmi-scaling-terms-Rayleigh1}
\end{align}
Inserting $\hat P = P/E_1(t)$, we thus have
\begin{align}
    \lim_{P \rightarrow \infty} \left[ I_1(A;Y) - \log\left( 1 + \frac{P}{E_1(t)} \right) \right] = 0.
    \label{eq:TCI-high-SNR-scaling}
\end{align}
We further have $E_1(t)\approx \lambda \log P$ by using~\eqref{eq:E1-approx-small-x} in Appendix~\ref{appendix:cHandIndep}, and the high-SNR slope of the GMI matches the slope of $\log P$ but the additive gap to $\log P$ increases. The high SNR rates are shown as the curve labeled ``TCI, GMI, K=2'' in Fig.~\ref{fig:rf4}
for $\lambda=\lambda_R=0.4$.
\item
For low SNR, we choose
\begin{align}
   t=-\log (P/c) \text{ and } t_R=\hat P
   \label{eq:gmi-rf1-parameters2}
\end{align}
for a constant $c>0$.
As $P$ decreases, both $t$ and $\hat P = P/E_1(t)$ increase and Appendix~\ref{appendix:gmi-rf1} shows that
\begin{align}
\begin{array}{l}
   \Pr{\set E_2} \approx e^{-t-1}  \vspace{0.1cm} \\
   \E{|Y|^2 | \set E_2} \Big/ \left( 1+ 2\hat P \right) \rightarrow 1  \vspace{0.1cm} \\
   \E{\left|Y - \sqrt{\hat P} \, U\right|^2 | \set E_2} \rightarrow 1.
\end{array}
\label{eq:gmi-scaling-terms-Rayleigh2}
\end{align}
Using~\eqref{eq:E1-approx-large-x}, we have $I_1(A;Y) \approx e^{-t-1}\log t$ which vanishes as $t$ grows. But we also have
\begin{align}
   \frac{E_b}{N_0} = \frac{P}{R} \log 2
   & \approx \frac{c\,e^{-t} \log 2}{e^{-t-1} \log t} \approx \frac{c\, e \log2}{\log(- \log P)}
   \label{eq:Rayleigh-EbNo-scaling}
\end{align}
which decreases (very slowly) as $P$ decreases. The minimal $E_b/N_0$ is therefore $-\infty$. The low SNR rates are shown as the curve labeled ``TCI, GMI, K=2'' in Fig.~\ref{fig:rf5} for $c=1.4$.
\end{itemize}

\begin{figure}[t]
    \centering
    \includegraphics[width=0.95\columnwidth]{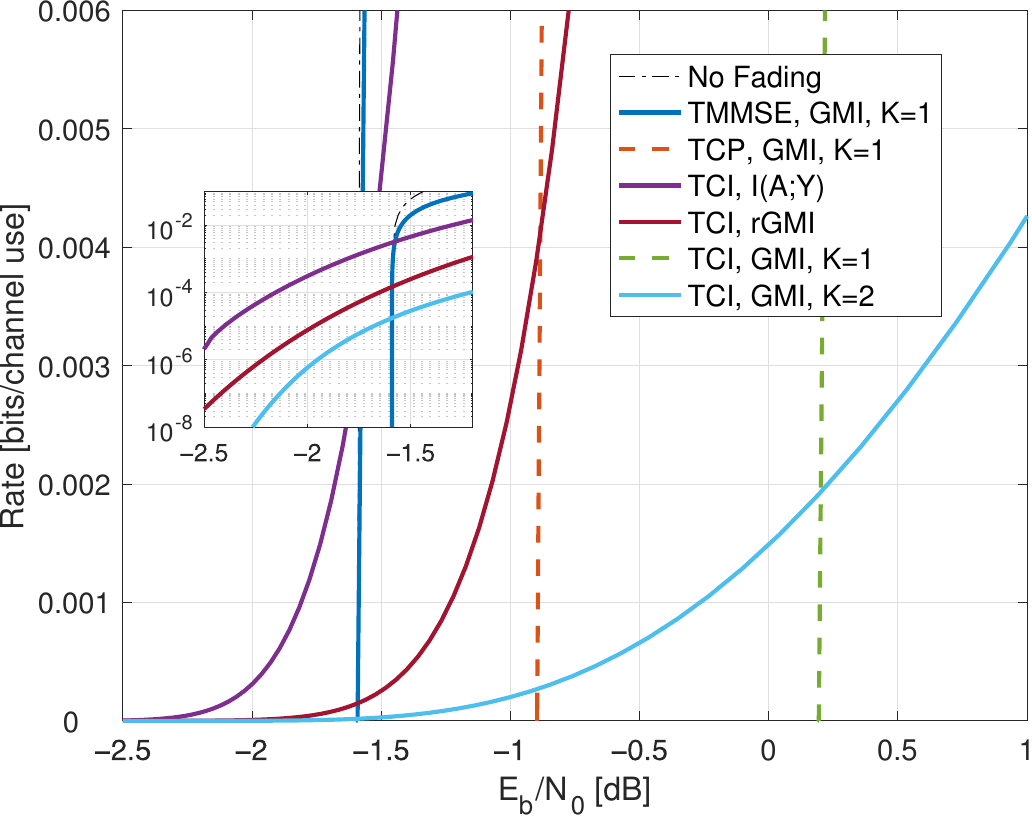}
    \caption{Low-SNR rates for Rayleigh fading with $S_T=H$ and $S_R=0$. The threshold $t$ was optimized for the $K=1$ curves, while
    $t=-\log (P/1.4)$ for the $I(A;Y)$, rGMI, and $K=2$ curves. The $K=2$ GMI uses $t_R=\hat P$. The TMF and TMMSE GMIs are indistinguishable for this range of rates. 
    }
    \label{fig:rf5}
\end{figure}

Fig.~\ref{fig:rf5} shows that the TCI mutual information achieves a minimal $E_b/N_0$ below $-1.59$ dB. At $E_b/N_0=-2$ dB, we computed $I_1(A;Y)\approx 6 \times 10^{-7}$ and $I(A;Y)\approx 3 \times 10^{-4}$. The $K=2$ partition is thus useful to prove that TCI can achieve $E_b/N_0$ arbitrarily close to zero. Fig.~\ref{fig:rf5} also shows the reverse model GMI as the curve labeled ``TCI, rGMI'' which has the rate $I_1(A;Y)\approx 8 \times 10^{-6}$ at $E_b/N_0=-2$ dB.

We compare the full CSIR and full CSIT rates. At high SNR, the GMI for $S_R=0$ achieves the same capacity pre-log as $S_R=H$. At low SNR, recall from \eqref{eq:full-CSIR-full-CSIT-low-SNR} that with full CSIR/CSIT we have $E_b/N_0 \approx \log(2)/\lambda$. To compare the rates for similar $E_b/N_0$, we set $\lambda = \log t$, where $t$ is as in \eqref{eq:gmi-rf1-parameters2} and $c \approx 1$. The TCI $K=2$ GMI without CSIR is approximately $e^{-t} \log t$ while the full CSIR rate~\eqref{eq:full-CSIR-full-CSIT-low-SNR} is approximately $e^{-\lambda}/\lambda \approx 1/(t \log(t))$. Thus, the $K=2$ GMI with no CSIR is a fraction $t e^{-t} \log(t)^2$ of the full CSIR capacity.

\subsubsection{$S_R=1(G \ge t)$}
The power in \eqref{eq:heuristic-PC-policies-Phat} is again \eqref{eq:heuristic-PC-policies-Phat-noCSIR}
and the rate~\eqref{eq:R-fullCSIT-CSIR} is
\begin{align}
   & I_1(A;Y|S_R) = e^{-t} \cdot \log \Bigg( 1 + \nonumber \\
   & \; \left. \frac{P\, e^{2t} \, \Gammafun{\frac{3+a}{2},t}^2}
   {\Gammafun{1+a,t} + P\left[ e^t \, \Gammafun{2+a,t} - e^{2t} \,\Gammafun{\frac{3+a}{2},t}^2 \right]} \right) .
   \label{eq:R-fullCSIT-CSIR-Rayleigh}
\end{align}
Moreover, the expression \eqref{eq:minEbNo-fullCSIT-CSIR} is
\begin{align}
   & \left.\frac{E_b}{N_0}\right|_{\text{min}} 
   = \frac{\Gammafun{1+a,t}}{e^t \cdot \Gammafun{\frac{3+a}{2},t}^2} \cdot \log 2
   \label{eq:minEbNo-fullCSIT-CSIR-Rayleigh} 
\end{align}
which is the same as \eqref{eq:minEbNo-fullCSIT-noCSIR-Rayleigh} except for the factor $e^t$ in the denominator. This implies that the minimal $E_b/N_0$ can be improved for $t>0$.

The TCP, TMF, and TCI rates \eqref{eq:R-fullCSIT-CSIR-Rayleigh} are the respective
\begin{align}
   & I_1(A;Y|S_R) \nonumber \\
   & \; = e^{-t} \log\left( 1 + \frac{P \, e^{2t} \, \Gammafun{\frac{3}{2},t}^2}
   {e^{-t} + P \left[t+1 - e^{2t} \, \Gammafun{\frac{3}{2},t}^2 \right]} \right)
   \label{eq:AGMI-fullCSIT-CSIR-TCP} \\
   & I_1(A;Y|S_R) = e^{-t} \log\left( 1 + \frac{P \, (t+1)^2}{e^{-t} (t+1) + P} \right)
   \label{eq:AGMI-fullCSIT-CSIR-TMF} \\
   & I_1(A;Y|S_R) = e^{-t} \log\left( 1 + \frac{P}{E_1(t)} \right).
   \label{eq:AGMI-fullCSIT-CSIR-TCI}
\end{align}

\begin{remark}
As pointed out in Remark~\ref{remark:TCI-MI}, the TCI GMI \eqref{eq:AGMI-fullCSIT-CSIR-TCI} is $I(A;Y|S_R)$. One can also understand this by observing that the receiver knows $\sqrt{G P(G)}$ for all $G$. The mutual information is thus related to the rate \eqref{eq:GMI-AYH2} of Proposition~\ref{proposition:part-CSIT2}.
\end{remark}

\begin{figure}[t]
      \centering
      \includegraphics[width=0.95\columnwidth]{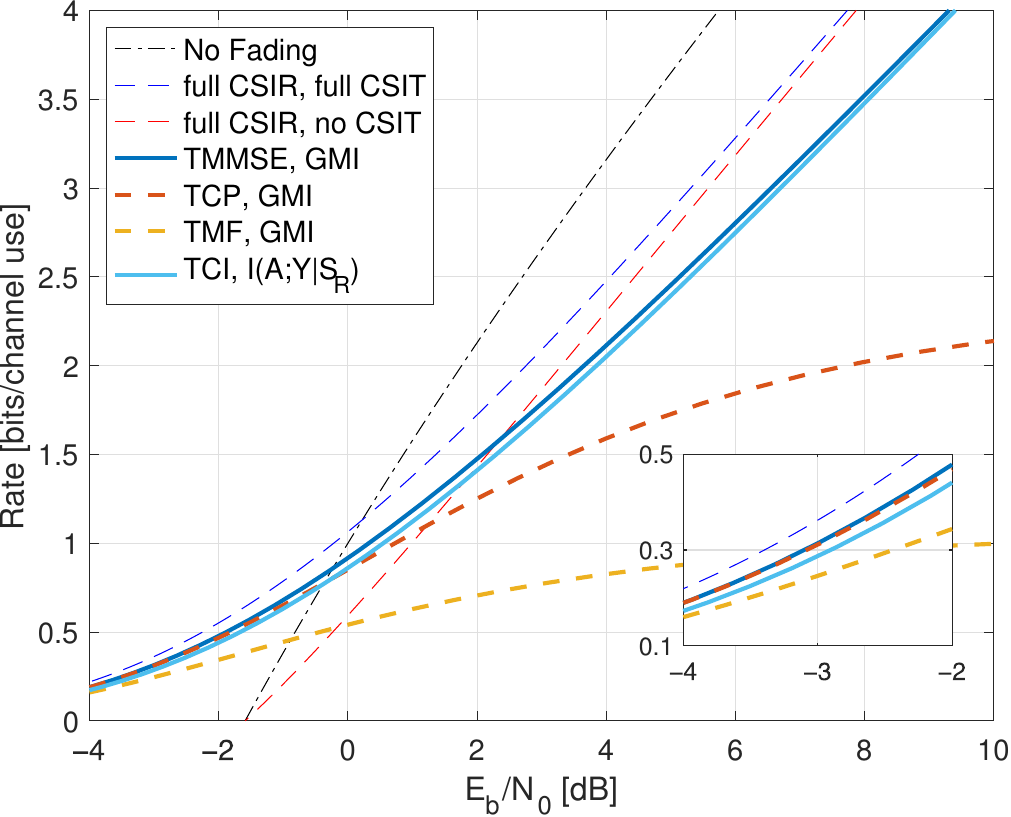}
      \caption{Rates for Rayleigh fading with full CSIT and $S_R=1(G\ge t)$.}
      \label{fig:rf6}
\end{figure}

The minimal $E_b/N_0$ in \eqref{eq:minEbNo-fullCSIT-CSIR-Rayleigh} are the respective
\begin{align}
   & \left.\frac{E_b}{N_0}\right|_{\text{min}}
   = \frac{1}{e^{2t} \cdot \Gammafun{\frac{3}{2},t}^2} \cdot \log 2
   \label{eq:minEbNo-fullCSIT-CSIR-Rayleigh-TCP} \\ 
   & \left.\frac{E_b}{N_0}\right|_{\text{min}}
   = \frac{1}{t+1} \cdot \log 2
   \label{eq:minEbNo-fullCSIT-CSIR-Rayleigh-TMF} \\ 
   & \left.\frac{E_b}{N_0}\right|_{\text{min}}
   = e^t \, E_1(t) \cdot \log 2 .
   \label{eq:minEbNo-fullCSIT-CSIR-Rayleigh-TCI}
\end{align}
The above expressions mean that, for all three policies, we can make the minimal $E_b/N_0$ as small as desired by increasing $t$. For example, for TCI, we can bound (see~\eqref{eq:expint-bounds} below)
\begin{align}
   \frac{1}{t+1} < e^t E_1(t) < \frac{1}{t}.
\end{align}
TCI thus has a slightly larger (slightly worse) minimal $E_b/N_0$ than TMF for the same $t$, as discussed after \eqref{eq:minEbNo-fullCSIT-TCI}.

For large $P$, the TCP rate \eqref{eq:AGMI-fullCSIT-CSIR-TCP} is optimized by $t \approx 0.163$ 
and the rate saturates at $\approx 2.35$ bits per channel use. 
The TMF rate \eqref{eq:AGMI-fullCSIT-CSIR-TMF} is optimized with $t=0$, and the rate saturates at 1 bit per channel use. For the TCI rate \eqref{eq:AGMI-fullCSIT-CSIR-TCI}, we again choose $t=1/P$ and use $E_1(t) \approx \log(1/t)$ for small $t$ to show that the capacity grows with pre-log 1:
\begin{align*}
   I_1(A;Y|S_R) \approx \log\left( 1 + \frac{P}{\log P} \right) .
\end{align*}
Again, TMF is best at small $P$ while TCI is best at large $P$.

\begin{remark} \label{remark:TCI-high-SNR}
Comparing~\eqref{eq:TCI-high-SNR-scaling} and~\eqref{eq:AGMI-fullCSIT-CSIR-TCI}, the $S_R=0$, $K=2$, TCI GMI in~\eqref{eq:AGMI-2-part-Rayleigh} approaches the $S_R=1(G\ge t)$ mutual information $I(A;Y|S_R)$ in~\eqref{eq:AGMI-fullCSIT-CSIR-TCI} at high SNR.
\end{remark}

\subsubsection{Optimal Policy}
Consider now the optimal power control policy. Suppose first that $S_R=0$ for which Theorem~\ref{theorem:Partial-CSIR-Full-CSIT} gives the TMMSE policy with $t=0$:
\begin{align}
  \sqrt{P(h)} & = \frac{\alpha |h|}{\beta + |h|^2} .
\end{align}
For Rayleigh fading, we thus have (see~\eqref{eq:int4} below)
\begin{align}
   P = \int_0^\infty e^{-g} \frac{\alpha^2 g}{(\beta + g)^2} \, dg
   = \alpha^2\left[ (\beta + 1 ) e^\beta E_1(\beta) -1 \right] \label{eq:AGMI-2-Rayleigh-P}
\end{align}
with the two expressions (see~\eqref{eq:int3} and~\eqref{eq:int5} below)
\begin{align}
   & \tilde P = \int_0^\infty e^{-g} \frac{\alpha^2 g}{\beta + g} \, dg 
   = \alpha^2\left[ 1 - \beta e^\beta E_1(\beta) \right]^2 \label{eq:AGMI-2-Rayleigh-Pt}
\end{align}
\begin{align}
  \E{G P(H)} & = \int_0^\infty e^{-g} \frac{\alpha^2 g^2}{(\beta + g)^2} \, dg \nonumber \\
  & = \alpha^2\left[ 1 + \beta - \beta (\beta + 2 ) e^\beta E_1(\beta) \right].
\end{align}
Given $P$ and $\beta$, we may compute $\alpha^2$ from \eqref{eq:AGMI-2-Rayleigh-P}. We then search for the optimal $\beta$ for fixed $P$. The rates are shown as the curve labeled ``TMMSE, GMI, K=1'' in Figs.~\ref{fig:rf4}--\ref{fig:rf5} and we see that the TMMSE strategy has the best $K=1$ rates.

Consider next $S_R=1(G\ge t)$ and the TMMSE policy. We compute (see~\eqref{eq:int4} below)
\begin{align}
   P & = \int_t^\infty e^{-g} \frac{\alpha^2 g}{(\beta + g)^2} \, dg \nonumber \\
   & = \alpha^2\left[ (\beta + 1 ) e^\beta E_1(t+\beta) - e^{-t} \frac{\beta}{t+\beta} \right] \label{eq:AGMI-2-Rayleigh-P2}
\end{align}
and (see~\eqref{eq:int3} and~\eqref{eq:int5} below)
\begin{align}
   & \sqrt{\tilde P(1)} 
   = \int_t^\infty \frac{e^{-g}}{e^{-t}} \frac{\alpha g}{\beta + g} \, dg \nonumber \\
   & \quad = \alpha \left[ 1 - \beta e^{t+\beta} E_1(t+\beta) \right] \label{eq:AGMI-2-Rayleigh-Pt2} \\
   & \E{|Y|^2|S_R=1} = \int_t^\infty \frac{e^{-g}}{e^{-t}} \left(1 + \frac{\alpha^2 g^2}{(\beta + g)^2} \right) \, dg \nonumber \\
   & \quad = 1 + \alpha^2\left[ 1 + \frac{\beta^2}{t+\beta} - \beta (\beta + 2 ) e^{t+\beta} E_1(t+\beta) \right].
\end{align}
We optimize as for the $S_R=0$ case: given $P$, $\beta$, $t$, we compute $\alpha^2$ from \eqref{eq:AGMI-2-Rayleigh-P2}. We then search for the optimal $\beta$ for fixed $P$ and $t$. The optimal $t$ is approximately a factor of 1.1 smaller than for the TCI policy. The rates are shown in Fig.~\ref{fig:rf6} as the curve labeled ``TMMSE, GMI''.

\subsection{Partial CSIR, CSIT\at R}
\label{subsec:Rayleigh-partialCSIR-CSITatR}
Suppose $S_R$ is defined by (see~\eqref{eq:H-mmse-estimate})
\begin{align*}
   H= \sqrt{\epsb}\, S_R + \sqrt{\eps}\, Z_R
\end{align*}
where $0 \le \epsilon \le 1$ and $S_R,Z_R$ are independent with distribution $\mathcal{CN}(0,1)$. We further consider the CSIT $S_T=|S_R|^2$. 

The reverse model GMI again requires $\Var{U|Y,S_R}$, which can be computed by simulation; see Appendix~\ref{appendix:csos-4}. However, as in Sec.~\ref{subsec:oof-Partial-CSIR-CSITatR} and~\ref{subsec:Rayleigh-full-CSIR-partial-CSIT}, optimizing the powers seems difficult, and we instead focus on forward models. The expressions \eqref{eq:partial-CSIR-gt}--\eqref{eq:partial-CSIR-sigmat} are
\begin{align}
   \tilde g(s_R) = \epsb \, s_T, \quad 
   \tilde \sigma^2(s_R) = \epsilon. 
\end{align}
The GMI \eqref{eq:GMI-XYH} of Theorem~\ref{theorem:partial-CSIR} is
\begin{align}
  I_1(X;Y | S_R) 
  & = \int_{\lambda/\epsb}^\infty e^{-s_T} \log\left( 1 + \frac{\epsb \, s_T P(s_T)}
  {1 + \epsilon \, P(s_T)} \right) \, d s_T
  \label{eq:GMI-XYH-Rayleigh}
\end{align}
where the power control policy $P(s_T)$ is given by~\eqref{eq:waterfilling-noCSIR2}. The parameter $\lambda$ is chosen so that $\E{P(S_T)}=P$. For example, for $\eps \rightarrow 0$ we recover the waterfilling solution \eqref{eq:waterfilling}. Fig.~\ref{fig:rf7} shows the quadratic and conventional waterfilling rates, which lie almost on top of each other. For example, the inset shows the rates for $\epsilon=0.2$ and a small range of $E_b/N_0$.

\begin{figure}[t]
      \centering
      \includegraphics[width=0.95\columnwidth]{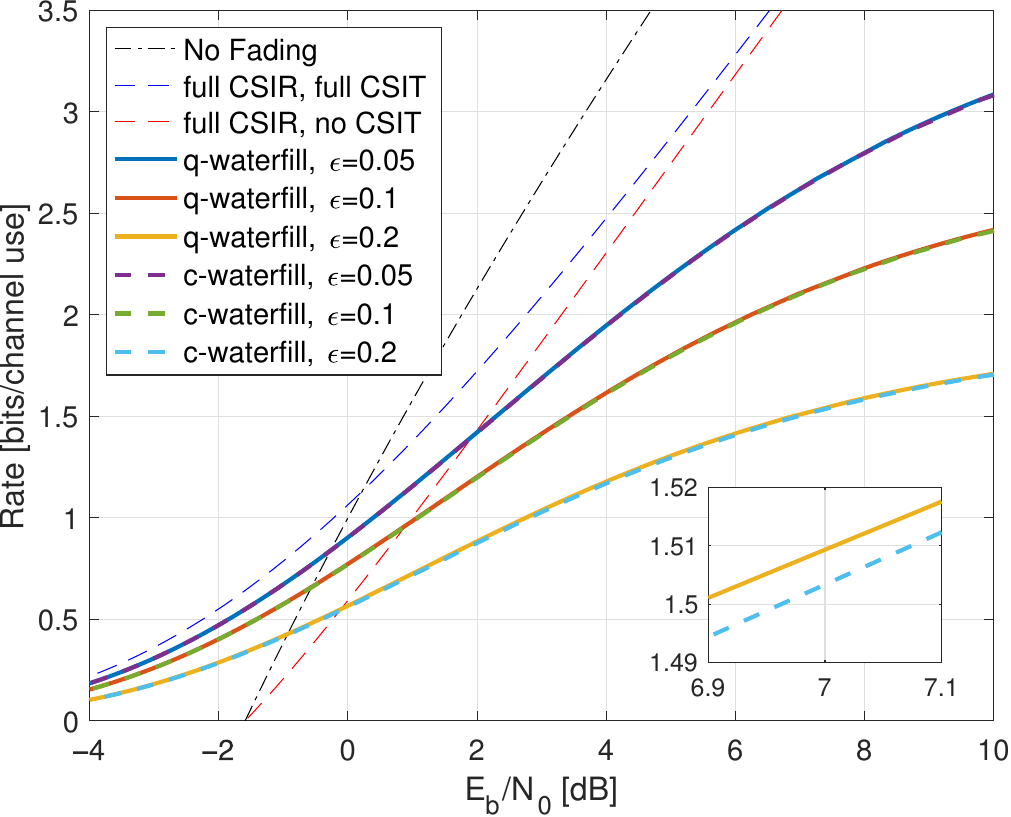}
      \caption{Rates for Rayleigh fading with partial CSIR and CSIT\at R.
      The curves labeled `q-waterfill' and `c-waterfill' are the quadratic and
      conventional waterfilling rates, respectively.}
      \label{fig:rf7}
\end{figure}

\section{Channels with In-Block Feedback}
\label{sec:ibf-CSIT}
This section generalizes Shannon's model described in Sec.~\ref{subsec:CSIT-model} to include block fading with in-block feedback. For example, the model lets one include delay in the CSIT and permits many other generalizations for network models~\cite{Kramer14}.

\subsection{Model and Capacity}
\label{subsec:ibf-model-capacity}
The problem is specified by the FDG in Fig.~\ref{fig:ptp-iBM-2}.  The model has a message $M$, and the channel input and output strings
\begin{align*}
 X_i^L   & = (X_{i1},\ldots,X_{iL}) \\
 Y_i^L & =(Y_{i1},\ldots,Y_{iL})
\end{align*}
for blocks $i=1,\ldots,n$. The channel is specified by a string $S_H^n=(S_{H1},\ldots,S_{Hn})$ of i.i.d. hidden channel states. The CSIR $S_{Ri\ell}$ is a (possibly noisy) function of $S_{Hi}$ for all $i$ and $\ell$. The receiver sees the channel outputs (see~\eqref{eq:Y-Gauss})
\begin{align}
   & (Y_{i\ell}, S_{Ri\ell}) = \left( f_\ell\left( X_{i}^{\ell}, S_{Hi}, Z_i^L\right), S_{Ri\ell} \right)
   \label{eq:ibm-Y}
\end{align}
for some functions $f_\ell(\cdot)$, $\ell=1,\dots,L$. Observe that the $X_{i}^{\ell}$ influence the $Y_{i\ell}$ in a causal fashion. The random variables $M,S_{H1},\ldots,S_{Hn},Z_1^L,\ldots,Z_n^L$ are mutually independent.

\begin{figure}[t]
      \centering
      \includegraphics[width=0.95\columnwidth]{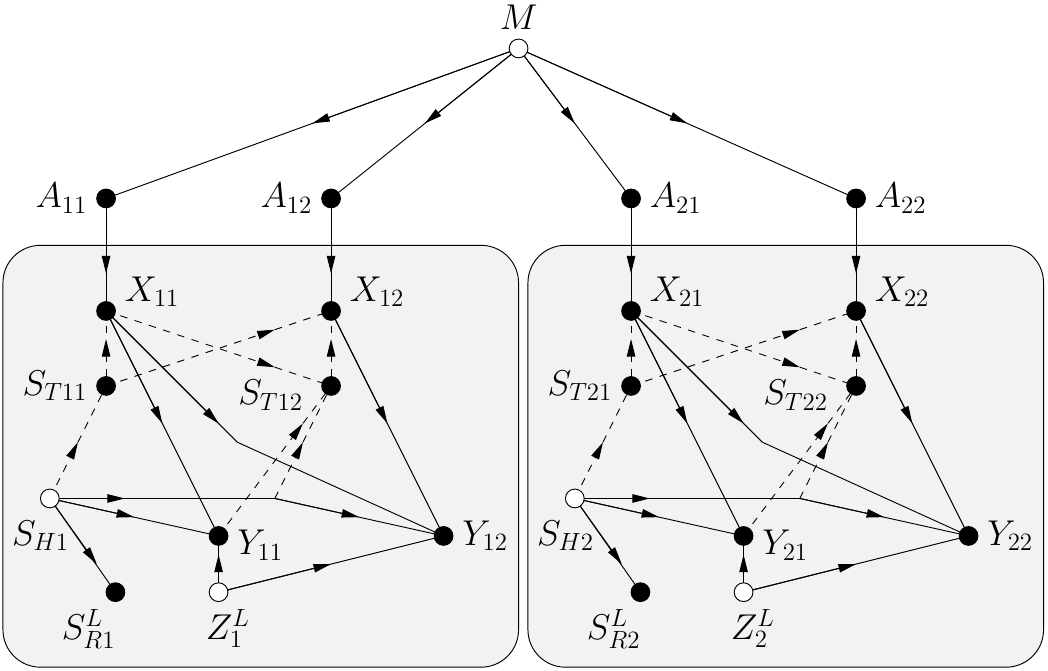}
      \caption{FDG for a block fading model with $n=2$ blocks of length $L=2$ and in-block feedback.
      Across-block dependence via past $S_{Ti\ell}$ is not shown.}
      \label{fig:ptp-iBM-2}
\end{figure}

We now permit past channel symbols to influence the CSIT; see Sec.~\ref{subsec:feedback}.  Suppose the CSIT has the form
\begin{align}
   S_{Ti\ell}=f_{T\ell}\left( S_{Hi}, X_i^{\ell-1},Y_i^{\ell-1} \right)
   \label{eq:ibm-ST}
\end{align}
for some function $f_{T\ell}(.)$ and
for all $i$ and $\ell$. The motivation for~\eqref{eq:ibm-ST} is that useful CSIR may not be available until the end of a block or even much later. In the meantime, the receiver can, e.g., quantize the $Y_i^{\ell-1}$ and transmit the quantization bits via feedback. This lets one study fast power control and beamforming without precise knowledge of the channel coefficients.

Define the string of past and current states as
\begin{align}
   s_T^{i\ell}=\left( s_{T1}^L,\dots,s_{T(i-1)}^L, s_{Ti}^\ell \right).
   \label{eq:past-present-state}
\end{align}
The channel input at time $i\ell$ is $X(s_T^{i\ell})$ and the adaptive codeword $A^{nL}$ is defined by the ordered lists
\begin{align}
  A_{i\ell} & = \left[ X(s_T^{i\ell}), \;\forall\; s_T^{i\ell} \right]
\end{align}
for $1\le i \le n$ and $1\le\ell\le L$. The adaptive codeword $A^{nL}$ is a function of $M$ and is thus independent of $S_H^{n}$ and $S_R^{nL}$.

The model under consideration is a special case of the channels introduced in~\cite[Sec.~V]{Kramer14}. However, the model in~\cite{Kramer14} has transmission and reception begin at time $\ell=2$ rather than $\ell=1$. To compare the theory, one must thus shift the time indexes by 1 unit and increase $L$ to $L+1$. The capacity for our model is given by~\cite[Thm.~2]{Kramer14} which we write as
\begin{align}
   C & \overset{(a)}{=} \max_{A^L} \frac{1}{L} I(A^L ; Y^L, S_R^L )  \nonumber \\
   & \overset{(b)}{=} \max_{A^L} \frac{1}{L} I(A^L ; Y^L \big| S_R^L ).
   \label{eq:ibm-C}
\end{align}
where $(a)$ follows by normalizing by $L$ rather than $L+1$, and step $(b)$ follows by the independence of $A^L$ and $S_R^L$.

\subsection{GMI for Scalar Channels}
\label{subsec:ibf-GMI-scalar}
We will study scalar block fading channels; extensions to vector channels follow as described in Sec.~\ref{subsec:CSIT-GMI}. Let $\ul Y=[Y_1,\dots,Y_L]^T$ be the vector form of $Y^L$ and similarly for other strings with $L$ symbols. The GMI with parameter $s$ is
\begin{align}
    I_s(A^L; Y^L \big| S_R^L)
    = \E{\log \frac{q(\ul Y \big| \ul A, \ul S_R)^s}
    {q(\ul Y \big| \ul S_R)} }
   \label{eq:AGMI-SRL}
\end{align}

\subsubsection{Reverse Model}
For the reverse model, let $\ul A$ be a column vector that stacks the $X_\ell(s_T^{\ell})$ for all $s_T^{\ell}$ and $\ell$. Consider a reverse density as in~\eqref{eq:I-LB-q2-v}:
\begin{align*}
    q\big(a^L|y^L\big) = \frac{\exp\left(- \ul z(\ul y, \ul s_R)^\dag \, {\bf Q}_{\ul A | \ul Y = \ul y, \ul S_R = \ul s_R}^{-1} \, \ul z(\ul y, \ul s_R) \right)}{\pi^N \det {\bf Q}_{\ul A | \ul Y = \ul y, \ul S_R = \ul s_R}}
\end{align*}
where 
\begin{align*}
    \ul z(\ul y, \ul s_R) = \ul a - \E{\ul A| \ul Y = \ul y, \ul S_R = \ul s_R}.
\end{align*}
Using the forward model $q(y^L|a^L)=q(a^L|y^L)/p(a^L)$, the GMI with $s=1$ becomes
\begin{align}
   I_1(A^L ; Y^L,S_R^L) = \E{ \log \frac{\det {\bf Q}_{\ul A}}{\det {\bf Q}_{\ul A | \ul Y, \ul S_R}} }.
   \label{eq:AGMI-r-bif}
\end{align}
To simplify, consider adaptive symbols as in~\eqref{eq:adaptive-codeword-conventional2} (cf.~\eqref{eq:adaptive-codeword-conventional2-v}):
\begin{align}
   X_\ell(S_T^\ell) = \sqrt{P_\ell(S_T^\ell)} \, e^{j\phi_\ell(S_T^\ell)} \, U_\ell
   \label{eq:adaptive-codeword-conventional2-ibf}
\end{align}
where $\ul U \sim \mathcal{CN}(\ul 0,{\bf I})$. In other words, consider a conventional codebook represented by the $U_\ell$ and adapt the power and phase based on the available CSIT. The mutual information becomes $I(A^L;Y^L,S_R^L) = I(U^L;Y^L,S_R^L)$ (cf.~\eqref{eq:Shannon-cap2}) and the GMI with $s=1$ is (cf.~\eqref{eq:AGMI-r2})
\begin{align}
   I_1(A^L;Y^L \big| S_R^L) = \E{ - \log \det {\bf Q}_{\ul U \left| \ul Y, \ul S_R \right.} }.
   \label{eq:AGMI-r2-ibf}
\end{align}
In fact, one may also consider choosing $U_\ell = U$ for all $\ell$ in which case we compute (cf.~\eqref{eq:AGMI-r3})
\begin{align}
   I_1(A^L;Y^L \big| S_R^L) = \E{ - \log \Var{U \big | \ul Y, \ul S_R} }.
   \label{eq:AGMI-r2-ibf2}
\end{align}

\subsubsection{Forward Model}
Consider the following forward model (cf.~\eqref{eq:AGMI-qya} and \eqref{eq:AGMI-qyasr}):
\begin{align}
   q(\ul y \big| \ul a, \ul s_R) 
   = \frac{\exp
   \left(- \ul z(\ul s_R)^\dag {\bf Q}_{\ul{Z}}(\ul s_R)^{-1} \ul z(\ul s_R) \right)}
   {\pi^L \det {\bf Q}_{\ul Z}(\ul s_R)} .
   \label{eq:AGMI-qyasr-L}
\end{align}
with 
\begin{align*}
    \ul z(\ul s_R) = \ul y - {\bf H}(\ul s_R)\, \ul{\bar x}(\ul s_R)
\end{align*}
and where similar to \eqref{eq:xbar-sR} we define
\begin{align}
   \ul{\bar X}(\ul s_R) = \sum_{\ul s_T} {\bf W}(\ul s_T, \ul s_R) \, \ul X(\ul s_T)
   \label{eq:xbar-sRL}
\end{align}
where the ${\bf W}(\ul s_T, \ul s_R)$ are $L \times L$ complex matrices. Note that
\begin{align}
    \ul X(\ul s_T) = [X_1(s_{T1}), X_2(s_T^2),\dots,X_2(s_T^L)]^T
\end{align}
so $X_\ell$ is a function of $A^L$ and $S_T^\ell$, $\ell=1,\dots,L$.

We have the following generalization of Lemma~\ref{lemma:AGMI-max-MIMO} (see also Theorem~\ref{theorem:AGMI-SR-max}) where the novelty is that $S_T$ is replaced with $\ul S_T$. Define $\ul U(\ul s_T)\sim \mathcal{CN}(\ul 0,{\bf I})$ and $\ul X(\ul s_T) = {\bf Q}_{\ul X(\ul s_T)}^{1/2} \, \ul U(\ul s_T)$ for all $\ul s_T$.

\begin{theorem}
\label{theorem:AGMI-SR-max-block-fading}
A GMI \eqref{eq:AGMI-SRL} for the scalar block fading channel $p(y^L|a^L,s_R^L)$, an adaptive codeword $A^L$ with jointly CSCG entries, the auxiliary model \eqref{eq:AGMI-qyasr-L}, and with fixed ${\bf Q}_{X(\ul s_T)}$ is
\begin{align}
   & I_1(A^L; Y^L | S_R^L) \nonumber \\
   & = \E{ \log \left( \frac{\det {\bf Q}_{\ul Y}(\ul S_R)}
   {\det\left( {\bf Q}_{\ul Y}(\ul S_R) - \tilde {\bf D}(\ul S_R) \, \tilde {\bf D}(\ul S_R)^\dag \right)} \right) }.
   \label{eq:AGMI-3-I1-block-fading}
\end{align}
where
\begin{align}
  {\bf Q}_{\ul Y}(\ul s_R) & = \E{ \left.\ul Y\, \ul Y^\dag \right| \ul S_R = \ul s_R}
\end{align}
and for $M \times M$ unitary ${\bf V}_R(\ul s_T , \ul s_R)$ the matrix $\tilde {\bf D}(\ul s_R)$ is
\begin{align}
  \E{ \left. {\bf U}_T(\ul S_T, \ul s_R) \, {\bf \Sigma}(\ul S_T, \ul s_R) \, {\bf V}_R(\ul S_T, \ul s_R)^\dag \right| \ul S_R = \ul s_R}
   \label{eq:AGMI-2-block-fading}
\end{align}
and ${\bf U}_T(\ul s_T, \ul s_R)$ and ${\bf \Sigma}(\ul s_T, \ul s_R)$ are $N\times N$ unitary and $N\times M$ rectangular diagonal matrices, respectively, of the SVD
\begin{align}
   & \E{\left. \ul Y \, {\ul U(\ul s_T)}^\dag \right| \ul S_T = \ul s_T, \ul S_R = \ul s_R} \nonumber \\
   & \quad =  {\bf U}_T(\ul s_T, \ul s_R) \, {\bf \Sigma}(\ul s_T, \ul s_R) \, {\bf V}_T(\ul s_T, \ul s_R)^\dag
   \label{eq:CSIT-SVD-decomposition-block-fading}
\end{align}
for all $\ul s_T$, $\ul s_R$ and the ${\bf V}_T(\ul s_T,\ul s_R)$ are $M\times M$ unitary matrices. One may maximize \eqref{eq:AGMI-3-I1-block-fading} over the unitary ${\bf V}_R(\ul s_T, \ul s_R)$. 
\end{theorem}

Suppose next that the actual channel is $\ul Y= H \ul X + \ul Z$ where $\ul Z \sim \mathcal{CN}(\ul 0,{\bf I})$. The extension of \eqref{eq:CSIT-AGMI-3-vector-linear} and~\eqref{eq:AGMI-SR-Gauss} to block fading channels with CSIR is
\begin{align}
    & I_1(A^L; Y^L | S_R^L) \nonumber \\
    & = \sum_{\ell=1}^L \E{\log \left( 1 + \frac{\tilde P_\ell(\ul S_R)} {1 + \E{G P_\ell(S_T^\ell) \big| \ul S_R} - \tilde P_\ell(\ul S_R)} \right)}
    \label{eq:AGMI-3-vector-linear-block-fading}
\end{align}
where (cf.~\eqref{eq:AGMI-Pt}--\eqref{eq:AGMI-Y2})
\begin{align*}
    & \tilde P_\ell(\ul s_R) = \E{\,\left| \E{\left. H \sqrt{P_\ell(S_T^\ell)} \, \right| S_T^\ell, \ul S_R = \ul s_R} \right| \,}^2
    \\
    & \E{|Y_\ell|^2 \big| \ul S_R = \ul s_R}=1 + \E{G P_\ell(S_T^\ell) \big| \ul S_R = \ul s_R}.
\end{align*}

\subsection{CSIT\at R}
\label{subsec:ibf-CSITatR}
Continuing as in Sec.~\ref{subsec:CSITatR}, suppose the CSIT in~\eqref{eq:ibm-ST} can be written by replacing $S_{Hi}$ with $S_{Ri}^\ell$ for all $i$ and $\ell$:
\begin{align}
    S_{Ti\ell}=f_{T\ell}\left(S_{Ri}^\ell,X_i^{\ell-1},Y_i^{\ell-1}\right) .
   \label{eq:ibm-f}
\end{align}
The capacity \eqref{eq:ibm-C} then simplifies to a directed information. To see this, expand the mutual information in \eqref{eq:ibm-C} as
\begin{align}
   I(A^L ; Y^L \big| S_R^L)
   & \overset{(a)}{=} \sum_{\ell=1}^L I\left( A^L, X^\ell ; Y_\ell \big| S_R^L, Y^{\ell-1} \right) \nonumber \\
   & \overset{(b)}{=} \sum_{\ell=1}^L I( X^\ell ; Y_\ell \big| S_R^L, Y^{\ell-1} ) \label{eq:ibm-f-information}
\end{align}
where step $(a)$ follows because $X^\ell$ is a function of $A^L$ and $S_T^\ell$ in~\eqref{eq:ibm-f}, and step $(b)$ follows by the Markov chains
\begin{align}
   & A^L - [ S_R^L, X^\ell, Y^{\ell-1} ] - Y_\ell .
\end{align}
The capacity is therefore (see the definition~\eqref{eq:DI})
\begin{align}
   C = \max_{X_\ell(S_T^\ell), \;\ell=1,\ldots,L} \frac{1}{L} I(X^L \rightarrow Y^L \big| S_R^L) .
   \label{eq:ibm-cap-f}
\end{align}
The maximization in \eqref{eq:ibm-cap-f} under a cost constraint becomes a constrained maximization for which $\E{c(X^L,Y^L)}\le LP$ for some cost function $c(\cdot)$.

\begin{remark} \label{remark:shift}
As outlined at the end of Sec.~\ref{subsec:ibf-model-capacity}, the capacity \eqref{eq:ibm-cap-f} is a special case of the theory in~\cite[eq.~(48)]{Kramer14}. To see this, define the extended and time-shifted strings
\begin{align*}
    \hat A^{L+1} = (0,A^L), \;\;
    \hat X^{L+1} = (0,X^L), \;\;
    \hat Y^{L+1} = (0,Y^L).
\end{align*}
Since $A^L$ and $S_R^L$ are independent, one may expand \eqref{eq:ibm-f-information} as
\begin{align}
   & I(A^L ; Y^L \big| S_R^L ) = I(A^L \,;\,  (S_{R2},\dots,S_{RL},0), Y^L \big| S_{R1} ) \nonumber \\
   & \quad \overset{(a)}{=} \sum_{\ell=1}^L I(A^L, X^\ell \,;\, S_{R(\ell+1)}, Y_\ell \big| S_R^\ell, Y^{\ell-1} ) \nonumber \\
   & \quad \overset{(b)}{=} \sum_{\ell=1}^L I(X^{\ell} \,;\, S_{R(\ell+1)}, Y_\ell \big| S_R^\ell, Y^{\ell-1} ) \nonumber \\
   & \quad = \sum_{\ell=2}^{L+1} I(\hat X^\ell \,;\, S_{R\ell}, \hat Y_\ell \big| S_R^{\ell-1}, \hat Y^{\ell-1} )
   \label{eq:ibm-f-information2}
\end{align}
where step $(a)$ follows because $X^\ell$ is a function of $A^L$ and $S_T^\ell$ in~\eqref{eq:ibm-f}, and where $S_{R(L+1)}=0$, and step $(b)$ follows by the Markov chains
\begin{align}
   & A^L - [ X^{\ell}, Y^{\ell-1}, S_R^\ell ] - [Y_{\ell},S_{R(\ell+1)}].
\end{align}
The expression \eqref{eq:ibm-f-information2} is the desired directed information
\begin{align}
   I(A^L ; Y^L, S_R^L ) = I( \hat X^{L+1} \rightarrow \hat Y^{L+1}, S_R^{L+1} ).
   \label{eq:ibm-f-information3}
\end{align}
\end{remark}

\begin{remark}
Consider the basic CSIT model
\begin{align}
   S_{Ti\ell}=f_T(S_{Ri\ell})
   \label{eq:ibm-f0}
\end{align}
for some function $f_T(\cdot)$ and for $\ell=1,\ldots,L$ and $i=1,\ldots,n$.
This model was studied in~\cite[Sec.~III.C]{Goldsmith-Medard-IT07} and its
capacity is given as (see~\cite[eq.~(35) with eq.~(13)]{Goldsmith-Medard-IT07})
\begin{align} 
   C & = \max_{X_\ell(S_T^\ell), \;\ell=1,\ldots,L} \frac{1}{L} I(X^L ; Y^L \big| S_R^L, S_T^L).
   \label{eq:ibm-cap2} 
\end{align}
To see that \eqref{eq:ibm-cap2} is a special case of \eqref{eq:ibm-cap-f},
observe that
\begin{align} 
   I(X^L \rightarrow Y^L | S_R^L)
   & \overset{(a)}{=} \sum_{\ell=1}^L I(X^\ell ; Y_\ell \big| S_R^L, S_T^L, Y^{\ell-1}) \nonumber \\
   & \overset{(b)}{=} \sum_{\ell=1}^L I(X^L ; Y_\ell \big| S_R^L, S_T^L, Y^{\ell-1}) 
   \label{eq:ibm-cap-f2} 
\end{align}
where step $(a)$ follows by~\eqref{eq:ibm-f-information}, and step $(b)$ follows by the Markov chains
\begin{align}
   [X_{\ell+1},\ldots,X_L ] - [ S_R^L, S_T^L, Y^{\ell-1}, X^\ell ] - Y_\ell.
\end{align}
The expression \eqref{eq:ibm-cap-f2} gives \eqref{eq:ibm-cap2}. Related results are available in~\cite[Sec.~III]{Caire-Shamai-IT99} and~\cite{Jelinek65,Das-Narayan-CISS98}.
\end{remark}

\begin{remark} \label{remark:CSITatR2}
The capacity \eqref{eq:ibm-cap-f} has only $S_R^L$ in the conditioning while \eqref{eq:ibm-cap2} has both $S_R^L$ and $S_T^L$ in the conditioning.  This subtle difference is due to permitting $X^{\ell-1}$ to influence the $S_{T\ell}$ in \eqref{eq:ibm-f}, and it complicates the analysis. On the other hand, if we remove only $X^{\ell-1}$ from \eqref{eq:ibm-f} then the receiver knows $S_{T\ell}$ at time $\ell$ and the capacity \eqref{eq:ibm-cap-f} can be written as (see the definition \eqref{eq:ccDI})
\begin{align}
   C = \max_{X_\ell(S_T^\ell), \;\ell=1,\ldots,L} \frac{1}{L} I(X^L \rightarrow Y^L \big\| S_T^L \big| S_R^L).
   \label{eq:ibm-cap-f3}
\end{align}
We treat such a model in Sec.~\ref{subsec:ibf-Rayleigh-Feedback} below.
\end{remark}

\subsection{Fading Channels with AWGN}
\label{subsec:ibf-fading-AWGN}
The expression~\eqref{eq:ibm-cap-f} is valid for general statistics. We next specialize to the block-fading AWGN model
\begin{align}
    Y_\ell = H X_\ell + Z_\ell
    \label{eq:Y-Gauss-block-fading}
\end{align}
where $\ell=1,\dots,L$, $Z^L\sim\mathcal{CN}(\ul 0,{\bf I})$, and $(H,S_R^L)$, $A^L$, $Z^L$ are mutually independent. Consider the power constraint 
\begin{align}
    \sum_{\ell=1}^L \E{P_\ell\left(S_T^\ell\right)} \le LP
    \label{eq:ibm-power}
\end{align}
where $P_\ell(s_T^\ell)=\E{|X_\ell(s_T^\ell)|^2}$. The optimization of~\eqref{eq:ibm-cap-f} under the constraint \eqref{eq:ibm-power} is usually intractable, and we again desire expressions with $\log(1+\text{SNR})$ terms to obtain insight.

\subsubsection{Capacity Upper Bound}
Using similar steps as in~\eqref{eq:IAYSR-UB}, we have
\begin{align}
  & I( A^L ; Y^L | S_R^L) \le  I( A^L ; Y^L, H \,|\, S_R^L) \nonumber \\
  & = \sum_{\ell=1}^L  I\left( A^L ; Y_\ell \,\big|\, S_R^L, H, Y^{\ell-1}\right) \nonumber \\
  & \le \sum_{\ell=1}^L \left[ h( Y_\ell | S_R^L, H, Y^{\ell-1}) - h(Z_\ell) \right] \nonumber \\
  & \overset{(a)}{\le} \sum_{\ell=1}^L \E{\log \left( 1 + \E{ G P_\ell(S_T^\ell) \big| S_R^L, H, Y^{\ell-1}} \, \right)}
  \label{eq:A-cap-SR-AWGN-UB}
\end{align}
where $G=|H|^2$ and step $(a)$ follows by~\eqref{eq:h-UB4}. However, CSCG inputs do not necessarily maximize the RHS of~\eqref{eq:A-cap-SR-AWGN-UB} because the inputs affect the CSIT.

\begin{remark} \label{remark:ibf-special}
The expectation inside the logarithm in~\eqref{eq:A-cap-SR-AWGN-UB} becomes $G P_\ell(S_T^\ell)$ if  $S_T^\ell$ is a function of $S_R^L,H,Y^{\ell-1}$; see~\eqref{eq:Shannon-cap-SR-AWGN-UB}, Remark~\ref{remark:CSITatR2}, and Proposition~\ref{proposition:part-CSIT2-block-fading} below. 
\end{remark}

\subsubsection{Achievable Rates}
Deriving achievable rates is more subtle than in Sec.~\ref{sec:Gauss-fading}. Consider the CSIT model \eqref{eq:ibm-f} where for each block, we have
\begin{align*}
   S_{T\ell}=f_{T\ell}( H, X^{\ell-1}, Y^{\ell-1} )
\end{align*}
for all $\ell$. The capacity \eqref{eq:ibm-cap-f} is
\begin{align}
   C(P) & = \max_{X_\ell(S_T^\ell), \;\ell=1,\ldots,L}  \frac{1}{L} I(X^L \rightarrow Y^L \big| H)
    \label{eq:ibm-cap-f4} \\
   & = \max_{X_\ell(S_T^\ell), \;\ell=1,\ldots,L}  \left[ \frac{1}{L} h(Y^L \big| H) \right] - \log(\pi e).
   \label{eq:ibm-cap-f41}
\end{align}
However, CSCG inputs are not necessarily optimal since the inputs affect the CSIT.

Instead of trying to optimize the input, consider $X_\ell$ that are CSCG. We may write
\begin{align}
   I(X^L \rightarrow Y^L | H) = \sum_{\ell=1}^L \E{\log\left( 1 + G P_\ell(S_T^{\ell}) \right)}
   \label{eq:ibm-cap-f5}
\end{align}
and the Lagrangians to maximize \eqref{eq:ibm-cap-f5} are
\begin{align}
   \sum_{\ell=1}^L \E{\log\left( 1 + G P_\ell(S_T^\ell) \right)} + \lambda \left(LP - \sum_{\ell=1}^L \E{P_\ell(S_T^{\ell})} \right).
   \label{eq:ibm-Lagrange}
\end{align}
Suppose the $S_{T\ell}$ are discrete random variables. Taking the derivative with respect to $P_\ell(s_T^\ell)$, we obtain
\begin{align}
   & \lambda 
   = \int_0^\infty p(g | s_T^\ell) \frac{g}{1+ g P_\ell(s_T^\ell)} dg \nonumber \\
   & \; + \sum_{k=\ell+1}^L \sum_{s_T^k} \int_0^\infty p(g)
   \frac{dP_{S_T^k | G}( s_T^k | g)}{dP_\ell(s_T^\ell)}
   \frac{\log\left( 1 + g P_k(s_T^k) \right)}{P_{S_T^\ell}(s_T^\ell)} dg \label{eq:ibm-lambda}
\end{align}
as long as $P_\ell(s_T^\ell)>0$. This expression is complicated
because the choice of transmit powers $P_\ell(s_T^\ell)$ influences the statistics of the future CSIT $S_{T(\ell+1)},\dots,S_{TL}$.  If \eqref{eq:ibm-lambda} cannot be satisfied, choose $P_\ell(s_T^\ell)=0$. Finally, set $\lambda$ so that $\sum_{\ell=1}^L \E{P_\ell(S_T^\ell)}=LP$.

Instead of the above, consider the simpler CSIT model with $S_{T\ell}=f_{T\ell}( H )$ for all $\ell$, cf.~\eqref{eq:ibm-f0}. The capacity \eqref{eq:ibm-cap2} is now given by \eqref{eq:ibm-cap-f5} with CSCG inputs and \eqref{eq:ibm-lambda} simplifies because the derivatives with respect to $P_\ell(s_T^\ell)$ are zero, i.e., the double sum in~\eqref{eq:ibm-lambda} disappears and for all $\ell$ and $s_T^\ell$ we have
\begin{align}
   & \lambda 
   = \int_0^\infty p(g | s_T^\ell) \frac{g}{1+ g P_\ell(s_T^\ell)} dg.
   \label{eq:ibm-lambda2}
\end{align}
We use~\eqref{eq:ibm-lambda2} for~\eqref{eq:dqf-1}--\eqref{eq:dqf-3} in Sec.~\ref{subsec:ibf-Rayleigh-Feedback} below.

\subsection{Full CSIR, Partial CSIT}
\label{subsec:ibf-full-CSIR-part-CSIT}
We next generalize Proposition~\ref{proposition:part-CSIT2} in Sec.~\ref{subsec:Full-CSIR-part-CSIT} to the block-fading AWGN model \eqref{eq:Y-Gauss-block-fading} with the CSIR 
\begin{align}
    S_{R\ell} = H\sqrt{P(S_T^\ell)}, \quad \ell=1,\dots,L
    \label{eq:ibm-full-CSIR}
\end{align}
and where $S_{T\ell}=f_{T\ell}(S_H)$, i.e., we have discarded $X_i^{\ell-1}$ and $Y_i^{\ell-1}$ in~\eqref{eq:ibm-ST}.  We then have the following capacity result that implies this CSIR is the best possible since one achieves a capacity upper bound similar to \eqref{eq:Shannon-cap-SR-AWGN-UB}.

\begin{proposition}
\label{proposition:part-CSIT2-block-fading}
The capacity of the channel \eqref{eq:Y-Gauss-block-fading} with the CSIR~\eqref{eq:ibm-full-CSIR} and $S_{T\ell}=f_{T\ell}(S_H)$ for $\ell=1,\dots,L$ is
\begin{align}
  C(P)
  & = \max \frac{1}{L} \sum_{\ell=1}^L \E{\log \left( 1 + G P_\ell(S_T^\ell) \right)}
  \label{eq:GMI-AYH2-block-fading} 
\end{align}
where the maximization is over the power control policies $P_\ell(S_T^\ell)$ such that $\sum_{\ell=1}^L \E{P_\ell(S_T^\ell)}\le LP$. One may use \eqref{eq:ibm-lambda2} to compute the $P_\ell(S_T^\ell)$.
\end{proposition}
\begin{proof}
For achievability, apply \eqref{eq:AGMI-3-vector-linear-block-fading} with
\begin{align*}
    \tilde P_\ell(\ul S_R)=G P_\ell(S_T^\ell) \text{ and } \E{|Y_\ell|^2|\ul S_R}=1+\tilde P_\ell(\ul S_R).
\end{align*}
The converse follows by applying similar steps as in~\eqref{eq:IAYSR-UB}:
\begin{align}
  & I( A^L ; Y^L | S_R^L) \le I( A^L ; Y^L, S_T^L, H | S_R^L) \nonumber \\
  & = \sum_{\ell=1}^L  I\left( A^L ; Y_\ell \,\big|\, S_R^L, S_T^L, H, Y^{\ell-1} \right) \nonumber \\
  & \le \sum_{\ell=1}^L \left[ h( Y_\ell | S_R^L, S_T^L, H, Y^{\ell-1}) - h(Z_\ell) \right] \nonumber \\
  & \overset{(a)}{\le} \sum_{\ell=1}^L \E{\log \Var{Y_\ell | S_R^L, S_T^L, H, Y^{\ell-1}}}.
  \label{eq:A-cap-SR-AWGN-UB2}
\end{align}
Finally, insert $\Var{Y_\ell | S_R^L, S_T^L, H, Y^{\ell-1}}=1+G P_\ell(S_T^\ell)$.
\end{proof}

The RHS of~\eqref{eq:A-cap-SR-AWGN-UB2} is at most the RHS of~\eqref{eq:A-cap-SR-AWGN-UB}, and hence~\eqref{eq:A-cap-SR-AWGN-UB2} gives a better bound. However, the bound~\eqref{eq:A-cap-SR-AWGN-UB2} is valid only for particular CSIT, as in Remark~\ref{remark:ibf-special}.

\vspace{-0.15cm}
\subsection{On-Off Fading with Delayed CSIT}
\label{subsec:ibf-Simple-Fading}
Consider on-off fading where the CSIT is delayed by $D$ symbols, i.e., we have $S_{T\ell}=0$ for $\ell=1,\dots,D$ and $S_{T(D+1)}=H$. Define the transmit powers as $P_\ell(s_T^\ell)=\E{|X(s_T^\ell)|^2}$ for $\ell=1,\dots,L$. The capacity is
\begin{align*}
   C(P) & =\frac{D}{2L} \log\left( 1 + 2P_1 \right) +  \frac{L-D}{2L} \log\left( 1 + 2P_{D+1} \right)
\end{align*}
where we write $P_{D+1}=P_{D+1}\big(s_T^{D+1}\big)$.
Optimizing the powers, we obtain 
\begin{align*}
   & \left\{
   \begin{array}{l}
   P_1 = P - \frac{L-D}{4L} \\
   P_{D+1} = 2P + \frac{D}{2L}
   \end{array} \right\} \quad \text{if } P \ge \frac{L-D}{4L} \\
   & \left\{
   \begin{array}{l}
   P_1 = 0 \\
   P_{D+1} = \frac{2LP}{L-D}
   \end{array} \right\} \quad \text{else.} 
\end{align*}
For large $P$, we thus have $C(P)\approx \frac{1}{2} \log(P)$ for all $0\le D\le L$.
For small $P$, we have 
\begin{align*}
  C(P) & = \left\{ \begin{array}{ll}
  \frac{L-D}{2L} \log\left( 1 + \frac{4LP}{L-D} \right), & \text{if $0\le D<L$} \\
  \log(1+2P)/2, & \text{if $D=L$}
  \end{array} \right. \nonumber \\
  & \approx \left\{ \begin{array}{ll}
  \left( 2P - \frac{4L}{L-D} P^2 \right) \log(e), & \text{if $0\le D<L$} \\
  \left( P - P^2 \right) \log(e), & \text{if $D=L$}.
  \end{array} \right.
\end{align*}
The CSIT thus gives a 3 dB power gain at low SNR since $C(P) \approx 2P \log(e)$ for $0\le D<L$ and $C(P) \approx P \log(e)$ for $D=L$. Furthermore, using~\eqref{eq:wideband}, the slope of the capacity versus
$E_b/N_0$ in bits/s/Hz/(3 dB) is
\begin{align*}
\begin{array}{ll}
   1-D/L& \text{if $0\le D<L$}  \\
   1 & \text{if $D=L$.} 
\end{array}
\end{align*}
In other words, the delay reduces the low-SNR rate by a factor of $1-D/L$ for $0 \le D<L$.

\vspace{-0.15cm}
\subsection{Rayleigh Fading and One-Bit Feedback}
\label{subsec:ibf-Rayleigh-Feedback}
Let $q_u(.)$ be the one-bit ($B=1$) quantizer in Sec.~\ref{subsec:quantizer}. We study Rayleigh fading for two scenarios with $S_R^L=H$,
i.e., the receiver knows $H$ after the $L$ transmissions of each block.
\begin{itemize}
\item For the CSIT \eqref{eq:ibm-f0}, we study delayed feedback where $S_{T\ell}=0$ for $\ell=1,\dots,L-1$ and $S_{TL}=q_u(G)$. The delay is thus $D=L-1$ in the sense of Sec.~\ref{subsec:ibf-Simple-Fading}.
\item For the CSIT \eqref{eq:ibm-f}, we study  the case $S_{T1}=0$, $S_{T2}=q_u(|Y_1|)$, and $S_{T\ell}=0$ for $\ell=3,\dots,L$. The delay is thus $D=1$ in the sense of Sec.~\ref{subsec:ibf-Simple-Fading}.
\end{itemize}

\subsubsection{Delayed Quantized CSIR Feedback}
Consider $S_{T\ell}=0$ for $\ell=1,\dots,L-1$ and $S_{TL}=q_u(G)$. CSCG inputs are optimal, and \eqref{eq:ibm-cap-f2} has the same form as \eqref{eq:GMI-AYH2-block-fading}. The Lagrangians are given by \eqref{eq:ibm-Lagrange}, and we again obtain \eqref{eq:ibm-lambda2}. For the case at hand, we have $L+1$ equations for $\lambda$, namely
\begin{align}
   \lambda & = \int_0^{\infty} e^{-g} \frac{g}{1+ g P_{\ell}} \,dg, \quad \ell=1,\dots,L-1 \label{eq:dqf-1} \\
   \lambda & = \int_0^{\Delta} \frac{e^{-g}}{1-e^{-\Delta}} \frac{g}{1+ g P_L(\Delta/2)} \,dg \label{eq:dqf-2} \\
   \lambda &= \int_{\Delta}^{\infty} \frac{e^{-g}}{e^{-\Delta}} \dfrac{g}{1+ g P_L(3\Delta/2)} \,dg \label{eq:dqf-3} 
\end{align}
where we used \eqref{eq:PST}--\eqref{eq:pgs2} and abused notation by writing $P_L(s_{TL})$ for $P_L(s_T^L)$. We thus have $P_1=\dots=P_{L-1}$ and obtain three equations. We now search for $\lambda$ such that
\begin{align*}
	(L-1)P_1 + \sum_{s} P_{S_{TL}}(s) P_L(s)=LP 
\end{align*}
and the capacity \eqref{eq:ibm-cap-f4} is
\begin{align}
   C(P) & =  \frac{L-1}{L} e^{1/P_1} E_1\left(1/P_1 \right) \nonumber \\
   & \quad + \frac{1}{L} \sum_{s} \int_{\set{I}(s)} e^{-g} \log\left( 1 + g P_L(s) \right) dg
   \label{eq:Shannon-cap7}
\end{align}
where the sums are over $s=\Delta/2,3\Delta/2$ and
\begin{align*}
    \set{I}(\Delta/2)=[0,\Delta), \quad \set{I}(3\Delta/2)=[\Delta,\infty).
\end{align*}
We remark that, if $P_1=0$, then we set $e^{1/P_1} E_1\left(1/P_1 \right)=0$ since
$\lim_{x\rightarrow\infty} e^x E_1(x) = 0$.

Fig.~\ref{fig:rf8} shows these capacities for $L=1,2,3$ and $\Delta=1$. At low SNR (e.g. for $L=3$ below $-2.97$ dB) we have $P_1=0$ and $P_L(\Delta/2)=0$, i.e., the transmitter is silent unless $S_{TL}=3\Delta/2$ and it uses power at time $\ell=L$ only. Observe that, as in Sec.~\ref{subsec:ibf-Simple-Fading}, a delay of $L$ steps reduces the low-SNR slope, and therefore the low-SNR rates, by a factor of $L$. Delay can thus be costly at low SNR.

\begin{figure}[t]
      \centering
      \includegraphics[width=0.95\columnwidth]{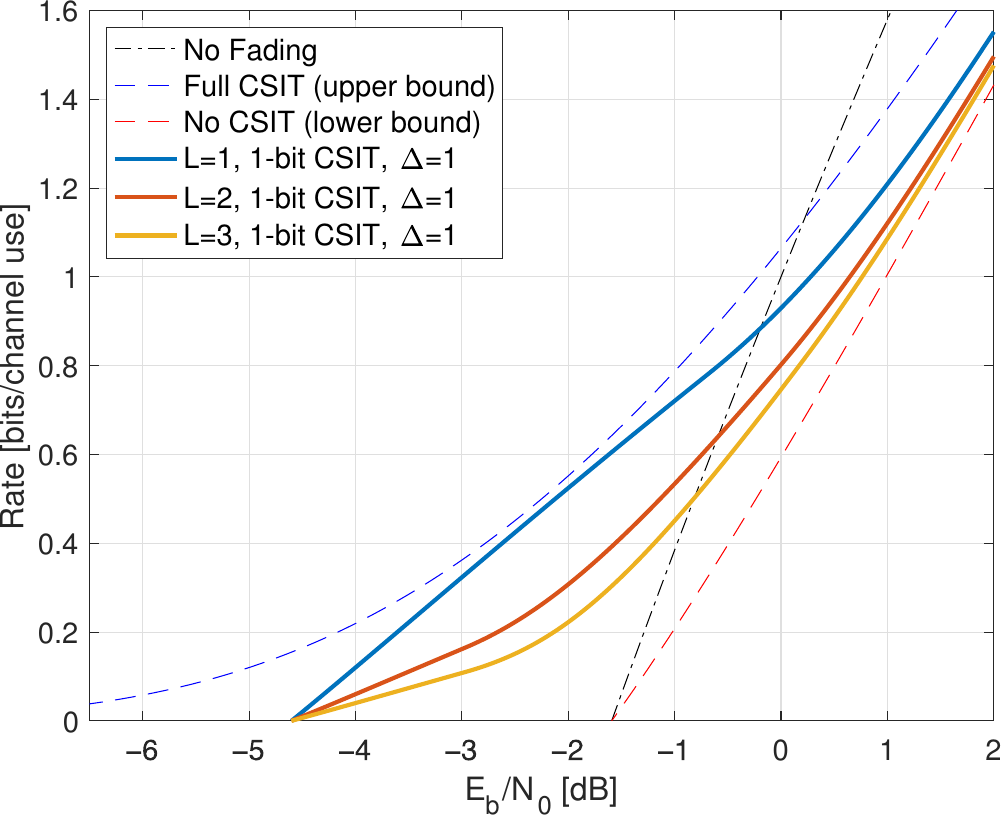}
      \caption{Capacities for Rayleigh block fading with $L=1,2,3$ and a CSIT delay of $D=L-1$.
      The CSIT at symbol $L$ is $S_{TL}=q_u(G)$.}
      \label{fig:rf8}
\end{figure}

\subsubsection{Quantized Channel Output Feedback}
%
Consider $S_{T1}=0$, $S_{T2}=q_u(|Y_1|)$, and $S_{T\ell}=0$ for $\ell=3,\dots,L$. As discussed in Remark~\ref{remark:CSITatR2}, the capacity is given by the directed information expression~\eqref{eq:ibm-cap-f3}. However, optimizing the input statistics seems difficult, i.e., CSCG inputs are not necessarily optimal. Instead, we compute achievable rates for a strategy where one symbol partially acts as a pilot.

Suppose the transmitter sends $X_1=\sqrt{P_1} e^{j\Phi}$ as the first symbol of each block, where $\Phi$ is uniformly distributed in $[0,2\pi)$. The idea is that $|X_1|=\sqrt{P_1}$ is known at the receiver, and thus $X_1$ acts as a pilot to test the channel amplitude.
Next, we choose a variation of flash signaling. Define the event $\set E =\{|Y_1| \ge \Delta \} = \{S_{T2}=3\Delta/2\}$. If this event does not occur, the transmitter sends $X_\ell=0$ for $\ell=2,\dots,L$. Otherwise, the transmitter sends independent CSCG $X_\ell$ with variance $P_2/\Pr{\set E}$ for $\ell=2,\dots,L$. Define $P_\ell(s_T^\ell)=\E{|X(s_T^\ell)|^2}$.  We have $P_\ell=P_2$ for $\ell\ge2$ and the power constraint is $P_1+(L-1)P_2 \le LP$.

We use \eqref{eq:ibm-cap-f2} to write
\begin{align}
   & C(P) \ge \frac{1}{L} I(X_1 ; Y_1 | H) + \frac{L-1}{L} I(X^2 ; Y_2 | H, Y_1).
   \label{eq:ibm-cap-LB}
\end{align}
The first mutual information in \eqref{eq:ibm-cap-LB} is
\begin{align*}
   I(X_1 ; Y_1 | H) & = h(Y_1 | H) - \log(\pi e)
\end{align*}
and we compute (see~\cite[App.~A]{essiambre2010jlt})
\begin{align*}
   p(y_1| h) = \frac{1}{\pi} e^{-(|y_{1}|^2 + P_1 |h|^2) } I_0\left( 2 \, |y_{1}| \, |h| \, \sqrt{P_1} \right)
\end{align*}
where $I_0(.)$ is the modified Bessel function of the first kind of order zero. The Jacobian of the mapping from Cartesian coordinates $[\Re(y_1),\Im(y_1)]$ to polar coordinates $[|y_1|,\arg y_1]$ is $|y_1|$, so we have
\begin{align*}
   h(Y_1|H=h) = \int_0^\infty - \, p(y_1 | h)  \log(p(y_1|h)) \, 2 \pi |y_1| \, d|y_1| .
\end{align*}
We further compute
\begin{align}
  & I\left(X^2 ; Y_2 | H, Y_1 \right) 
  \nonumber \\
  & = \int_0^\infty e^{-g} \Pr{\set{E}|G=g}
      \log\left( 1 + \frac{g P_2}{\Pr{\set{E}}} \right) dg.
\end{align}
The conditional probability of a high-energy $Y_1$ is
\begin{align*}
  \Pr{\set{E} | G=g} = Q_1\left( \sqrt{2g P_1}, \sqrt{2} \Delta \right)
\end{align*} 
where $Q_1(.)$ is the Marcum Q-function of order 1; see \eqref{eq:Marcum-Q2} in Appendix~\ref{appendix:non-central-chi-squared-2}.
For Rayleigh fading, we compute
\begin{align*}
  \Pr{\set{E}} 
  = \Pr{\left| H\sqrt{P_1}e^{j\Phi}+ Z_1 \right|^2 \ge \Delta^2} = e^{-\Delta^2 / (P_1+1)}.
\end{align*} 

The resulting rates are shown in Fig.~\ref{fig:rf9} for the block lengths $L=10,20,100$. Observe that each curve turns back on itself, which reflects the non-concavity of the directed information rates in $P$; see~\cite[Sec.~III]{kramer_SE_2003}. All rates below the curves are achievable by ``time-wasting'', i.e., by transmitting for some fraction of the time only. This suggests that flash signaling~\cite{Verdu02} will improve the rates since one sends information by choosing whether to transmit energy.

\begin{figure}[t]
      \centering
      \includegraphics[width=0.95\columnwidth]{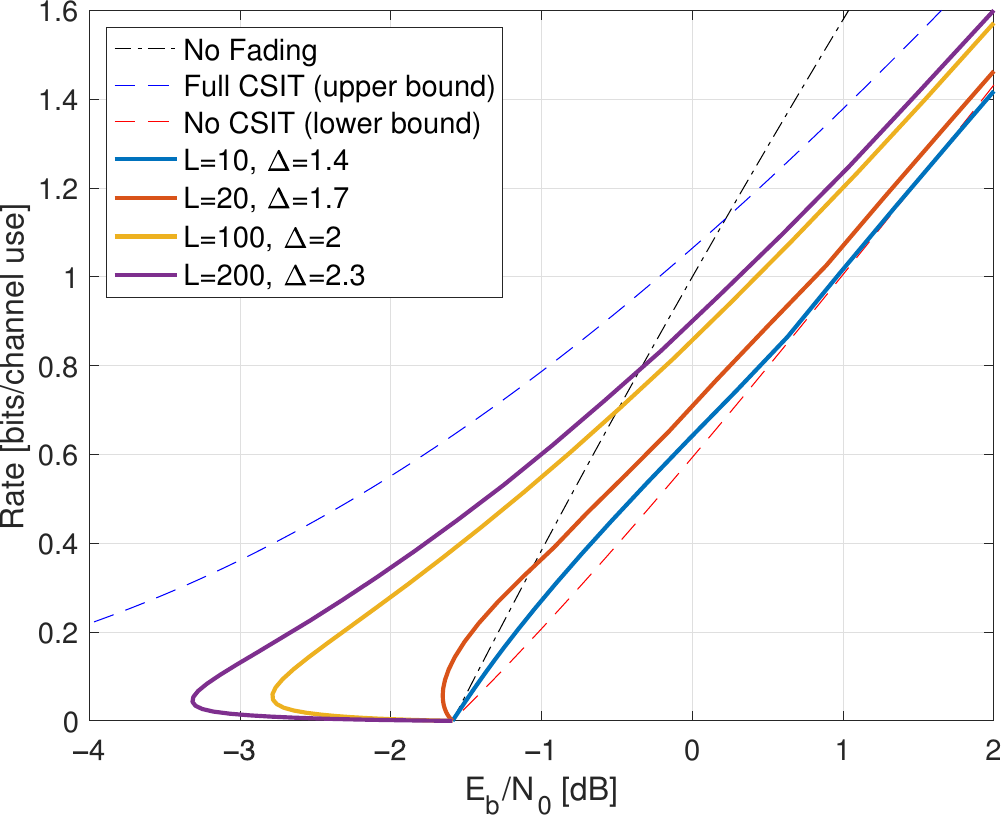}
      \caption{Rates for Rayleigh block fading with block lengths $L=10,20,100$.
      The CSIT at symbol 2 is $S_{T2}=q_u(|Y_1|)$.}
      \label{fig:rf9}
\end{figure}

\section{Conclusions}
\label{sec:Conclusion}

This paper reviewed and derived achievable rates for channels with CSIR, CSIT, block fading, and in-block feedback. GMI expressions were developed for adaptive codewords and two classes of auxiliary channel models with AWGN and CSCG inputs: reverse and forward channel models. The forward model inputs were chosen as linear functions of the adaptive codeword's symbols. We showed that, for scalar channels, an input distribution that maximizes the GMI generates a conventional codebook, where the codeword symbols are multiplied by a complex number that depends on the CSIT. The GMI increases by partitioning the channel output alphabet and modifying the auxiliary model parameters for each partition subset. The partitioning helps to determine the capacity scaling at high and low SNR. Power control policies were developed for full CSIT, including TMMSE policies. The theory was applied to channels with on-off fading and Rayleigh fading. The capacities with in-block feedback simplify to directed information expressions if the CSIT is a function of the CSIR and past channel inputs and outputs.

There are many possible applications and extensions of this work. For example, adaptive coding and modulation are important for all practical communication systems, including wireless, copper, and fiber-optic networks. Shannon's adaptive codewords can improve current systems since the CSIT is usually a noisy version of the CSIR; see Remark~\ref{remark:MLC-MSD}. Moreover, the information theory for in-block feedback~\cite{Kramer14} applies to beamforming~\cite{VanVeen-ASSP88} and intelligent reflecting surfaces~\cite{Lisakos-CM18,DiRenzo-JWCN19}. One may also apply GMI to multi-user channels with in-block feedback, such as multi-access and broadcast channels. Finally, it is important to develop improved capacity upper bounds. The standard approach here is the duality framework described in~\cite{Lapidoth:03,Thangaraj-IT17}; see also~\cite[p.~128]{Csiszar81}.

\section*{Acknowledgements}
The author wishes to thank the reviewers for their helpful comments and W. Zhang for sending his recent paper~\cite{Wang-Zhang-IT22}.

\setcounter{section}{0}
\renewcommand{\thesection}{\Alph{section}}
\renewcommand{\appendix}[1]{%
  \refstepcounter{section}%
  \par\begin{center}%
    \begin{sc}%
      Appendix \thesection\par\nobreak%
      #1%
    \end{sc}%
  \end{center}\nobreak%
}

\nopagebreak
\appendix{Special Functions}
\label{appendix:special-functions}

This appendix reviews three classes of functions that we use to analyze information rates: the non-central chi-squared distribution, the exponential integral, and gamma functions.

\subsection{Non-Central Chi-Squared Distribution}
\label{appendix:non-central-chi-squared-2}

The non-central chi-squared distribution with two degrees of freedom is the probability distribution of $Y=|x+Z|^2$ where $x\in\mathbb C$ and $Z\sim\mathcal{CN}(0,2)$. The density is
\begin{align}
   p(y) = \frac{1}{2} e^{-(y+|x|^2)/2} I_0(|x| \sqrt{y}) \cdot 1(y \ge 0)
   \label{eq:chi-2-zero}
\end{align}
where $I_0(.)$ is the modified Bessel function of the first kind of order zero. The cumulative distribution function is
\begin{align}
  \Pr{Y\le t} = 1 - Q_1\left(|x|, \sqrt{t} \right)
  \label{eq:Marcum-Q}
\end{align} 
where $Q_1(.)$ is the Marcum Q-function of order 1. Observe that if we change $Z$ to $Z\sim\mathcal{CN}(0,\sigma^2)$ then for $Y=|x+Z|^2$ we instead have
\begin{align}
  \Pr{Y\le t} = 1 - Q_1\left(\sqrt{2 |x|^2 / \sigma^2}, \sqrt{2t / \sigma^2} \right).
  \label{eq:Marcum-Q2}
\end{align}

\subsection{Exponential Integral}
\label{appendix:cHandIndep}

The exponential integral is defined for $x>0$ as
\begin{align}
   E_1(x)=\int_x^{\infty} \frac{e^{-t}}{t} dt.
   \label{eq:expint}
\end{align}
The derivative of $E_1(x)$ is
\begin{align}
   \frac{d E_1(x)}{dx} = \frac{-e^{-x}}{x}.
   \label{eq:expint-derivative}
\end{align}

For small $x$ one may apply~\cite[Eq.~(3)]{Cody68}
\begin{align}
  E_1(x) \approx -\gamma - \log x + x
  \label{eq:E1-approx-small-x}
\end{align}
where $\gamma \approx 0.57721$ is Euler's constant. For large $x$ we have
\begin{align}
  E_1(x) \approx \frac{e^{-x}}{x}\left(1 - \frac{1}{x} + \frac{2}{x^2} - \frac{6}{x^3} \right).
  \label{eq:E1-approx-large-x}
\end{align}
We have the bounds~\cite{Nantomah_2021}
\begin{align}
   \frac{1}{2} \log\left(1+\frac{2}{x} \right) & < e^x E_1(x) < \log\left(1+\frac{1}{x} \right) 
   \label{eq;expint-bounds0} \\
   \frac{1}{x+1} & < e^x E_1(x) < \frac{x+1}{x(x+2)}.
   \label{eq:expint-bounds}
\end{align}

Using integration by parts, for $x>0$, we have
\begin{align}
   & \int_x^{\infty} e^{-t} \log t \; dt = E_1(x) + e^{-x} \log(x)
   \label{eq:int1} \\
   & \int_x^{\infty} e^{-t} \frac{1}{t^2} \; dt = \frac{e^{-x}}{x} - E_1(x).
   \label{eq:int2}
\end{align}
Using the translation $\tilde t = t+y$ we also have
\begin{align}
   & \int_{x}^{\infty} e^{-t} \frac{t}{t+ y} dt = e^{-x} - y \, e^{y} \, E_1(x + y) 
   \label{eq:int3} \\
   & \int_{x}^{\infty} e^{-t} \frac{t}{(t+ y)^2} dt =  - e^{-x} \frac{y}{x+y} + (y+1) \, e^{y} \, E_1(x + y) 
   \label{eq:int4} \\
   & \int_{x}^{\infty} e^{-t} \frac{t^2}{(t+ y)^2} dt = e^{-x} \left( 1+ \frac{y^2}{x+y} \right) \nonumber \\
   & \qquad  \qquad \qquad  \qquad \qquad - y(y+2) \, e^{y} \, E_1(x + y) .
   \label{eq:int5}
\end{align}

\subsection{Gamma Functions}
\label{appendix:gamma}

The upper and lower incomplete gamma functions are the respective
\begin{align}
   \Gamma(s,t) = \int_{t}^{\infty} e^{-g} \, g^{s-1}  \, dg
   \label{eq:upper-incomplete-gamma} \\
  \gamma(s,t) = \int_{0}^{t} e^{-g} \, g^{s-1}  \, dg.
   \label{eq:lower-incomplete-gamma}
\end{align}
For instance, we have $\Gamma(1,t)=e^{-t}$ and $\gamma(1,t)=1-e^{-t}$. We further have $\Gamma(0,t)=E_1(t)$ where $E_1(x)$ is the exponential integral defined in Appendix~\ref{appendix:cHandIndep}.

The Gamma function is $\Gamma(s)=\Gamma(s,0)=\gamma(s,\infty)$
and for positive integers $n$ we have
\begin{align*}
   \Gamma(n) = (n-1)!, \quad
   \Gamma\left(n - \frac{1}{2} \right) = \frac{(2n-2)!}{4^{n-1} \, (n-1)!} \sqrt{\pi}.
\end{align*}
For example, the following cases are used in Sec.~\ref{subsec:Rayleigh-partial-CSIR-full-CSIT}:
\begin{align*}
   & \Gamma(1) = \Gamma(2) = 1 , \\
   & \Gamma\left(\frac{1}{2}\right) = \sqrt{\pi}, \quad 
   \Gamma\left(\frac{3}{2}\right) = \frac{\sqrt{\pi}}{2}, \quad 
   \Gamma\left(\frac{5}{2}\right) = \frac{3}{4} \sqrt{\pi} .
\end{align*}
The value $\Gamma(0)$ is undefined but we have $\lim_{x\rightarrow 0^+} \Gamma(x) = \infty$.

\appendix{Forward Model GMIs with $K=2$}
\label{appendix:2-gmi}

This appendix studies $K=2$ GMIs to develop high and low SNR capacity scaling results. Consider the independent random variables $Z \sim \mathcal{CN}(0,1)$ and $X \sim \mathcal{CN}(0,P)$. We need the following expression for the event $\set E=\{|X + Z|^2  \ge t_R\}$:
\begin{align}
    & \E{\left. |Z|^2 \right| \set E}
    = \int_{\mathbb C} p_{Z|\set E}(z) \, |z|^2 \, dz \nonumber \\
    & = \frac{1}{\Pr{\set E}} \int_{\mathbb C} \frac{e^{-|z|^2}}{\pi} \, |z|^2 \, \Pr{|X+z|^2 \ge t_R} \, dz \nonumber \\
   & = e^{t_R/(1+P)} \int_0^\infty e^{-g} \, g \, Q_1\left(\sqrt{\frac{2g}{P}},\sqrt{\frac{2t_R}{P}}\right) \, dg .
    \label{eq:expectation1}
\end{align}
The integral can be computed directly using~\cite[Eq.~(12)]{Sofotasios-SP15} with $k=2$, $m=1$, $p=1$, the Gamma functions above, and the following identities for Kummer's confluent hypergeometric function:
\begin{align*}
    _1F_1(1;2;z) = (e^z-1)/z, \quad _1F_1(2;2;z) = e^z .
\end{align*}
The result is
\begin{align}
    \E{\left. |Z|^2 \right| |X + Z|^2  \ge t_R } = 1 + \frac{t_R}{(1+P)^2}.
    \label{eq:expectation2}
\end{align}
Alternatively, define $Y=X+Z$ and $\Zt\sim\mathcal{CN}\left(0,\frac{P}{1+P}\right)$ independent of $Y$ so that $Z=Y/(1+P)+\Zt$. The expectation~\eqref{eq:expectation2} can then be written as
\begin{align*}
    \frac{\E{\left. |Y|^2 \right| |Y|^2  \ge t_R }}{(1+P)^2} + \E{\left| \Zt \right|^2}
    = \frac{1 + P + t_R}{(1+P)^2} + \frac{P}{1+P}.
\end{align*}

\subsection{On-Off Fading}
\label{appendix:gmi-oof}

Consider on-off fading as in Sec.~\ref{subsec:GMI-oof} and the $K=2$ partition in Remark~\ref{remark:GMI-part4} with $h_2=\sqrt{2}$. We compute
\begin{align}
   \Pr{\set E_2} & = \sum_{h=0,\sqrt{2}} \Pr{H=h} \Pr{\set E_2 | H=h }  \nonumber \\
   & = \frac{1}{2} e^{-t_R} + \frac{1}{2} e^{-t_R / (1+2P)}.
   \label{eq:app-gmi-off-PE2}
\end{align}
If $t_R = P^{\lambda_R}+b$ where $0<\lambda_R<1$ and $b$ is a real constant then $\Pr{\set E_2} \rightarrow 1/2$ as $P \rightarrow \infty$, as desired. We further have
\begin{align}
   \Pr{\left. H = 0 \, \right| \set E_2} & = \frac{e^{-t_R}}{2 \Pr{\set E_2}}  \\
   \Pr{\left. H = \sqrt{2} \, \right| \set E_2} & = \frac{e^{-t_R/ (1+2P)}}{2 \Pr{\set E_2}}.
\end{align}
The choice $t_R = P^{\lambda_R}+b$ gives $\Pr{\left. H = \sqrt{2} \, \right| \set E_2} \rightarrow 1$ as $P \rightarrow \infty$. In other words, the receiver can reliably determine $H$ by choosing $t_R$ to grow with $P$, but not too fast.

We next compute
\begin{align}
   & \E{|Y|^2 | \set E_2} = \sum_{h=0,\sqrt{2}} \Pr{H=h | \set E_2} \E{|Y|^2 | \set E_2, H=h } \nonumber \\
   & = \frac{e^{-t_R} (t_R+1) + e^{-t_R / (1+2P)} (t_R + 1 + 2P)}{2 \Pr{\set E_2}} .
   \label{eq:app-gmi-off-EY2}
\end{align}
The choice $t_R = P^{\lambda_R}+b$ makes $\E{|Y|^2 | \set E_2}/(1+2P) \rightarrow 1$ as $P \rightarrow \infty$. Finally, we compute
\begin{align}
   & \E{|Y - \sqrt{2} X|^2 | \set E_2} \nonumber \\
   & = \sum_{h=0,\sqrt{2}} \Pr{H=h | \set E_2} \E{\left. \left|Y - \sqrt{2} X\right|^2 \right|  \set E_2, H=h }  \nonumber \\
   & = \frac{1}{2 \Pr{\set E_2}} \bigg\{ e^{-t_R} (t_R+1+2P)  \nonumber \\
   & \qquad \qquad \quad \left. + \, e^{-t_R / (1+2P)} \left(1 + \frac{t_R}{(1+2P)^2}\right) \right\}
   \label{eq:app:gmi-oof-EY21}
\end{align}
where the last step uses~\eqref{eq:expectation2}. The choice $t_R = P^{\lambda_R}+b$ makes $\E{|Y - \sqrt{2} X|^2 | \set E_2} \rightarrow 1$ as $P \rightarrow \infty$.

\subsection{On-Off Fading, Partial CSIR, and Full CSIT}
\label{appendix:gmi-oof2}

The analysis for Sec.~\ref{subsec:oof-Partial-CSIR-full-CSIT} is similar to that of Appendix~\ref{appendix:gmi-oof}. Consider the GMI \eqref{eq:GMI-2-part-oof} and observe that we can replace $2P$ with $4P$ in \eqref{eq:app-gmi-off-PE2}--\eqref{eq:app-gmi-off-EY2}. We also have
\begin{align}
   & \E{\left. |Y - \sqrt{4P} \, U|^2 \right| \set E_2} \nonumber \\
   & = \sum_{h=0,\sqrt{2}} \Pr{H=h | \set E_2} \E{\left. |Y - \sqrt{4P} \, U |^2 \right.|  \set E_2, H=h }  \nonumber \\
   & = \frac{1}{2 \Pr{\set E_2}} \bigg\{ e^{-t_R} (t_R+1+4P) \nonumber \\
   & \qquad \qquad \quad \left. + \, e^{-t_R / (1+4P)} \left(1 + \frac{t_R}{(1+4P)^2}\right) \right\}.
   \label{eq:app:gmi-oof-EY21-b}
\end{align}
The choice $t_R = P^{\lambda_R}+b$ as in Appendix~\ref{appendix:gmi-oof} gives \eqref{eq:gmi-scaling-terms-oof}.

\subsection{On-Off Fading, Partial CSIR, and CSIT\at R}
\label{appendix:gmi-oof3}

The analysis for Sec.~\ref{subsec:oof-Partial-CSIR-CSITatR} is similar to that of Appendices~\ref{appendix:gmi-oof} and~\ref{appendix:gmi-oof2}. We compute
\begin{align}
   \Pr{\set E_2 | S_R=0} & = \epsb\, e^{-t_R} + \eps\, e^{-t_R / [1+2P(0)]} \\
   \Pr{\set E_2 | S_R=\sqrt{2}} & = \eps\, e^{-t_R} + \epsb\, e^{-t _R/ [1+2P(\sqrt{2})]} \;.
\end{align}
Suppose $P(0)$ and $P(\sqrt{2})$ both scale in proportion to $P$. If we choose $t_R = P^{\lambda_R}+b$ as in Appendix~\ref{appendix:gmi-oof} then $\Pr{\set E_2 | S_R=0} \rightarrow \eps$ and $\Pr{\set E_2 | S_R=\sqrt{2}} \rightarrow \epsb$ as $P \rightarrow \infty$. We also have
\begin{align}
   \Pr{\left. H = 0 \, \right| \set E_2, S_R=0} & = \frac{\epsb\,e^{-t_R}}{\Pr{\set E_2 | S_R=0}}  \\
   \Pr{\left. H = \sqrt{2} \, \right| \set E_2, S_R=0} & = \frac{\eps\,e^{-t_R / [1+2P(0)]}}{\Pr{\set E_2 | S_R=0}}
\end{align}
and similarly for the probabilities $\Pr{H=0 \,| \set E_2, S_R=\sqrt{2}}$ and $\Pr{H=\sqrt{2} \,| \set E_2, S_R=\sqrt{2}}$. Choosing $t_R = P^{\lambda_R}+b$ gives the desired behavior $\Pr{H = \sqrt{2} \,| \set E_2, S_R=0} \rightarrow 1$ and $\Pr{H=\sqrt{2} \,| \set E_2, S_R=\sqrt{2}} \rightarrow 1$ as $P \rightarrow \infty$. Again, the receiver can reliably determine $H$ by choosing $t_R$ to grow with $P$, but not too fast.

We next write $\E{|Y|^2 | \set E_2, S_R=0}$ as
\begin{align}
   \frac{\epsb\,e^{-t_R} (t_R+1) + \eps\,e^{-t_R / [1+2P(0)]} (t_R + 1 + 2P(0))}{\Pr{\set E_2 | S_R=0}} .
\end{align}
The expression for $\E{|Y|^2 | \set E_2, S_R=\sqrt{2}}$ is similar but $\eps$ and $\epsb$ are swapped and $P(0)$ is replaced with $P(\sqrt{2})$. We also have
\begin{align}
   & \E{\left. |Y - \sqrt{2} \, X(0)|^2 \right| \set E_2, S_R=0} \nonumber \\
   & = \frac{1}{\Pr{\set E_2 | S_R=0}} \bigg\{ \epsb\,e^{-t_R} (t_R+1+2P(0)) \,+ \nonumber \\
   & \qquad \left. \eps\,e^{-t_R / [1+2P(0)]} \left(1 + \frac{t_R}{(1+2P(0))^2}\right) \right\} .
   \label{eq:app:gmi-oof-EY21-c}
\end{align}
The expression for $\E{\left. |Y - \sqrt{2} \, X(0)|^2 \right| \set E_2, S_R=\sqrt{2}}$ is similar: swap $\eps$ and $\epsb$ and replace $P(0)$ with $P(\sqrt{2})$. The choice $t = P^{\lambda_R}+b$ makes all terms in \eqref{eq:AGMI-SR-part2-oof} behave as desired. We thus obtain \eqref{eq:GMI-2-example-oof2}.

\subsection{Rayleigh Fading, No CSIR, full CSIT, and TCI}
\label{appendix:gmi-rf1}

The analysis for Sec.~\ref{subsec:Rayleigh-partial-CSIR-full-CSIT} is similar to that of Appendices~\ref{appendix:gmi-oof} to \ref{appendix:gmi-oof3}, but we now have a continuous $H$. Recall that $\set E_2 = \{ |Y|^2 \ge t_R \}$ and $Y=\sqrt{P(h)} \, U + Z$ where $P(h)=0$ for $g<t$ and $P(h)=\hat P$ otherwise. We compute
\begin{align}
   \Pr{\set E_2} & = \Pr{G<t} \Pr{\set E_2 | G<t} \nonumber \\
   & \quad + \Pr{G\ge t} \Pr{\set E_2 | G\ge t} \nonumber \\
   & = (1 - e^{-t}) e^{-t_R} + e^{-t} e^{-t_R/(1+\hat P)}
   \label{eq:app-gmi-rf1-PE2}
\end{align}
where we used $\Pr{\set E_2 | G<t}=\Pr{|Z|^2 \ge t_R}$ and similarly  for $\Pr{\set E_2 | G\ge t}$. For example, for the $t$ and $t_R$ in \eqref{eq:gmi-rf1-parameters1} we find that $\Pr{\set E_2}\rightarrow 1$ as $P$ grows. Similarly, for the $t$ and $t_R$ in \eqref{eq:gmi-rf1-parameters2} we find that $\Pr{\set E_2}\approx e^{-t-1}$ as $P$ decreases.

We write
\begin{align}
   & \E{|Y|^2 | \set E_2} =
   \Pr{G<t | \set E_2} \E{\left. |Z|^2 \right| \set E_2, G<t} \nonumber \\
   & \qquad + \Pr{G\ge t | \set E_2} \E{\left. \left| \sqrt{\hat P}\, U + Z \right|^2 \right| \set E_2, G \ge t}  \nonumber \\
   & = \frac{(1-e^{-t}) e^{-t_R} (t_R+1) + e^{-t} e^{-t_R /(1+\hat P)} (t_R + 1 + \hat P)}{\Pr{\set E_2}}.
   \label{eq:app-gmi-rf1-EY2}
\end{align}
For the $t$ and $t_R$ in \eqref{eq:gmi-rf1-parameters1} we have $\E{|Y|^2 | \set E_2}/(1+\hat P) \rightarrow 1$ as $P$ grows. Similarly, for the $t$ and $t_R$ in \eqref{eq:gmi-rf1-parameters2} we find that $\E{|Y|^2 | \set E_2}/(1+2\hat P) \rightarrow 1$ as $P$ decreases. Next, we write
\begin{align}
   & \E{\left. \left| Y - \sqrt{\hat P}\, U \right|^2 \right| \set E_2} \nonumber \\
   & = \Pr{G<t | \set E_2} \E{\left. \left| Z - \sqrt{\hat P}\, U \right|^2 \right| |Z|^2 \ge t_R} \nonumber \\
   & \quad + \Pr{G\ge t | \set E_2} \E{|Z|^2 \left| \left|\sqrt{\hat P}\,U + Z\right|^2  \ge t_R \right.}  \nonumber \\
   & = \frac{1}{\Pr{\set E_2}} \bigg\{ (1-e^{-t}) e^{-t_R} (t_R+1+\hat P) + \nonumber \\
   & \quad \left. e^{-t} e^{-t_R /(1+\hat P)} \left(1 + \frac{t_R}{\left(1+\hat P\right)^2}\right) \right\} . \label{eq:app-gmi-rf1-EYPU2}
\end{align}
For the $t$ and $t_R$ in \eqref{eq:gmi-rf1-parameters1} the expression \eqref{eq:app-gmi-rf1-EYPU2} approaches 1 as $P$ grows. Similarly, for the $t$ and $t_R$ in \eqref{eq:gmi-rf1-parameters2} we find that \eqref{eq:app-gmi-rf1-EYPU2} approaches $1$ as $P$ decreases.

\appendix{Conditional Second-Order Statistics}
\label{appendix:c-second-order-statistics}

This appendix shows how to compute conditional second-order statistics for the reverse model GMIs and the forward model GMIs with $K=\infty$. Suppose $U,Y$ are jointly CSCG given $H=h$. Using \eqref{eq:LMMSE-est}--\eqref{eq:LMMSE-indep}, we have
\begin{align}
  & \E{U | Y=y, H=h} = \frac{ \E{U Y^\text{*} \big| H=h } }{ \E{|Y|^2 \big| H=h} } \cdot y  \label{eq:GMI-5-part-Es1} \\
  & \Var{U | Y=y, H=h} \nonumber \\
  & \quad = \E{|U|^2 \big| H=h } - \frac{ \left| \E{U Y^\text{*} \big| H=h} \right|^2 }{ \E{|Y|^2 \big| H=h} } .
  \label{eq:GMI-5-part-Es2}
\end{align}
Now consider the channel $Y=H X + Z$ where $X=\sqrt{P(S_T)}\, e^{j\phi(S_T)} U$ with $U\sim\mathcal{CN}(0,1)$. We may write
\begin{align}
    & \E{U \big| Y=y, S_R=s_R} \nonumber \\
    & = \int_{\mathbb C \times \set S_T} p(h,s_T | y,s_R) \, \frac{h^\text{*}\sqrt{P(s_T)}e^{j\phi(s_T)} y}{1+|h|^2 P(s_T)} \, ds_T \, dh
    \label{eq:AGMI-5-part-Es1a-app} 
\end{align}
and
\begin{align}
    & \E{|U|^2 \big| Y=y, S_R=s_R}
    = \int_{\mathbb C \times \set S_T} p(h,s_T | y,s_R) \nonumber \\
    & \quad \left( \frac{1}{1 + |h|^2 P(s_T)} + \frac{ |h|^2 P(s_T) |y|^2 }{\left(1 + |h|^2 P(s_T)\right)^2} \right) \, ds_T \, dh.
   \label{eq:AGMI-5-part-Es2a-app} 
\end{align}

\subsection{No CSIR, No CSIT}
\label{appendix:csos-1}

Consider $S_R=S_T=0$. The expectations in~\eqref{eq:AGMI-5-part-Es1a-app}--\eqref{eq:AGMI-5-part-Es2a-app} are computed via
\begin{align}
    p(h | y) = \frac{p(h) \, p(y|h)}{p(y)} .
    \label{eq:reverse-avg-density0}
\end{align}
The expression~\eqref{eq:AGMI-5-part-Es1a-app} with $\phi(0)=0$ gives
\begin{align}
  \E{U | Y=y} = \int_{\mathbb C} p(h | y) \, \frac{ h^\text{*} \sqrt{P} \, y }{1 + |h|^2 P} \, dh .
   \label{eq:GMI-5-part-Es1a} 
\end{align}
Similarly, the expression~\eqref{eq:AGMI-5-part-Es2a-app} gives
\begin{align}
  & \E{|U|^2\big|Y=y}  \nonumber \\
  & = \int_{\mathbb C} p(h | y) \, \E{\left. |X|^2 \right| Y=y, H=h} \, dh \nonumber \\
  & = \int_{\mathbb C} p(h | y) \left(\frac{1}{1 + |h|^2 P} + \frac{ |h|^2 P |y|^2 }{\left(1 + |h|^2 P\right)^2} \right) dh
   \label{eq:GMI-5-part-Es2a} 
\end{align}
We may now compute $\Var{U|Y=y}$ using~\eqref{eq:GMI-5-part-Es1a} and~\eqref{eq:GMI-5-part-Es2a}. For the expressions~\eqref{eq:I-LB2a} and~\eqref{eq:GMI-5-part}, one may use
\begin{align*}
    \E{X | Y=y} & = \sqrt{P} \, \E{U | Y=y}  \\
    \E{|X|^2\big|Y=y} & = P \, \E{|U|^2\big|Y=y}.
\end{align*}

For example, for on-off fading as in Sec.~\ref{subsec:GMI-oof} we compute
\begin{align}
    & \E{X | Y=y} =  P_{H|Y}\left(\left. \sqrt{2} \, \right | y \right) \frac{ \sqrt{2} P}{1 + 2 P} \cdot y \label{eq:GMI-5-part-Es1b}
    \\
    & \E{\left. |X|^2 \right| Y=y} = P_{H|Y}(0 | y ) \, P \nonumber \\
    & \qquad + P_{H|Y}\left(\left. \sqrt{2} \, \right | y \right) \left( \frac{P}{1 + 2 P} + \frac{2 P^2 |y|^2}{(1+2P)^2} \right)
  \label{eq:GMI-5-part-Es2b}
\end{align}
and therefore
\begin{align}
    & \Var{X | Y=y} =  P_{H|Y}(0 | y ) \, P \nonumber \\
    & + P_{H|Y}\left(\left. \sqrt{2} \, \right | y \right) \left( \frac{P}{1 + 2 P} + \frac{2 P^2 |y|^2}{(1+2P)^2} \, P_{H|Y}(0|y) \right)
  \label{eq:GMI-5-part-Es2c} 
\end{align}
where $P_{H|Y}\left(\left. \sqrt{2} \, \right | y \right) = 1 - P_{H|Y}(0 | y )$ and
\begin{align}
  P_{H|Y}(0 | y ) = \frac{e^{-|y|^2}}{e^{-|y|^2} + \frac{1}{1+2P} e^{-|y|^2/(1+2P)}} .
  \label{eq:GMI-5-part-PHY0}
\end{align}

For Rayleigh fading as in Sec.~\ref{subsec:Rayleigh-noCSIR-noCSIT}, the density~\eqref{eq:reverse-avg-density0} is
\begin{align*}
    p(h | y) = \frac{e^{-g} \, e^{-|y|^2/(1+g P)}}{\pi^2 (1+g P)} \cdot \frac{1}{p(y)}
\end{align*}
where $g=|h|^2$. Moreover, $p(y)$ in~\eqref{eq:r-py} depends on $g$ only. We thus have $\E{U|Y=y}=0$ and the integrand in~\eqref{eq:GMI-5-part-Es2a} depends on $g$ and $|y|^2$ only.

\subsection{Full CSIR, Partial CSIT}
\label{appendix:csos-2}
Consider $S_R=H$ and partial $S_T$. The expectations in~\eqref{eq:AGMI-5-part-Es1a-app}--\eqref{eq:AGMI-5-part-Es2a-app} are computed via~\eqref{eq:reverse-avg-density2} that we repeat here:
\begin{align*}
    p(h,s_T | y,s_R) = \delta(h-s_R) \,
    \frac{p(s_T|h) \, p(y | h,s_T)}{p(y|h)} .
\end{align*}

For on-off fading as in Sec.~\ref{subsec:oof-full-CSIR-partial-CSIT}, the expression~\eqref{eq:AGMI-5-part-Es1a-app} with $\phi(0)=0$ gives $\E{U|Y=y,H=0}=0$ and
\begin{align*}
    & \E{U \big| Y=y,H=\sqrt{2}} \nonumber \\
    & = \sum_{s_T=0,2} P_{S_T|Y,H}(s_T | y, \sqrt{2}\,) \,\frac{\sqrt{2P(s_T)} \, y}{1+2P(s_T)}
\end{align*}
and, similarly,~\eqref{eq:AGMI-5-part-Es2a-app} gives $\E{|U|^2|Y=y,H=0}=1$ and
\begin{align*}
    & \E{|U|^2 \big| Y=y,H=\sqrt{2}} = \sum_{s_T=0,2} \nonumber \\
    &  P_{S_T|Y,H}(s_T | y, \sqrt{2}\,) \, \left(\frac{1}{1+2P(s_T)} + \frac{2P(s_T) |y|^2}{\left(1+2P(s_T)\right)^2} \right)
\end{align*}
where $P_{S_T|Y,H}(2 | y, \sqrt{2}\,) = 1 - P_{S_T|Y,H}(0 | y, \sqrt{2}\,)$ and
\begin{align*}
  & P_{S_T|Y,H}(0 | y, \sqrt{2}\,) \\
  & = \frac{\frac{\epsilon}{1+2P(0)} \, e^{-|y|^2/(1+2P(0))}}{\frac{\epsilon}{1+2P(0)} \, e^{-|y|^2/(1+2P(0))} + \frac{\epsb}{1+2P(2)} \, e^{-|y|^2/(1+2P(2))}} .
\end{align*}

For Rayleigh fading as in Sec.~\ref{subsec:Rayleigh-full-CSIR-partial-CSIT}, the sums over $s_T=0,2$ become sums over $s_T=0,1$ and the probabilities $P(s_T | y, h)$ take on similar forms as above.

\subsection{Partial CSIR, Full CSIT}
\label{appendix:csos-3}
Consider $S_T=H$ and partial $S_R$. The expectations in~\eqref{eq:AGMI-5-part-Es1a-app}--\eqref{eq:AGMI-5-part-Es2a-app} are computed via~\eqref{eq:reverse-avg-density3} that we repeat here:
\begin{align*}
    p(h,s_T | y,s_R) = \delta(s_T-h) \,
    \frac{p(h|s_R) \, p(y | h,s_R)}{p(y|s_R)} .
\end{align*}

For on-off fading with $S_R=0$ as in Sec.~\ref{subsec:oof-Partial-CSIR-full-CSIT}, the expression~\eqref{eq:AGMI-5-part-Es1a-app} with $\phi(0)=0$ gives
\begin{align*}
    & \E{U \big| Y=y} = P_{H|Y}\left(\left. \sqrt{2} \, \right | y \right) \, \frac{\sqrt{4P} \, y}{1+4P}
\end{align*}
and~\eqref{eq:AGMI-5-part-Es2a-app} gives
\begin{align*}
    & \E{|U|^2 \big| Y=y} = P_{H|Y}(0 | y ) \nonumber \\
    & \qquad + \: P_{H|Y}\left(\left. \sqrt{2} \, \right | y \right) \, \left(\frac{1}{1+4P} + \frac{4P |y|^2}{(1+4P)^2} \right)
\end{align*}
where $P_{H|Y}\left(\left. \sqrt{2} \, \right | y \right) = 1 - P_{H|Y}(0 | y )$ and
\begin{align*}
  & P_{H|Y}(0 | y ) = \frac{e^{-|y|^2}}{e^{-|y|^2} + \frac{1}{1+4P} e^{-|y|^2/(1+4P)}} .
\end{align*}

For Rayleigh fading with $S_R=0$ and TCI as in Sec.~\ref{subsec:Rayleigh-partial-CSIR-full-CSIT}, the expressions~\eqref{eq:AGMI-5-part-Es1a-app}--\eqref{eq:AGMI-5-part-Es2a-app} give (cf.~\eqref{eq:GMI-5-part-Es1b}--\eqref{eq:GMI-5-part-Es2b})
\begin{align*}
    & \E{U \big| Y=y} = \Pr{G \ge t | Y=y} \frac{\sqrt{\hat P}\, y}{1+\hat P} \\
    & \E{|U|^2 \big| Y=y} = \Pr{G < t | Y=y} \nonumber \\
    & \qquad + \: \Pr{G \ge t | Y=y} \left(\frac{1}{1+\hat P} + \frac{\hat P\, |y|^2}{(1+\hat P)^2} \right)
\end{align*}
and therefore (cf.~\eqref{eq:GMI-5-part-Es2c})
\begin{align*}
    \Var{U | Y=y} & =  \Pr{G < t | Y=y} + \Pr{G \ge t | Y=y} \nonumber \\
    & \quad \cdot \left( \frac{1}{1 + \hat P} + \frac{\hat P \, |y|^2}{(1+\hat P)^2} \, \Pr{G < t | Y=y} \right)
\end{align*}
where (cf.~\eqref{eq:GMI-5-part-PHY0})
\begin{align*}
  & \Pr{G < t | Y=y} = \frac{\big( 1-e^{-t} \big)\, e^{-|y|^2}}{\big( 1-e^{-t} \big)e^{-|y|^2} + e^{-t} \frac{1}{1+\hat P} e^{-|y|^2/(1+ \hat P)}} .
\end{align*}

\subsection{Partial CSIR, CSIT\at R}
\label{appendix:csos-4}
Consider $S_T=S_R$ and partial $S_R$. The expectations in~\eqref{eq:AGMI-5-part-Es1a-app}--\eqref{eq:AGMI-5-part-Es2a-app} are computed via~\eqref{eq:reverse-avg-density4} that we repeat here:
\begin{align*}
    p(h,s_T | y,s_R) = \delta\big(s_T - f(s_R)\big) \,
    \frac{p(h|s_R) \, p(y | h,s_R)}{p(y|s_R)} .
\end{align*}

For on-off fading as in Sec.~\ref{subsec:oof-Partial-CSIR-full-CSIT}, the expression~\eqref{eq:AGMI-5-part-Es1a-app} with $\phi(0)=0$ gives
\begin{align*}
    \E{U \big| Y=y,S_R=0} & = P_{H|Y,S_R}(1 | y,0) \, \frac{\sqrt{2P(0)} \, y}{1+2P(0)} \\
    \E{U \big| Y=y,S_R=\sqrt{2}} & = P_{H|Y,S_R}(1 | y, \sqrt{2}\,) \, \frac{\sqrt{2P(\sqrt{2}\,)} \, y}{1+2P(\sqrt{2}\,)}
\end{align*}
and~\eqref{eq:AGMI-5-part-Es2a-app} gives
\begin{align*}
    & \E{|U|^2 \big| Y=y,S_R=0} = P_{H|Y,S_R}(0 | y,0) \nonumber \\
    & \qquad + P_{H|Y,S_R}(1 | y,0) \left(\frac{1}{1+2P(0))} + \frac{2P(0) |y|^2}{\left(1+2P(0)\right)^2} \right) \\
    & \E{|U|^2 \big| Y=y,S_R=\sqrt{2}} = P_{H|Y,S_R}(0 | y, \sqrt{2}\,) \nonumber \\
    & \qquad + P_{H|Y,S_R}(1 | y, \sqrt{2}\,) \left(\frac{1}{1+2P(\sqrt{2}\,)} + \frac{2P(\sqrt{2}\,) |y|^2}{\left(1+2P(\sqrt{2}\,)\right)^2} \right)
\end{align*}
where $P_{H|Y,S_R}(\sqrt{2}\, | y,s_R) = 1 - P_{H|Y,S_R}(0 | y,s_R)$ and
\begin{align*}
  & P_{H|Y,S_R}(0 | y,0)  = \frac{\epsb \, e^{-|y|^2}}{\epsb \, e^{-|y|^2} + \frac{\epsilon}{1+2P(0)} \, e^{-|y|^2/(1+2P(0))}} \\
  & P_{H|Y,S_R}(0 | y, \sqrt{2}\,)  = \frac{\epsilon \, e^{-|y|^2}}{\epsilon \, e^{-|y|^2} + \frac{\epsb}{1+2P(\sqrt{2}\,)} \, e^{-|y|^2/(1+2P(\sqrt{2}\,))}} .
\end{align*}
For Rayleigh fading as in Sec.~\ref{subsec:Rayleigh-partialCSIR-CSITatR}, the probabilities $P(h | y,s_R)$ take on similar forms as above.

\newpage
\appendix{Proof of Lemma~\ref{lemma:AGMI} and ~\eqref{eq:AGMI-matrix-3-part}}
\label{appendix:proof-lemma-AGMI}
We prove Lemma~\ref{lemma:AGMI} by using
the same steps as in the proof of Proposition~\ref{prop:GMI}. The GMI \eqref{eq:AGMI} with a vector $\ul Y$ is
\begin{align}
   & I_s(A;\ul Y) = \log\det \left( {\bf I} + \left( {\bf Q}_{\ul Z} / s \right)^{-1} {\bf H} {\bf Q}_{\ul{\bar X}} {\bf H}^\dag \right) \nonumber \\
   & \quad + \E{ {\ul Y}^\dag \left( {\bf Q}_{\ul Z}/s + {\bf H} {\bf Q}_{\ul{\bar X}} {\bf H}^\dag \right)^{-1} \ul Y } \nonumber \\
   & \quad - \E{ \left( \ul Y - {\bf H}\, \ul{\bar X} \right)^\dag \left( {\bf Q}_{\ul{Z}} / s \right)^{-1} \left( \ul Y - {\bf H}\, \ul{\bar X} \right) }.
   \label{eq:AGMI-4}
\end{align}
One can again set $s=1$. Choosing ${\bf H}=\tilde {\bf H}$ and ${\bf Q}_{\ul Z}=\tilde {\bf Q}_{\ul{\Zt}} $ then gives \eqref{eq:AGMI-3}.

Next, consider the channel $\ul Y_a = \tilde {\bf H} \ul{\bar X} + \ul{\Zt}$ where $\ul{\Zt}$ is CSCG with covariance matrix ${\bf Q}_{\ul{\Zt}}$ and $\ul{\Zt}$ is independent of $\ul {\bar X}$. Generalizing~\eqref{eq:aux-equiv1}--\eqref{eq:aux-equiv2}, we compute ${\bf Q}_{\ul Y_a} = {\bf Q}_{\ul Y}$ and
\begin{align}
   & \E{\left( \ul Y_a - \tilde {\bf H}\, \ul{\bar X} \right) \left( \ul Y_a - \tilde {\bf H}\, \ul{\bar X} \right)^\dag } \nonumber \\
   & \qquad = \E{\left( \ul{Y} - {\bf H}\, \ul{\bar X} \right) \left( \ul{Y} - {\bf H}\, \ul{\bar X} \right)^\dag }.
   \label{eq:AGMI-second-order-stats-22}
\end{align}
In other words, the second-order statistics for the two channels with outputs $\ul Y$ (the actual channel output) and $\ul Y_a$ are the same. Moreover, the GMI \eqref{eq:AGMI-3} is the mutual information $I(A;\ul Y_a)$. Using \eqref{eq:AGMI-identity} and \eqref{eq:AGMI-4}, for any $s$, $\bf H$ and ${\bf Q}_{\ul Z}$ we have
 \begin{align}
   I(A;\ul Y_a) & = \log\det \left( {\bf I} + {\bf Q}_{\ul{\Zt}}^{-1} \tilde {\bf H}
   {\bf Q}_{\ul{\bar X}} \tilde {\bf H}^\dag \right) \nonumber \\
   & \ge I_s(A;\ul Y_a) = I_s(A;\ul{Y})
   \label{eq:AGMI-bound}
\end{align}
and equality holds if ${\bf H}= \tilde {\bf H}$ and ${\bf Q}_{\ul Z}/s={\bf Q}_{\ul{\Zt}}$.

To prove~\eqref{eq:AGMI-matrix-3-part}, recall that $\tr{\bf AB}=\tr{\bf BA}$ for matrices $\bf A$ and $\bf B$ with appropriate dimensions. Furthermore, for Hermitian matrices $\bf A, B,C$ with the same dimensions we have
\begin{align}
    \tr{\bf A B C} = \tr{ ({\bf A B C})^\dag} = \tr{\bf C B A} = \tr{\bf A C B}.
    \label{eq:trace-triple}
\end{align}
For notational convenience, consider the covariance matrix~\eqref{eq:AGMI-matrix-parameters-part} with $s=1$ and use
\begin{align*}
    {\bf A} & = {\bf Q}_{\ul{\bar Z}}, \quad
    {\bf B} = \left( {\bf H} {\bf Q}_{\ul{\bar X}} {\bf H}^\dag \right)^{-1/2} \left( {\bf Q}_{\ul Y} 
    - {\bf Q}_{\ul{\bar Z}} \right)^{1/2} \\
    {\bf C} & = {\bf Q}_{\ul{\bar Z}}^{-1} \left( {\bf Q}_{\ul Y}
    - {\bf Q}_{\ul{\bar Z}} \right)^{1/2} \left( {\bf H} {\bf Q}_{\ul{\bar X}} {\bf H}^\dag \right)^{-1/2}
\end{align*}
to compute (cf.~\eqref{eq:AGMI-4})
\begin{align}
    & \E{\left( \ul Y - {\bf H}\, \ul{\bar X} \right)^\dag {\bf Q}_{\ul{Z}}^{-1} \left( \ul Y - {\bf H}\, \ul{\bar X} \right)}  = \tr{ {\bf Q}_{\ul{\bar Z}} {\bf Q}_{\ul{Z}}^{-1} } \nonumber \\
    & \overset{(a)}{=} \tr{\left( {\bf Q}_{\ul Y} - {\bf Q}_{\ul{\bar Z}} \right) \left( {\bf H} {\bf Q}_{\ul{\bar X}} {\bf H}^\dag \right)^{-1}}
\end{align}
where step $(a)$ follows by \eqref{eq:trace-triple}. Next, by using~\eqref{eq:AGMI-matrix-parameters-part} we have
\begin{align}
    & \left( {\bf Q}_{\ul Z} + {\bf H} {\bf Q}_{\ul{\bar X}} {\bf H}^\dag \right)^{-1} \nonumber \\
    & \quad = \left( {\bf Q}_{\ul Y} - {\bf Q}_{\ul{\bar Z}} \right)^{1/2} \left( {\bf H} {\bf Q}_{\ul{\bar X}} {\bf H}^\dag \right)^{-1/2} {\bf Q}_{\ul{Y}}^{-1} \nonumber \\
    & \qquad \cdot \left( {\bf H} {\bf Q}_{\ul{\bar X}} {\bf H}^\dag \right)^{-1/2} \left( {\bf Q}_{\ul Y}
    - {\bf Q}_{\ul{\bar Z}} \right)^{1/2}
\end{align}
and therefore (cf.~\eqref{eq:AGMI-4})
\begin{align}
    & \E{ {\ul Y}^\dag \left( {\bf Q}_{\ul Z} + {\bf H} {\bf Q}_{\ul{\bar X}} {\bf H}^\dag \right)^{-1} \ul Y } \nonumber \\
    & \quad = \tr{ {\bf Q}_{\ul Y} \left( {\bf Q}_{\ul Z} + {\bf H} {\bf Q}_{\ul{\bar X}} {\bf H}^\dag \right)^{-1} } \nonumber \\
    & \quad \overset{(a)}{=} \tr{\left( {\bf Q}_{\ul Y}
    - {\bf Q}_{\ul{\bar Z}} \right) \left( {\bf H} {\bf Q}_{\ul{\bar X}} {\bf H}^\dag \right)^{-1}}
\end{align}
where step $(a)$ again follows by \eqref{eq:trace-triple}. We are thus left with the logarithm term in \eqref{eq:AGMI-4}. Finally, the determinant in~\eqref{eq:AGMI-4} is
\begin{align}
    & \det \left( {\bf I}
    + {\bf Q}_{\ul Z}^{-1} \,
    {\bf H} {\bf Q}_{\ul{\bar X}} {\bf H}^\dag \right)  = \det \left( 
   {\bf Q}_{\ul{\bar Z}}^{-1} {\bf Q}_{\ul Y} 
   \right)
\end{align}
where we applied~\eqref{eq:AGMI-matrix-parameters-part} and Sylvester's identity~\eqref{eq:Sylvester}.

\appendix{Proof of Lemma~\ref{lemma:AGMI-max}}
\label{appendix:proof-lemma-AGMI-max}

Let $\bar P = \E{|\bar X|^2}$ and write
\begin{align}
   \bar X = \sqrt{\bar P} \, \bar U, \quad
   X(s_T) = \sqrt{P(s_T)} \, U(s_T) .
\end{align}
Since the $U(s_T)$ are CSCG we have
\begin{align}
  U(s_T') = \rho(s_T',s_T) \, U(s_T) + Z(s_T')
\end{align}
where $\rho(s_T',s_T)=\E{U(s_T') \, U(s_T)^\text{*}}$ and
\begin{align}
   Z(s_T')\sim\mathcal{CN}(0,1-|\rho(s_T',s_T)|^2)
\end{align}
is independent of $U(s_T)$. As in \eqref{eq:xbar}, define
\begin{align}
  \bar X & = \sum_{s_T'} w(s_T' ) X(s_T') \nonumber \\
  & = \sum_{s_T} w(s_T') \sqrt{P(s_T')} \left[ U(s_T)  \rho(s_T',s_T) + Z(s_T') \right] \nonumber \\
  & = \sqrt{\bar P} \, \bar \rho(s_T) \, U(s_T)  + \sum_{s_T'} w(s_T') \sqrt{P(s_T')} Z(s_T')
  \label{eq:xbar-expression}
\end{align}
where, assuming that $\bar P>0$, we have
\begin{align}
   \bar \rho(s_T) = \E{\bar U \, U(s_T)^\text{*}}
   = \sum_{s_T'} w(s_T') \sqrt{\frac{P(s_T')}{\bar P}} \, \rho(s_T',s_T) . \label{eq:rhobar} 
\end{align}
Observe that  $\sqrt{\bar P} \, \bar \rho(s_T) \, U(s_T)$ is the LMMSE estimate of $\bar X$ given $U(s_T)$.

Using Lemma~\ref{lemma:AGMI}, we have the auxiliary variables
\begin{align}
    \tilde h = \frac{\E{Y \bar X^\text{*}}}{\bar P}, \quad
   \tilde \sigma^2 = \E{|Y|^2} - | \tilde h |^2 \bar P
   \label{eq:AGMI-parameters-RkF}
\end{align}
and the GMI
\begin{align}
   I_1(A;Y) = \log \left( \frac{\E{|Y|^2}}{\E{|Y|^2} - | \tilde h |^2 \bar P} \right).
   \label{eq:appendix-AGMI}
\end{align}
If the $P(s_T)$ are fixed, then so is $\E{|Y|^2}$ because $U(s_T)$ is CSCG and independent of $Z$ given $S_T=s_T$.
The GMI \eqref{eq:appendix-AGMI} is thus maximized by maximizing $|\tilde h|^2 \bar P$. We compute
\begin{align}
   & | \tilde h |^2 \bar P
   = \left| \sum_{s_T} P_{S_T}(s_T) \frac{\E{\left. Y \bar X^\text{*} \right| S_T = s_T}}{\sqrt{\bar P}} \right|^2 \nonumber \\
   & \overset{(a)}{=}  \left| \sum_{s_T} P_{S_T}(s_T) \E{\left. Y \, U(s_T)^\text{*} \right| S_T=s_T}
   \bar \rho(s_T)^\text{*}  \right|^2 \label{eq:Lemma-AGMI-Ptilde1} \\
   & \le \left( \sum_{s_T} P_{S_T}(s_T) \left| \E{\left. Y \, U(s_T)^\text{*} \right| S_T=s_T} \right| \right)^2
   \label{eq:Lemma-AGMI-Ptilde2}
\end{align}
where step $(a)$ follows because we have the Markov chain $A-[U(S_T),S_T]-Y$ which implies that $Y$ and the $Z(s_T')$ in \eqref{eq:xbar-expression} are independent give $S_T=s_T$.

Equality holds in \eqref{eq:Lemma-AGMI-Ptilde2} if the summands in \eqref{eq:Lemma-AGMI-Ptilde1} all have the same phase and $|\bar \rho(s_T)|=1$ for all $s_T$. But this is possible by choosing $X(s_T)$ as given in \eqref{eq:barX-max} so that $U(s_T) = e^{j\phi(s_T)} \, U$. Moreover, choose the receiver weights as
\begin{align}
   w(\tilde s_T) = \sqrt{\frac{\bar P}{P(\tilde s_T)}} \, e^{-j\phi(\tilde s_T)}
   \label{eq:w-sT}
\end{align}
for one $\tilde s_T \in \set S_T$ with $P(\tilde s_T)>0$, and $w(s_T)=0$ otherwise. We then have $\bar X = \sqrt{\bar P}\, U$ and
\begin{align}
   \rho(s_T',s_T) = e^{j(\phi(s_T')-\phi(s_T))}, \quad 
   \bar \rho(s_T) = e^{-j\phi(s_T)}  \label{eq:rho_sT}
\end{align}
and the resulting maximal $I_1(A;Y)$ is given by \eqref{eq:AGMI-2}--\eqref{eq:AGMI-P}.

\begin{remark} \label{remark:independent}
The full correlation permits many choices for the $w(s_T)$; hence, these weights do not seem central to the design. However, including weights can be useful if the codebook is not designed for the CSIR. For example, suppose $A$ has independent entries $X(s_T)$ for which we compute
\begin{align}
   \bar \rho(s_T) & = \frac{w(s_T) \sqrt{P(s_T)}}{\sqrt{\sum_{s_T'}  |w(s_T')|^2 P(s_T') }}
\end{align}
and thus \eqref{eq:Lemma-AGMI-Ptilde1} becomes
\begin{align}
   \frac{\left| \sum_{s_T} P_{S_T}(s_T) \E{\left. Y X(s_T)^\text{*} \right| S_T=s_T} \, w(s_T)^\text{*} \right|^2}
   {\sum_{s_T}  |w(s_T)|^2 P(s_T)} .
   \label{eq:P-indep-codewords}
\end{align}
Using Bergstr\"om's inequality (or the Cauchy-Schwarz inequality), the expression
\eqref{eq:P-indep-codewords} is maximized by
\begin{align}
   w(s_T) = P_{S_T}(s_T) \, \frac{\E{\left. Y X(s_T)^\text{*} \right| S_T = s_T}}{P(s_T)} \cdot c
\end{align}
for some constant $c\ne 0$. The expression \eqref{eq:Lemma-AGMI-Ptilde1} is therefore
\begin{align}
   \sum_{s_T}  P_{S_T}(s_T|h)^2 \left|\E{ \left. Y U(s_T)^\text{*} \right| S_T=s_T} \right|^2
   \label{eq:P-indep-codewords2}
\end{align}
which is generally smaller than $\E{ \, \left| \E{\left. Y U(S_T)^\text{*} \right| S_T} \right| \, }^2$
(apply $\sum_i a_i^2 \le (\sum_i a_i)^2$ for non-negative $a_i$).
\end{remark}

\begin{remark}
\label{remark:CSIT-asymmetric}
The following example shows that more general signaling and more general $\bar X$ can be useful. Consider the channel with two equally-likely states $\set{S}_T=\{+1,-1\}$ and $Y = |X| \exp(j s_T \arg(X))+Z$. We compute
\begin{align*}
   & \E{Y \, U(+1)^\text{*} | S_T=+1} = \sqrt{P(1)} \\
   & \E{Y \, U(-1)^\text{*} | S_T=-1} = 0 \\
   & {\bar \rho}(+1) = \frac{w(1) \sqrt{P(1)} + w(-1) \sqrt{P(-1)} \rho(-1,+1)}{\sqrt{\bar P}}
\end{align*}
and one should choose $P(-1)=0$ and $P(1)=2P$ if the power constraint is $\E{P(S_T)}\le P$. We thus have
\begin{align*}
   \E{|Y|^2} = P+1, \quad 
   \tilde P = \frac{P}{2}
\end{align*}
and therefore \eqref{eq:AGMI-2} gives
\begin{align*}
   I_1(A;Y) = \log\left( 1 + \frac{P}{2 + P} \right).
\end{align*}
However, one can achieve the rate $\log(1+P)$ with other Gaussian $\bar X$, namely linear combinations of both the $X(s_T)$ and the $X(s_T)^\text{*}$ in \eqref{eq:xbar-expression}. This idea permits circularly asymmetric $\bar X$, also known as \emph{improper} $\bar X$ \cite{Neeser-Massey-IT93}. Alternatively, the transmitter can send the complex-conjugate symbols if $S_T=-1$.
\end{remark}

\appendix{Large $K$ for Sec.~\ref{subsec:GMI-generalizations2}}
\label{appendix:GMI-MMSE2}
We complete Remark~\ref{remark:large-K2} by proceeding as in Appendix~\ref{appendix:csos-1}. To generalize~\eqref{eq:GMI-5-part}, we must deal with unit-rank matrices $\ul y\, \ul y^\dag$ that do not have inverses. Consider first finite $K$. Conditioned on the event $\set E_k$, we may write
\begin{align}
    \ul Y = \ul y_k + \epsilon^{1/2} \ul{\Zt}_k
\end{align}
where $\ul y_k=\E{\ul Y | \set E_k}$ and $\E{\left. \ul{\Zt}_k \right| \set E_k}=\ul 0$. We abuse notation and write the conditional covariance matrix of $\ul{\Zt}_k$ as ${\bf Q}_{\ul{\Zt}_k}$, and we assume that ${\bf Q}_{\ul{\Zt}_k}$ is invertible. Define $\ul{\yt}_k={\bf Q}_{\ul{\Zt}_k}^{-1/2}\ul y_k$ and compute
\begin{align}
   {\bf Q}_{\ul Y}^{(k)} & = \epsilon \, {\bf Q}_{\ul{\Zt}_k}^{1/2} \left[ {\bf I} + \frac{1}{\epsilon} \, \ul{\yt}_k \ul{\yt}_k^\dag \right] {\bf Q}_{\ul{\Zt}_k}^{1/2} \\
   \left({\bf Q}_{\ul Y}^{(k)}\right)^{-1} & = \frac{1}{\epsilon} {\bf Q}_{\ul{\Zt}_k}^{-1/2} \left[ {\bf I} - \frac{\ul{\yt}_k \ul{\yt}_k^\dag}{\epsilon+\|\ul{\yt}\|^2}\right] {\bf Q}_{\ul{\Zt}_k}^{-1/2}.
\end{align}
We further compute approximations for small $\epsilon$:
\begin{align}
    & \ul y_k^\dag \left({\bf Q}_{\ul Y}^{(k)}\right)^{-1} \ul y_k
    = \frac{\| \ul{\yt}_k \|^2}{\epsilon + \| \ul{\yt}_k \|^2} \approx 1 \\
    & {\bf H}_k = \left(\ul y_k \E{\left. \ul{\bar X}^\dag \right| \set E_k} + \epsilon^{1/2} \E{\left. \ul{\Zt}_k \ul{\bar X}^\dag \right| \set E_k} \right) \left( {\bf Q}_{\ul{\bar X}}^{(k)} \right)^{-1} \nonumber \\
    & \quad \;\; \approx \ul y_k \E{\left. \ul{\bar X}^\dag \right| \set E_k} \left( {\bf Q}_{\ul{\bar X}}^{(k)} \right)^{-1} .
\end{align}

We can now treat the limit of large $K$ for which $\epsilon$ approaches zero, i.e., we choose a different auxiliary model for each $\ul Y = \ul y$. Applying the Woodbury and Sylvester identities \eqref{eq:Woodbury}--\eqref{eq:Sylvester} several times, \eqref{eq:AGMI-3-part} becomes
\begin{align}
    & I_1(A ; \ul Y) = \int_{\mathbb C^N} p(\ul y) \nonumber \\
    & \left[
    \log \det \left( {\bf I}  + \left( {\bf Q}_{\ul{\bar X}}^{(\ul y)} - \ul E_{\ul y} \, \ul E_{\ul y}^\dag \right)^{-1} {\bf Q}_{\ul{\bar X}}
    \left( {\bf Q}_{\ul{\bar X}}^{(\ul y)} \right)^{-1} \ul E_{\ul y} \, \ul E_{\ul y}^\dag \right) \right. \nonumber \\
    & \left. - \tr{ \left( {\bf Q}_{\ul{\bar X}}^{(\ul y)} \left( {\bf D}_{\ul{\bar X}}^{(\ul y)}\right)^{-1} {\bf Q}_{\ul{\bar X}}^{(\ul y)} -  \ul E_{\ul y} \, \ul E_{\ul y}^\dag \right)^{-1} \ul E_{\ul y} \, \ul E_{\ul y}^\dag  } \right] d\ul y
   \label{eq:GMI-5-part-2}
\end{align}
where
\begin{align*}
    \ul E_{\ul y} & = \E{\ul{\bar X} | \ul Y = \ul y} \\ 
    {\bf Q}_{\ul{\bar X}}^{(\ul y)} & = \E{\left. \ul{\bar X}\,\ul{\bar X}^\dag \right| \ul Y = \ul y} \\
    {\bf D}_{\ul{\bar X}}^{(\ul y)} & = {\bf Q}_{\ul{\bar X}} - {\bf Q}_{\ul{\bar X}}^{(\ul y)} .
\end{align*}
If $\ul{\bar X},\ul Y$ are jointly CSCG, then using \eqref{eq:LMMSE-est}--\eqref{eq:LMMSE-indep} we have
\begin{align}
  & \ul E_{\ul y} = \E{\ul{\bar X} \, \ul Y^\dag} {\bf Q}_{\ul Y}^{-1} \cdot \ul y  \label{eq:GMI-5-part-Es1-2} \\
   & {\bf Q}_{\ul{\bar X}}^{(\ul y)} - \ul E_{\ul y} \, \ul E_{\ul y}^\dag = {\bf Q}_{\ul{\bar X}} -  \E{\ul{\bar X} \, \ul Y^\dag} {\bf Q}_{\ul Y}^{-1} \E{\ul{\bar X} \, \ul Y^\dag}^\dag \label{eq:GMI-5-part-Es2-2} .
\end{align}
For example, if $\ul Y = {\bf H} \ul X + \ul Z$ where ${\bf H},A, \ul Z$ are mutually independent and $\E{\ul Z}={\bf 0}$, then we have (cf.~\eqref{eq:GMI-5-part-Es1a})
\begin{align}
  \ul E_{\ul y}
  & =  \int_{\mathbb C^{N \times M}} p({\bf h} | \ul y) \, \E{\ul{\bar X} | \ul Y = \ul y, {\bf H} = {\bf h} }  \, d {\bf h}  \nonumber \\
  & = \int_{\mathbb C^{N \times M}} p({\bf h} | \ul y) \, {\bf Q}_{\ul{\bar X}} \, {\bf h}^\dag 
  \left( {\bf I} + {\bf h} {\bf Q}_{\ul X} {\bf h}^\dag \right)^{-1} \ul y \, d{\bf h} \label{eq:GMI-5-part-Es1a-2-step} \\
  & = \E{ \left. {\bf Q}_{\ul{\bar X}} {\bf H}^\dag
  \left( {\bf I} + {\bf H} {\bf Q}_{\ul X} {\bf H}^\dag \right)^{-1} \right| \ul Y = \ul y} \cdot \ul y
   \label{eq:GMI-5-part-Es1a-2} 
\end{align}
where we have applied \eqref{eq:GMI-5-part-Es1-2} with conditioning on the event ${\bf H}={\bf h}$. Similarly, we apply  a conditional version of~\eqref{eq:GMI-5-part-Es2-2} and the step~\eqref{eq:GMI-5-part-Es1a-2-step} to compute  (cf.~\eqref{eq:GMI-5-part-Es2a})
\begin{align}
  {\bf Q}_{\ul{\bar X}}^{(\ul y)}
  & = \int_{\mathbb C^{N \times M}} p({\bf h} | \ul y) \, \E{\left. \ul{\bar X}\,\ul{\bar{X}}^\dag \right| \ul Y = \ul y, {\bf H}={\bf h}} \, d{\bf h} \nonumber \\
  & = \int_{\mathbb C^{N \times M}} p({\bf h} | \ul y) \left( {\bf Q}_{\ul{\bar X}}^{(\ul y, {\bf h})} + \ul E_{\ul y,{\bf h}}\, \ul E_{\ul y,{\bf h}}^\dag \right) \, d{\bf h} \nonumber \\
  & = \E{ \left. {\bf Q}_{\ul{\bar X}}^{(\ul y, {\bf H})} + \ul E_{\ul y,{\bf H}}\, \ul E_{\ul y,{\bf H}}^\dag \right| \ul Y = \ul y}
\end{align}
where
\begin{align*}
    {\bf Q}_{\ul{\bar X}}^{(\ul y, {\bf h})} 
    & = {\bf Q}_{\ul{\bar X}} - {\bf Q}_{\ul{\bar X}} {\bf h}^\dag
     \left( {\bf I} + {\bf h} {\bf Q}_{\ul X} {\bf h}^\dag \right)^{-1} {\bf h} {\bf Q}_{\ul{\bar X}} 
   \\
    \ul E_{\ul y,{\bf h}} 
    & = \E{\ul{\bar X} | \ul Y = \ul y, {\bf H}={\bf h}} \nonumber \\
    & = {\bf Q}_{\ul{\bar X}} \, {\bf h}^\dag 
  \left( {\bf I} + {\bf h} {\bf Q}_{\ul X} {\bf h}^\dag \right)^{-1} \ul y .
\end{align*}

\appendix{Proof of Lemma~\ref{lemma:AGMI-max-MIMO}}
\label{appendix:proof-lemma-AGMI-max-vec}

We mimic the steps of Appendix~\ref{appendix:proof-lemma-AGMI-max}. Consider the SVDs
\begin{align*}
   {\bf Q}_{\ul{\bar X}} & = {\bf V}_{\ul{\bar X}} \, {\bf \Sigma}_{\ul{\bar X}} \, {\bf V}_{\ul{\bar X}}^\dag \\
   {\bf Q}_{\ul X(s_T)} & = {\bf V}_{\ul X(s_T)} \, {\bf \Sigma}_{\ul X(s_T)} \, {\bf V}_{\ul X(s_T)}^\dag.
\end{align*}
Let $\ul{\bar U}\sim\mathcal{CN}(0,I)$ and write
\begin{align*}
   \ul{\bar X} & =  {\bf Q}_{\ul{\bar X}}^{1/2} \, \ul{\bar U} .
\end{align*}
Since the $\ul U(s_T)$ are CSCG, we have
\begin{align}
  \ul U(s_T') = {\bf R}(s_T',s_T) \, \ul U(s_T) + \ul Z(s_T')
\end{align}
where ${\bf R}(s_T',s_T)=\E{\ul U(s_T') \, \ul U(s_T)^\dag}$ and
\begin{align}
  \ul Z(s_T')\sim\mathcal{CN}(0,{\bf I}-{\bf R}(s_T',s_T) {\bf R}(s_T',s_T)^\dag)
\end{align}
is independent of $\ul U(s_T)$. As in \eqref{eq:xbar}, define
\begin{align}
  & \ul{\bar X} = \sum_{s_T'} {\bf W}(s_T' ) \, \ul X(s_T') \nonumber \\
  & = \sum_{s_T'} {\bf W}(s_T') \, {\bf Q}_{\ul X(s_T')}^{1/2}
   \left[ {\bf R}(s_T',s_T) \, \ul U(s_T) + \ul Z(s_T') \right] \nonumber \\
  & = {\bf Q}_{\ul{\bar X}}^{1/2} \, \bar {\bf R}(s_T) \, \ul U(s_T)
  + \sum_{s_T'} {\bf W}(s_T') \, {\bf Q}_{\ul X(s_T')}^{1/2} \, \ul Z(s_T') 
  \label{eq:xbar-vec}
\end{align}
where as in \eqref{eq:rhobar}, and assuming ${\bf Q}_{\ul{\bar X}} \succ {\bf 0}$, we write
\begin{align}
   \bar {\bf R}(s_T) & = \E{\ul{\bar U} \, \ulU(s_T)^\dag} \nonumber \\
   & = \sum_{s_T'} {\bf Q}_{\ul{\bar X}}^{-1/2} \, {\bf W}(s_T') \, {\bf Q}_{\ul X(s_T')}^{1/2} \, {\bf R}(s_T',s_T) .
\end{align}
Observe that the vector ${\bf Q}_{\ul{\bar X}}^{1/2} \, \bar {\bf R}(s_T) \, \ul U(s_T)$ is the LMMSE estimate of $\ul{\bar X}$ given $\ul U(s_T)$.

Using Lemma~\ref{lemma:AGMI}, we have (see \eqref{eq:AGMI-parameters-RkF})
\begin{align}
   \tilde {\bf H} = \E{\ul Y \, \ul {\bar X}^\dag} {\bf Q}_{\ul{\bar X}}^{-1}, \quad
   {\bf Q}_{\tilde{\ul Z}} = {\bf Q}_{\ul Y} - \tilde {\bf H}\, {\bf Q}_{\ul{\bar X}} \tilde {\bf H}^\dag
   \label{eq:AGMI-parameters-RkF-vec}
\end{align}
and we have the GMI \eqref{eq:AGMI-3-I1} that we repeat here:
\begin{align}
   I_1(A;\ul Y) = \log \left( \frac{\det {\bf Q}_{\ul Y}}
   {\det\left( {\bf Q}_{\ul Y} - \tilde {\bf H}\, {\bf Q}_{\ul{\bar X}} \tilde {\bf H}^\dag\right)} \right).
   \label{eq:appendix-AGMI-3}
\end{align}
As in Appendix~\ref{appendix:proof-lemma-AGMI-max}, if the $ {\bf Q}_{\ul X(s_T)}$ are fixed, then so is ${\bf Q}_{\ul Y}$ because $\ul U(s_T) \sim \mathcal{CN}(\ul 0,{\bf I})$ is independent of $\ul Z$ given $S_T=s_T$. We want to maximize the GMI \eqref{eq:appendix-AGMI-3}. Similar to~\eqref{eq:Lemma-AGMI-Ptilde1}, we have the decomposition
\begin{align}
   \tilde {\bf H} {\bf Q}_{\ul{\bar X}} \tilde {\bf H}^\dag
   = \tilde {\bf D}\, \tilde {\bf D}^\dag
   \label{eq:Lemma-AGMI-HQH}
\end{align}
where
\begin{align}
  \tilde {\bf D} = \sum_{s_T} P_{S_T}(s_T)
   \E{\left. \ul Y \, \ul U(s_T)^\dag \right| S_T=s_T} \bar {\bf R}(s_T)^\dag .
   \label{eq:Lemma-AGMI-Dtilde}
\end{align}
As in \eqref{eq:Lemma-AGMI-Ptilde1}, we have the Markov chain $A-[\ul U(S_T),S_T]-\ul Y$ which implies that $\ul Y$ and the $\ul Z(s_T')$ in \eqref{eq:xbar-vec} are independent give $S_T=s_T$. It is natural to expect that the matrix  $\bar {\bf R}(s_T)$ of correlation coefficients should be ``maximized'' somehow. Indeed, the Cauchy-Schwarz inequality gives
\begin{align*}
   & \ul{v}_1^\dag \, \bar {\bf R}(s_T) \, \ul{v}_2
   = \E{\ul{v}_1^\dag \, \ul{\bar U} \cdot \ulU(s_T)^\dag \, \ul{v}_2} \nonumber \\
   & \le \sqrt{\E{\left| \ul{\bar U}^\dag \, \ul{v}_1 \right|^2}}
      \cdot \sqrt{\E{\left| \ul U(s_T)^\dag \, \ul{v}_2 \right|^2}} = \left\| \ul{v}_1 \right\| \cdot \left\| \ul{v}_2 \right\|
\end{align*}
for any complex $M$-dimensional vectors $\ul v_1$ and $\ul v_2$. The singular values of ${\bf R}(s_T)$ are thus at most 1. We will choose the $\ul U(s_T)$ so that the ${\bf R}(s_T) $ are unitary matrices, and thus all singular values are 1.

Consider the SVD decompositions \eqref{eq:CSIT-SVD-decomposition} and a codebook based on scaling and rotating a common $\ul U\sim\mathcal{CN}(\ul 0,{\bf I})$ of dimension $N$ (see \eqref{eq:barX-max}):
\begin{align}
   \ul U(s_T) = {\bf V}_T(s_T) \, \ul U .
   \label{eq:UsT-max-vec}
\end{align}
The receiver chooses $M\times M$ unitary matrices ${\bf V}_R(s_T)$ for all $s_T$ and uses the weighting matrix (cf.~\eqref{eq:w-sT})
\begin{align}
   {\bf W}(\tilde s_T) = {\bf Q}_{\ul{\bar X}}^{1/2} \, {\bf V}_R(\tilde s_T)
   \, {\bf V}_T(\tilde s_T)^\dag
   \, {\bf Q}_{\ul X(\tilde s_T)}^{-1/2}
\end{align}
for one $\tilde s_T \in \set S_T$ with $ {\bf Q}_{\ul X(\tilde s_T)} \succ 0$, and ${\bf W}(\tilde s_T)={\bf 0}$ otherwise. These choices give $\ul{\bar X} = {\bf Q}_{\ul{\bar X}}^{1/2} \, \ul U$ and (cf.~\eqref{eq:rho_sT})
\begin{align}
   {\bf R}(s_T',s_T) & = {\bf V}_T(s_T') \, {\bf V}_T(s_T)^\dag \nonumber \\
   \bar {\bf R}(s_T) & = {\bf V}_R(s_T) {\bf V}_T(s_T)^\dag.
   \label{eq:Rbarmatrix}
\end{align}
Using \eqref{eq:CSIT-SVD-decomposition}, \eqref{eq:Lemma-AGMI-Dtilde}, and \eqref{eq:Rbarmatrix}, we have
\begin{align}
  \tilde {\bf D} = 
  \sum_{s_T} P_{S_T}(s_T) \, {\bf U}_T(s_T) \, {\bf \Sigma}(s_T) \, {\bf V}_R(s_T)^\dag.
   \label{eq:Lemma-AGMI-Dtilde2}
\end{align}

\bibliographystyle{IEEEtran}


\end{document}